\documentclass{article}

\usepackage[english]{babel}

\usepackage[a4paper,top=2cm,bottom=2cm,left=3cm,right=3cm,marginparwidth=1.75cm]{geometry}

\usepackage{amssymb,amsmath,color,graphicx,dsfont,mathrsfs,amsthm,stmaryrd,mathtools}
\usepackage{authblk}
\usepackage{xcolor}
\usepackage{enumitem}
\usepackage[unicode,colorlinks=true,linkcolor=red,citecolor=blue,linktocpage=true]{hyperref}
\usepackage{braket}
\usepackage[capitalise]{cleveref}
\let\C\relax 
\usepackage[normalem]{ulem}

\usepackage{tikz,tikz-cd,tikzfig}
\usepackage[utf8]{inputenc}
\usepackage{csquotes}

\usepackage{centernot}
\usepackage{mathtools}
\usepackage{stmaryrd}

\usepackage{thmtools}
\usepackage{thm-restate}

\newcommand\rememberEqLabel[2]
  {\expandafter\gdef\csname labeled:#1\endcsname{#2}%
   \label{#1}#2}
\newcommand\recallEq[1]
   {\csname labeled:#1\endcsname\tag{\ref{#1}}}

\newtheorem{theorem}{Theorem}
\newtheorem{lemma}[theorem]{Lemma}
\newtheorem{remark}[theorem]{Remark}
\newtheorem{definition}[theorem]{Definition}
\newtheorem{proposition}[theorem]{Proposition}

\newtheorem{example}[theorem]{Example}

\usepackage[
   backend=biber,
   maxbibnames=6,
   maxcitenames=6,
   style=alphabetic,
   sorting=nyt,
   isbn=false,
   doi=true,
   url=false,
   giveninits=true,
   date=short 
]{biblatex}
\newcommand{\cA}{\mathcal{A}}
\newcommand{\cB}{\mathcal{B}}
\newcommand{\N}{\mathbb{N}}
\newcommand{\cN}{\mathcal{N}}
\newcommand{\Z}{\mathbb{Z}}
\newcommand{\cZ}{\mathcal{Z}}
\newcommand{\C}{\mathbb{C}}

\newcommand{\R}{\mathbb{R}}
\newcommand{\Q}{\mathbb{Q}}

\newcommand{\cS}{\mathcal{S}}
\renewcommand{\H}{\mathcal{H}}

\newcommand{\cD}{\mathcal{D}}
\newcommand{\K}{\mathbb{K}}
\newcommand{\cK}{\mathcal{K}}
\newcommand{\T}{\mathbb{T}}
\newcommand{\cT}{\mathcal{T}}
\renewcommand{\cS}{\mathcal{S}}
\newcommand{\cL}{\mathcal{L}}
\newcommand{\cF}{\mathcal F}
\newcommand{\I}{\mathcal I}
\newcommand{\J}{\mathcal J}
\newcommand{\cG}{\mathcal{G}}

\newcommand{\F}{\mathbb{F}}

\newcommand{\Span}[1]{\operatorname{Span}\left\{ #1 \right\} }
\renewcommand{\ker}[1]{\operatorname{ker}\left( #1 \right)}
\newcommand{\im}[1]{\operatorname{im}\left( #1 \right)}
\newcommand{\supp}[1]{\operatorname{supp}\left( #1 \right)}
\newcommand{\diag}[2]{\operatorname{diag}_{#1}\left( #2 \right)}
\newcommand{\kerpi}[1]{\operatorname{ker}_{\T}\left( #1 \right)}
\newcommand{\impi}[1]{\operatorname{im}_{\T}\left( #1 \right)}

\newcommand{\indicator}[1]{\boldsymbol{1}_{ #1 }}

\renewcommand{\a}{\boldsymbol{\hat{a}}}
\newcommand{\ha}{\hat{a}}
\newcommand{\n}{\boldsymbol{\hat{n}}}
\newcommand{\hn}{\hat{n}}
\newcommand{\balpha}{\boldsymbol{\alpha}}
\newcommand{\bbeta}{\boldsymbol{\beta}}
\newcommand{\bDelta}{\boldsymbol{\Delta}}
\newcommand{\bn}{\boldsymbol{n}}
\newcommand{\bx}{\boldsymbol{x}}
\newcommand{\by}{\boldsymbol{y}}
\newcommand{\bz}{\boldsymbol{z}}
\newcommand{\bw}{\boldsymbol{w}}
\newcommand{\bphi}{\boldsymbol{\phi}}
\newcommand{\bh}{\boldsymbol{h}}
\newcommand{\bg}{\boldsymbol{g}}
\newcommand{\bb}{\boldsymbol{b}}
\newcommand{\bp}{\boldsymbol{p}}
\newcommand{\bu}{\boldsymbol{u}}
\newcommand{\bv}{\boldsymbol{v}}
\newcommand{\bi}{\boldsymbol{i}}
\newcommand{\bj}{\boldsymbol{j}}
\newcommand{\bl}{\boldsymbol{l}}
\newcommand{\bk}{\boldsymbol{k}}
\newcommand{\ba}{\boldsymbol{a}}
\newcommand{\be}{\boldsymbol{e}}
\newcommand{\br}{\boldsymbol{r}}
\newcommand{\bgamma}{\boldsymbol{\gamma}}
\newcommand{\bpsi}{\boldsymbol{\psi}}

\binoppenalty=\maxdimen
\relpenalty=\maxdimen

\title{
	Algebraic Structure of Tiger Codes
}
\author[1,2]{Clément Poirson\thanks{clement.poirson@alice-bob.com}}
\author[2]{Anthony Leverrier\thanks{anthony.leverrier@inria.fr}}
\author[1]{Christophe Vuillot\thanks{christophe.vuillot@alice-bob.com}}
\affil[1]{Alice \& Bob}
\affil[2]{Inria Paris}

\date{\vspace{-1em}{\small \today}\vspace{-1.5em}}

\begin{document}
\maketitle

\begin{abstract}
   Tiger codes form a family of multimode bosonic quantum codes that unify 
   several previously known constructions, including cat, paircat, and 
   the two-mode binomial code. In this work, we give a rigorous algebraic treatment 
   of these codes. Starting from a kernel definition of the codespace, 
   we prove that the annihilation-type constraints admit a finite generating 
   set, construct an explicit orthonormal basis, and show that the logical 
   structure of the code is governed by the homology of an underlying chain 
   complex, as expected in \cite{xu2024lettingtigercagebosonic}. 
   We then develop a Fourier transform over the codespace 
   to prove that the span of phase-rotated projected coherent states 
   is dense therein, 
   and to yield dual \(X\)- and \(Z\)-type descriptions of the code. 
   We further extend the framework to non-linear number constraints, encompassing 
   codes such as the four-legged cat or the repetition cat code. 
   Finally, we investigate the implementation of logical operations. 
   We first generalise the construction of logical Pauli operators proposed 
   in \cite{xu2024lettingtigercagebosonic} to arbitrary logical spaces, 
   and then construct non-Clifford gates using physical polynomial phase 
   rotations of the form $e^{iP(\n)}$. 
   We derive criteria on the real polynomial $P$ which, 
   for positive single-logical-qubit Tiger codes satisfying an additional 
   sign assumption, such as the paircat code, 
   characterise the polynomials $P$ that preserve the codespace by decomposing them 
   into a family of univariate polynomials. Through this decomposition, 
   we relate the degrees of the resulting components to the induced 
   logical action in the Clifford hierarchy.
   These results establish Tiger codes as a mathematically robust framework 
   for describing a broad class of bosonic encodings.
\end{abstract}

\tableofcontents

\clearpage

\section{Introduction}

Although quantum computing promises computational speed-ups for a 
variety of problems, this advantage over classical computers remains 
largely hindered by noise. Due to imperfections in the current 
available hardware, errors accumulate 
rapidly, preventing any practical large-scale computation.

A standard approach to mitigate this issue is to encode quantum 
information into a \emph{quantum error correcting code}, such as \emph{stabiliser codes} defined over qubits 
\cite{gottesman1997stabilizercodesquantumerror}.
However, this qubit-centric approach implicitly assumes that 
the fundamental building blocks of the hardware are two-level 
systems. In practice, this is rarely the case.

\paragraph{Bosonic encodings.}
In many physical implementations—most notably in superconducting 
architectures—the underlying degrees of freedom are bosonic modes 
rather than genuine two-level systems. A common strategy is 
therefore to artificially restrict such a system to its two 
lowest energy levels, thereby defining an effective qubit. 
\emph{Transmon qubits} are the canonical example of this approach, 
and can be viewed as a simple form of \emph{bosonic encoding}. 
Their main advantage lies in their relative ease of experimental 
implementation \cite{Schreier2008,Place2021TransmonTantalum}.

Beyond this simple construction, a variety of more sophisticated 
bosonic encodings have been proposed. A non-exhaustive list includes 
\emph{GKP codes} \cite{GottesmanKitaevPreskill_2001}, \emph{cat qubits} \cite{Mirrahimi_2014}, 
\emph{two-mode binomial codes} \cite{Michael_2016}, \emph{rotation-invariant codes} \cite{Grimsmo_2020},
\emph{spherical codes} \cite{Jain_2024}, and more recently \emph{bosonic cyclic codes} \cite{wetherbee2026bosoniccycliccodestrading}. 
While typically more challenging to realise experimentally, these 
encodings can offer significantly improved protection against 
specific noise channels.

In particular, cat qubits have recently led to several promising 
architectural proposals, especially when concatenated with 
stabiliser codes \cite{Guillaud_2019,Chamberland_2022,Ruiz_2025}. 
These works 
highlight a key paradigm: rather than treating bosonic modes 
as imperfect qubits, one can exploit their full structure to 
build more robust logical encodings at the hardware level.

\paragraph{Tiger codes: a unifying bosonic framework.}

Driven by the success of cat qubits, several new bosonic code
constructions have been introduced, such as the
\emph{4-legged cat code} \cite{Mirrahimi_2014} and the \emph{paircat code}
\cite{Albert_2019}. Recently, a unifying framework called
\emph{Tiger codes} was proposed in
\cite{xu2024lettingtigercagebosonic}. It encompasses a broad class of
previously known bosonic codes, including the cat code, the paircat code, 
the two-mode binomial code among others.

The dominant sources of noise in bosonic systems are dephasing and photon loss, 
which are commonly modeled as quantum channels \cite{Chuang_1997,Albert_2019,Arqand_2020}.
Tiger codes are \emph{multimode} bosonic codes designed in a
Calderbank-Shor-Steane-type fashion to protect against both these errors. More precisely,
let $\balpha\in\C^m$ be a complex vector with no zero entry, and
$\bDelta\in\Z^s$ an integer vector. Given two integer matrices
$H\in\Z^{s\times m}$ and $G\in\Z^{m\times r}$ such that $HG=0$\footnote{We do not follow the usual CSS-code convention $HG^\top=0$. Here, $G$ and $H$ are morphisms of $\Z$-modules, and we reserve the transpose notation for the induced maps on the Pontryagin dual groups, identified with morphisms $G^\top$ and $H^\top$ between subgroups of $[0,2\pi)^m$. We therefore write $HG=0$. This subtlety does not appear in the binary setting, since $\Z_2$ is self-dual.},
the rows
$\{\bh_i\}_{1\leq i\leq s}$ of $H$ define the $\n$-type constraints,
which protect against photon loss, while the columns of $G$ define the
$\a$-type constraints, which protect against dephasing. By definition 
of a Tiger code~\cite{xu2024lettingtigercagebosonic}, any $m$-mode
code state $\ket{\psi}$ should satisfy
\begin{equation}\label{eq:constraints}
   \forall 1\leq i\leq s,\quad (\bh_i\cdot\n-\Delta_i)\ket{\psi}=0
   \qquad\text{and}\qquad
   \forall \bg\in\im{G},\quad (\a^{\bg^+}-\balpha^{\bg}\a^{\bg^-})\ket{\psi}=0,
\end{equation}
where all powers are understood componentwise, and for
$\bg\in\Z^m$, we write $\bg=\bg^+-\bg^-$ with
$\bg^+:=\max(\bg,0)\in\N^m$ and $\bg^-:=\max(-\bg,0)\in\N^m$, 
$\balpha^{\bg}:=\prod_{j=1}^m \alpha_j^{g_j}$, 
$\a^{\bg^+}:=\prod_{j=1}^m \hat{a}_j^{g_j^+}$ and 
$\bh_i\cdot\n:=\sum_{j=1}^m h_{i,j}\hn_j$ with $\hat{a}_i$ and
$\hn_i$ the annihilation and number operators of the $i$-th mode.

The simplest non-trivial example of a bosonic code defined by both
$\a$-type and $\n$-type constraints is the paircat code
\cite{Albert_2019}. It consists of the two-mode states $\ket{\psi}$
satisfying
\begin{equation}
   (\hn_1-\hn_2-\Delta)\ket{\psi}=0
   \qquad\text{and}\qquad
   (\hat{a}_1^2\hat{a}_2^2-\alpha^4)\ket{\psi}=0,
\end{equation}
where $\Delta\in\Z$, $\alpha\in\C$ with $\alpha\neq0$, $H=\begin{pmatrix}
   1&-1
\end{pmatrix}$ and $G=\begin{pmatrix}
   2\\2
\end{pmatrix}$\footnote{Notice that, for the paircat code, only a single 
   $\a$-type constraint suffices to define all the 
   $\a$-type constraints since for any $p\in\N$,
   $\ker{\ha_1^{2p}\ha_2^{2p}-\alpha^{4p}}\subset \ker{\hat{a}_1^2\hat{a}_2^2-\alpha^4}$
   and $\ker{1-\alpha^{-4p}\ha_1^{2p}\ha_2^{2p}}\subset \ker{\hat{a}_1^2\hat{a}_2^2-\alpha^4}$.}. 

States satisfying the $\a$-type constraints can be constructed
from the coherent state $\ket{\balpha}$. Indeed, for any
$\bphi\in[0,2\pi)^m$ such that $G^\top\bphi=0 \text{ mod }2\pi$, one has
$\bphi\cdot \bg^+=\bphi\cdot\bg^- \text{ mod }2\pi$ for all
$\bg\in\im{G}$. Therefore, using the commutation relation 
$\a^{\bg^\pm}e^{i\bphi\cdot\n}=e^{i\bphi\cdot(\n+\bg^\pm)}\a^{\bg^\pm}$, we get
\begin{equation}
   (\a^{\bg^+}-\balpha^{\bg}\a^{\bg^-})e^{i\bphi\cdot\n}
   =e^{i\bphi\cdot(\n+\bg^+)}(\a^{\bg^+}-\balpha^{\bg}\a^{\bg^-}).
\end{equation}
Since
\begin{equation}
   (\a^{\bg^+}-\balpha^{\bg}\a^{\bg^-})\ket{\balpha}
   =(\balpha^{\bg^+}-\balpha^{\bg}\balpha^{\bg^-})\ket{\balpha}=0,
\end{equation}
all the states $e^{i\bphi\cdot\n}\ket{\balpha}$ satisfy the
$\a$-type constraints. To impose the $\n$-type constraints as well, one
introduces the orthogonal projector
\begin{equation}
   \Pi_H:=\sum_{\substack{\bn\in\N^m\\ H\bn=\bDelta}}\ket{\bn}\bra{\bn}
\end{equation}
onto $\bigcap_{i=1}^s\ker{\bh_i\cdot\n-\Delta_i}$, and observes that it
commutes with $e^{i\bphi\cdot\n}$ for every $\bphi\in[0,2\pi)^m$. Hence, defining
$\ket{\balpha_H}:=\Pi_H\ket{\balpha}$, the states
\begin{equation}
   \left\{ 
   e^{i\bphi\cdot\n}\ket{\balpha_H}
   :\bphi\in[0,2\pi)^m,\ G^\top\bphi=0 \text{ mod }2\pi
   \right\}
\end{equation}
satisfy both the $\a$-type and $\n$-type constraints.
For the paircat code, $\bphi_1=\begin{pmatrix}
   0&0
\end{pmatrix}^\top$ and \break $\bphi_2=\begin{pmatrix}
   \pi & 0
\end{pmatrix}^\top$ both satisfy $G^\top\bphi=0 \text{ mod }2\pi$, hence,
as $e^{i\bphi_1\cdot\n}=I$ and $e^{i\bphi_2\cdot\n}=e^{i\pi\hn_1}$,
the two states 
$\Pi_H\ket{\alpha,\alpha}$ and $\Pi_H\ket{-\alpha,\alpha}$
both belong to the codespace, with \break
$\Pi_H=\sum_{\substack{n_1,n_2\in\N\\ n_1-n_2=\Delta}} \ket{n_1,n_2}\bra{n_1,n_2}$.

This motivates the original definition of Tiger codes proposed in 
\cite{xu2024lettingtigercagebosonic}:
\begin{equation}\label{eq:original_def}
   t(G,H,\balpha,\bDelta)
   :=
   \Span{
      e^{i\bphi\cdot\n}\ket{\balpha_H}
      :
      \bphi\in[0,2\pi)^m,\ G^\top\bphi=0 \text{ mod }2\pi
   }.
\end{equation}
These codes have two important features. First, they provide a unified
framework to encode various logical systems---qubits, qudits, rotors,
or even bosonic modes---in a multimode bosonic setting. The second feature, detailed below, is that they
are natural candidates for hardware-efficient implementations.

A major strength of the cat code is that it admits dissipative
implementations. In an idealised setting, one engineers a Lindblad
dynamics of the form
\begin{equation}
   \frac{d\rho}{dt}=\cD_{\hat{a}^2-\alpha^2}(\rho),
\end{equation}
where
$\cD_L(\rho):=L\rho L^\dag-\frac{1}{2}(L^\dag L\rho+\rho L^\dag L)$.
Such a dynamics drives the system towards the codespace
$\ker{\hat{a}^2-\alpha^2}$ and therefore provides an \emph{autonomous}
protection against dephasing \break
\cite{Mirrahimi_2014,azouit2016wellposednessconvergencelindbladmaster}. 
More generally, given a finite family of
$\a$-type constraints
$\{\a^{\bg_i^+}-\balpha^{\bg_i}\a^{\bg_i^-}\}_{1\leq i\leq k}$, one may hope
to stabilise the corresponding Tiger code by combining one dissipative
process for each constraint. However, Tiger codes are \emph{a priori} 
defined by an infinite family of $\a$-type constraints in \cref{eq:constraints}.
This naturally raises the question of how 
many constraints are needed. 
Alternatively, the same constraints may also 
be enforced at the Hamiltonian
level, as in the Kerr-cat approach \cite{Grimm_2020}. In that setting, the cat
manifold is realised as the low-energy eigenspace of a time-dependent
Hamiltonian rather than as the steady space of a dissipative dynamics. A
standard choice is
\begin{equation}
   H(t)=-K\ha^{\dag 2}\ha^2 + P(t)(\ha^{\dag 2}+\ha^2),
\end{equation} 
where $K>0$ is the Kerr non-linearity and
$P(t)$ a two-photon drive amplitude. Starting from the vacuum and slowly
deforming the Hamiltonian, one adiabatically follows the instantaneous
ground-state branch from the harmonic-oscillator regime to the cat regime.
At the final time $T$, when $\Delta(T)=0$ and $P(T)\neq 0$, the
Hamiltonian has a low-energy manifold approximately spanned by the coherent
states $\ket{\pm\alpha}$, where $\alpha^2=\frac{P(T)}{K}$. 
Similarly, one may hope to engineer a Hamiltonian
for each $\a$-type constraint that 
slowly drives the vacuum state in each of the $m$-modes towards 
any state of the Tiger code.

Meanwhile, protection against photon loss can be achieved through
repeated measurements of the observables $(\bh_i\cdot\n)_{1\leq i\leq s}$. The outcomes
provide information about the photon-loss events that occurred. This is
similar in spirit to current proposals where cat qubits are
concatenated with a qubit stabiliser code: protection against
dephasing is autonomous, while stabiliser measurements are used to
detect and correct photon-loss errors
\cite{Guillaud_2019,Chamberland_2022,Ruiz_2025}.

\subsection*{Contributions}

The guiding principle of this work is that a state should belong to the
codespace precisely when it satisfies both the $\a$-type and the
$\n$-type constraints. This leads us to propose an alternative definition 
for Tiger as the common kernel
\begin{equation}
   \cT(G,H,\balpha,\bDelta):=
   \bigcap_{\bg\in\im{G}}\ker{\a^{\bg^+}-\balpha^{\bg}\a^{\bg^-}}
   \cap
   \bigcap_{i\in\llbracket1,s\rrbracket}\ker{\bh_i\cdot\n-\Delta_i}
\end{equation}
as opposed to the original definition recalled in \cref{eq:original_def}. 

\paragraph{Finite generation of the constraints.}
The definition of $\cT(G,H,\balpha,\bDelta)$ involves an
\emph{a priori} infinite family of $\a$-type constraints, since
$\im{G}$ is infinite in general. We prove that this family can always be
reduced to a finite one:
\begin{equation}
   \bigcap_{\bg\in\im{G}}\ker{\a^{\bg^+}-\balpha^{\bg}\a^{\bg^-}}
   =
   \bigcap_{\bb\in\{\bb_1,\ldots,\bb_M\}}\ker{\a^{\bb^+}-\balpha^b\a^{\bb^-}}
\end{equation}
for a Markov basis $\{\bb_1,\ldots,\bb_M\}\subset\im{G}$ that we can explicitly compute.

This result is important from the perspective of implementing Tiger
codes: it shows that only finitely many elementary constraints need to
be enforced, whether through dissipative dynamics or through
Hamiltonian terms.

\paragraph{Definition's soundness.}
As discussed above, one always has
\begin{equation}
   t(G,H,\balpha,\bDelta)\subseteq \cT(G,H,\balpha,\bDelta),
\end{equation}
but whether this inclusion is tight was left open in
\cite{xu2024lettingtigercagebosonic}. We prove that the
converse inclusion holds after taking the closure: 
\begin{equation}
   \cT(G,H,\balpha,\bDelta)=\overline{t(G,H,\balpha,\bDelta)}.
\end{equation}
This equality fails in general if $\balpha$ has a zero entry. For
instance, $\ker{\hat a^2-\alpha^2}$ is spanned by the two coherent
states $\ket{\alpha}$ and $\ket{-\alpha}$ when $\alpha\neq0$, whereas
$\ker{\hat a^2}$ is not spanned by coherent states alone when
$\alpha=0$, since $\ket{1}\in\ker{\hat a^2}$ in addition to $\ket{0}$.
Our proof relies on several structural results, which we summarise
below.

\paragraph{Structure of the code space.}
We construct an explicit orthonormal basis of
$\cT(G,H,\balpha,\bDelta)$ by projecting the coherent state
$\ket{\balpha}$ onto suitable lattice-supported subsets of the Fock
basis.
More precisely, define an equivalence relation on $\N^m$ by
$\bn\sim_G \bn'$ if $\bn'-\bn\in\im{G}$, and restrict it to the set
$K_H^{\bDelta}:=\{\bn\in\N^m:H\bn=\bDelta\}$. For each equivalence class $[\bn]$, let 
\begin{equation}
   \Pi_{[\bn]}:=\sum_{\bp\in[\bn]}\ket{\bp} \bra{\bp}\footnote{The operator $\Pi_{[\bn]}$ depends on $G$ even though it does not appear in the notation for conciseness.}
\end{equation} 
and define the corresponding normalised projected coherent state
\begin{equation}\label{eq:intro_pi}
   \ket{\balpha_{[\bn]}}
   :=
   \frac{1}{A_{[\bn]}}\Pi_{[\bn]}\ket{\balpha},
   \qquad
   [\bn]\in K_H^{\bDelta}/\!\!\sim_G.
\end{equation}
where $A_{[\bn]}:=\|\Pi_{[\bn]}\ket{\balpha}\|$.
We prove that the family
$\{\ket{\balpha_{[\bn]}}:[\bn]\in K_H^{\bDelta}/\!\!\sim_G\}$
forms an orthonormal basis of $\cT(G,H,\balpha,\bDelta)$.

The quotient
$K_H^{\bDelta}/\!\!\sim_G$ can be analysed through an injection in the 
\emph{homology} group \break $\ker{H}\!/\!\im{G}$ over $\Z$.
As a finitely generated abelian group, it admits a decomposition of the form
\begin{equation}
   \ker{H}\!/\!\im{G}\cong \Z_{d_1}\times\cdots\times\Z_{d_t}\times\Z^f.
\end{equation}
This decomposition determines the structure of the encoded logical system.
Whenever the injection is a bijection, the code encodes exactly $f$ logical rotors 
and $t$ logical qudits of dimensions $d_1,\ldots,d_t$. Otherwise, the codespace is 
a subspace of the code spanned by $f$ rotors and $t$ qudits of dimensions $d_1,\ldots,d_t$.

As an illustration, we recover the dimension of the paircat code
\cite{Albert_2019}. This code corresponds to
\begin{equation}
   \cT\!\left(
      \begin{pmatrix}
         2\\
         2
      \end{pmatrix},
      \begin{pmatrix}
         1&-1
      \end{pmatrix},
      \begin{pmatrix}
         \alpha\\
         \alpha
      \end{pmatrix},
      \Delta
   \right),
\end{equation}
for which 
\begin{equation}
   \text{ker} \begin{pmatrix}
         1&-1
      \end{pmatrix}
      \Big/
      \text{im}\begin{pmatrix}
         2\\
         2
      \end{pmatrix}
      \cong \Z_2.
\end{equation}
The homology group therefore has two elements, so the codespace is
two-dimensional: the paircat code encodes one logical qubit.

\paragraph{Fourier structure and dual bases.}
We introduce a Fourier transform adapted to \break
$\cT(G,H,\balpha,\bDelta)$.
We prove that each basis vector $\ket{\balpha_{[\bn]}}$ admits an integral
decomposition along the coherent states $e^{-i\bphi\cdot\n}\ket{\balpha_H}$
of the code
\begin{equation}
   \ket{\balpha_{[\bn]}}
   =
   \frac{1}{A_{[\bn]}}
   \int_{[0,2\pi)^m} e^{i\bphi\cdot(\n-\bn)}\ket{\balpha_H}\,d\bphi,
\end{equation}
for an appropriate measure $d\bphi$ over $[0,2\pi)^m$ detailed 
in \cref{sec:Fourier_coherent}.

Based on the above decomposition, we prove the soundness of the definition:
since the vectors $\ket{\balpha_{[\bn]}}$ form an orthonormal basis of the codespace,
this implies that the family
\begin{equation}\label{eq:Zbasis_intro}
   \left\{
      e^{i\bphi\cdot\n}\ket{\balpha_H}
      :
      \bphi\in[0,2\pi)^m,\ G^\top\bphi=0 \text{ mod }2\pi
   \right\}
\end{equation}
is total in $\cT(G,H,\balpha,\bDelta)$, meaning that all code states 
belong to the closure of the span of the above set.

This Fourier viewpoint leads to two complementary descriptions of the
code: an orthonormal ``$X$ basis'' given by \cref{eq:intro_pi}, and a
generally non-orthogonal ``$Z$ basis'' given by the phase-rotated
coherent states in \cref{eq:Zbasis_intro}, in close analogy with the cat
code.

\paragraph{Extension beyond linear constraints.}
A conceptual limitation of Tiger codes is that they only involve linear
$\n$-type constraints. However, several bosonic codes rely on
non-linear observables, such as parity measurements. This is notably
the case for concatenations of cat codes with stabiliser codes
\cite{Guillaud_2019,Chamberland_2022,Jain_2024,Ruiz_2025}.

We therefore introduce an extended class of Tiger codes:
\begin{equation}
   \cT(G,\{h_\gamma\},\balpha,\bDelta):=
   \bigcap_{\bg\in\im{G}}\ker{\a^{\bg^+}-\balpha^{\bg}\a^{\bg^-}}
   \cap
   \bigcap_{\gamma\in\llbracket1,s\rrbracket}\ker{h_\gamma(\n)},
\end{equation}
where $\left(h_\gamma:\Z^m\to\C\right)_{1\leq\gamma\leq s}$ is a finite family of functions whose
kernels are $G$-periodic, that is, $h_\gamma(\bn)=0\Rightarrow h_\gamma(\bn+\bg)=0$ 
for all $\bg\in\im{G}$. 

This extension captures codes based on modular constraints, such as the
single-mode 4-legged cat code with $\n$-type constraint
$(-1)^{\hat n}$, or repetition cat codes, while preserving the main
structural features of Tiger codes: finite generation, a homological
description of the codespace, a Fourier structure, dual bases, and a
description in terms of coherent states.

\paragraph{Logical operators.}
Finally, we study the physical implementation of logical operators. 
This part of the work focuses on Tiger codes with an underlying 
group structure, meaning that the set $K_H^{\bDelta}/\!\!\sim_G$ 
labelling the basis vectors is in bijection with $\ker H/\im G$. 
In particular, this excludes codespaces containing logical modes, 
whose computational theory is different: 
in that setting, Clifford operations become the analogue of 
Gaussian operations.

We begin with logical Pauli operators. For logical $Z$ Pauli 
operators, we show that every such operator can be induced by 
a physical rotation of the form 
\begin{equation}
   e^{i\bphi\cdot\n},
   \qquad
   \bphi\in[0,2\pi)^m,\quad G^\top\bphi=0 \text{ mod }2\pi .
\end{equation}
It has the practical advantage to be implementable
with a simple Hamiltonian dynamics.
We then turn to logical $X$ Pauli operators and investigate 
their physical implementation by operators polynomial in 
$\a$ and $\a^\dag$. Constructing exact representatives of 
this form is, however, an algebraically difficult problem. 
In the absence of a complete characterisation, we follow 
and generalise the construction of 
\cite{xu2024lettingtigercagebosonic} to obtain 
\emph{approximate} logical $X$ Pauli operators. 
The approximation error is controlled by the non-orthogonality 
of the states in \cref{eq:Zbasis_intro}; consequently, 
the approximation improves as the components of $\balpha$ increase.

To go beyond Pauli and Clifford operators, we also study logical 
$Z$ rotations induced by physical operators of the form 
$e^{2i\pi P(\n)}$. For \emph{positive} Tiger codes---that is, 
codes for which there exists $\bg\in\im G$ with $\bg>0$ 
componentwise---with finite-dimensional codespace, and under 
an additional compatibility assumption, we derive a necessary 
condition on the polynomial $P$ so that $e^{2i\pi P(\n)}$
preserves the codespace; see \cref{theorem:Zrotations}. 
In the single-logical-qubit case,  
we can also derive a sufficient condition on $P$
leading to an exact characterisation of the polynomials
implementing a logical $Z$ rotation of angle 
$\frac{(-1)^{d-1}}{2^{d-1}}\pi$; see \cref{prop:qubitZrotation}. 
For $e^{2i\pi P(\n)}$ to preserve the codespace,
$P$ needs to write 
\begin{equation}\label{eq:decomposition_P}
   P(x_1,\ldots,x_m)=
   \sum_{k=1}^{d}f_k(x_1,\ldots,x_m)P_k^{(2)}\left(\frac{\phi_1x_1+\cdots+\phi_mx_m}{2\pi}\right),
\end{equation}
where $\bphi\in[0,2\pi)^m$ satisfies $\bphi\neq0$ and
$G^\top\bphi=0 \text{ mod }2\pi$, the polynomials $\{f_k\}_{1\leq k\leq d}$
take integer values over $\{\bn\in\N^m:H\bn=\bDelta\}$ and 
\begin{equation}
   P_0^{(2)}(X):=1,\qquad
   P_{k+1}^{(2)}(X):=\sum_{i=0}^k\binom{i}{k-i}2^{2i-k}\binom{X/2}{i+1}\in\R[X].
\end{equation}
In particular, the polynomials $\{P_k^{(2)}\}_{k\in\N}$ satisfy
$P_k^{(2)}(n+2)-P_k^{(2)}(n)=0$ for any integer $n\in\Z$.
Hence, over the integers, they take only two values:
\begin{equation}
   P_k^{(2)}(0)=0,
   \qquad
   P_k^{(2)}(1)=\frac{(-1)^{k-1}}{2^{k-1}}\pi .
\end{equation}
It follows that a logical $Z$ rotation of angle $\frac{(-1)^{d-1}}{2^{d-1}}\pi$
can only be produced by the $P_d^{(2)}$ component in the decomposition of $P$ 
in \cref{eq:decomposition_P}. Since such a rotation belongs to the $d$-th 
level of the Clifford hierarchy 
\cite{Cui_2017,desilva2025cliffordhierarchyqubitqudit}, any polynomial 
$P$ for which $e^{2i\pi P(\n)}$ reaches this level must contain the 
degree-$d$ term $P_d^{(2)}$ in its decomposition.

Finally, since any operator of the form $e^{2i\pi P(\n)}$ 
commutes with any functionals in $\n$, all the 
logical $Z$ operators constructed in this work 
also preserve the codespace of the non-linear generalisation of Tiger codes.
However, inferring the logical operation induced is in general more difficult
and \cref{prop:qubitZrotation} has no direct analogous for the non-linear case.

\paragraph{Scope and limitations.}
This paper focuses on the algebraic structure of Tiger codes. We do not
aim to provide a detailed analysis of their error-correcting
performance, nor an extensive catalogue of explicit constructions.

Assessing error-correction performance would in particular require a
more refined study of the associated physical dynamics. We refer to
\cite{xu2024lettingtigercagebosonic} for heuristic distance measures and
for a variety of explicit examples that complement the present
work.

\paragraph{}
The remainder of this paper is organised as follows. In \cref{sec:code_construction}, 
based on 
the kernel definition of Tiger codes, we construct an 
orthonormal basis of the codespace, and prove the finite generation 
of the $\a$-type constraints. In \cref{sec:homology}, we use this basis 
to relate the 
structure of the code to the homology of the underlying 
chain complex. Then, in \cref{sec:Fourier_coherent}, we develop 
a Fourier transform over the codespace that both yields dual $X$- and $Z$-type 
bases and proves the definition's soundness. In \cref{sec:non_linear},
we extend the framework to non-linear $\n$-type 
constraints. Finally, we study logical operators of Tiger codes and 
their non-linear extension
with a focus on logical Pauli operators for codespaces composed of rotors and qudits,
and non-Clifford $Z$-rotations for finite-dimensional codespaces in \cref{sec:logop}.

\subsection*{Notations and conventions}

For clarity, we gather here the conventions used throughout the paper.

\begin{itemize}
   \item The physical Hilbert space of the \(m\) bosonic modes is
   \(\H:=L^2(\R)^{\otimes m}\cong L^2(\R^m)\).
   \item For a (possibly unbounded) operator \(A\), \(\cD(A)\) denotes
   its domain and 
   \begin{equation}
      \ker A:=\{\psi\in\cD(A):A\psi=0\}.
   \end{equation}
   If $(A,\cD(A))$ is closed, the kernel is closed as well and we 
   can define the orthogonal projector onto the kernel.
   Finally, the restriction of $A$ to a subspace $\H_0\subset\cD(A)$ is denoted 
   $A\!\restriction_{\H_0}$.
   \item For \(\bn\in\N^m\), \(\ket{\bn}=\ket{n_1}\otimes\ldots \otimes\ket{n_m}\)
   denotes the multimode Fock state, and
   \begin{equation}
      \ket{\balpha}:=
      e^{-|\balpha|^2/2}\sum_{\bn\in\N^m}\frac{\balpha^{\bn}}{\sqrt{\bn!}}\ket{\bn}
   \end{equation}
   denotes the multimode coherent state.
   \item For a single mode ($\H=L^2(\R)$), the annihilation 
   operator $\ha:=\frac{1}{\sqrt2}(x+\partial_{x})$ 
   is defined over the domain 
   \begin{equation}
      \cD(\ha):=\{\ket{\psi}=\sum_{n\in\N}\psi_n\ket{n}\in\H : 
      \sum_{n\in\N}n|\psi_n|^2<\infty\}.
   \end{equation}
   Its adjoint $\ha^\dag$, the creation operator, is defined on the same 
   domain $\cD(\ha^\dag)=\cD(\ha)$ making both these operators closed.
   Finally, $\hn=\ha^\dag\ha$ with domain the composition 
   of the two domains is also closed.
   \item For a subset $\omega\subset\Omega$, we write the indicator function 
   \begin{equation}
      \indicator{\omega}:x\mapsto\left\{\begin{aligned}
         &1 &\text{if } x\in\omega\\
         &0 &\text{otherwise}.
      \end{aligned}\right.
   \end{equation}
   When the argument $x\in\Omega$ is clear from context, we use the notation $\indicator{\{x\in\omega\}}$
   instead of $\indicator{\omega}(x)$. When $\Omega=\N^m$, using functional calculus, 
   this defines the projector 
   $\indicator{\{\n\in\omega\}}:=\sum_{\bn\in\omega}\ket{\bn}\bra{\bn}$.
   \item We define the unit circle $\T:=\{z\in\C:|z|=1\}$.
   \item All vector inequalities and powers are
   understood componentwise.
   \item For \(\bu\in\N^m\), we write
   \begin{equation}
      \bu!:=u_1!\cdots u_m! \quad \text{and} \quad 
      \a^{\bu}:=\ha_1^{u_1}\cdots \ha_m^{u_m}.
   \end{equation}
   Here \(\ha_i=\frac{1}{\sqrt2}(x_i+\partial_{x_i})\) is the
   annihilation operator of the \(i\)-th mode.
   \item For \(\bv\in\Z^m\),
   \begin{equation}
      \n\cdot\bv:=v_1\hn_1+\cdots+v_m\hn_m,
      \qquad \text{with} \quad
      \hn_i=\ha_i^\dag\ha_i.
   \end{equation}
   \item For \(\bu\in\Z^m\), we write \(\bu=\bu^+-\bu^-\), where
   \(\bu^\pm\in\N^m\), \(u^+_j:=\max(u_j,0)\), \(u^-_j:=\max(-u_j,0)\).
   \item For \(\balpha\in\C^m\) and \(\bu\in\Z^m\), we use the shorthand
   \begin{equation}
      \balpha^{\bu}:=\prod_{j=1}^m \alpha_j^{u_j}.
   \end{equation}
   We also write \(|\balpha|^2:=\sum_{j=1}^m |\alpha_j|^2\).
   \item We write \(\llbracket i,j\rrbracket=\{i,i+1,\ldots,j\}\) and
   \(\langle S\rangle_\Z\) for the additive subgroup generated by a set \(S\).
   \item For any \(\bn_0\in\Z^m\) such that  \(H\bn_0=\bDelta\), we
   use the shorthand
   \begin{equation}
      K_H^{\bDelta}:=\{\bn\in\N^m:H\bn=\bDelta\}=(\bn_0+\ker{H})\cap\N^m
   \end{equation}
   that is independent of the choice of $\bn_0$.
   \item For an integer matrix $A\in\Z^{r\times m}$,
   \begin{equation}
      \kerpi{A}:=\{\bphi\in[0,2\pi)^m:A\bphi=0 \text{ mod }2\pi\}
   \end{equation}
   and,
   \begin{equation}
      \impi{A}:=\{A\bphi \text{ mod } 2\pi:\bphi\in[0,2\pi)^m\}.
   \end{equation}
   \item For a topological space $X$, we denote by $\cB(X)$ its 
   Borel algebra that is the smallest $\sigma$-algebra containing 
   all the closed sets of $X$.
   \item For \(\bi\in\N^m\) and $n,a\in\N$, we use
   \begin{equation}
      \binom{x}{n}=\frac{x(x-1)\ldots(x-n+1)}{n!},
      \qquad
      \binom{x}{n}_a:=a^n\binom{x/a}{n},
      \qquad
      \binom{\bx}{\bi}:=\prod_{j=1}^m \binom{x_j}{i_j}.
   \end{equation}
   \item For a multivariate polynomial $P\in\C[x_1,\ldots,x_t]$
   and $i\in\llbracket1,t\rrbracket$, we denote by $deg_i(P)$ the 
   degree of $P$ in the $i$-th variable, meaning 
   \begin{equation}
      deg_i(P):=deg(x_i\mapsto P(x_1,\ldots,x_t))
   \end{equation}
   for $x_1,\ldots,x_{i-1},x_{i+1},\ldots,x_t$ fixed.
\end{itemize}

\section{Tiger codes construction}\label{sec:code_construction}

Tiger codes are a family of multimode bosonic codes that 
generalises many 
known bosonic encodings (cat, paircat, two-mode binomial\dots). 
They are defined using two kinds of constraints: the $\a$-type and the $\n$-type. 

\begin{definition}\label{def:TC}
   Let $\Z^r\xrightarrow{G}\Z^m\xrightarrow{H}\Z^s$ be a chain complex over 
   $\Z$, meaning that $HG=0$, $\bDelta\in \Z^s$ and $\balpha\in\C^m$ such that $\balpha$ has no 
   zero entry.
   Denote $\bh_1,\ldots,\bh_s$ the rows of $H$. 
   The Tiger code $\cT(G,H, \balpha, \bDelta)$
   is a subspace defined as
   \begin{equation}
      \bigcap_{\bg\in\im{G}}\ker{\a^{\bg^+} - \balpha^{\bg}\a^{\bg^-}} \cap 
      \bigcap_{i\in\llbracket1,s\rrbracket}\ker{\bh_i\cdot\n - \Delta_i}.
   \end{equation}
\end{definition}
Our first focus is to find an orthonormal basis for every Tiger code. 
Then, although \cref{def:TC} suggests that Tiger codes consist in an 
infinite intersection of kernels, we constructively prove that any Tiger 
code can be generated by a finite intersection of kernels.

\subsection{Orthonormal basis}\label{sec:orth_basis}

Fix $G\in\Z^{m\times r}$, $H\in\Z^{s\times m}$, $\bDelta\in\Z^s$ and $\balpha\in\C^m$ 
with all components non-zero. 
We construct an orthonormal basis of $\cT(G,H, \balpha, \bDelta)$
in two steps. First, we provide an orthonormal basis of 
$\bigcap_{\bg\in\im{G}}\ker{\a^{\bg^+} - \balpha^{\bg}\a^{\bg^-}}$.
Then we show how to turn it into a basis of the code.

\paragraph{No $\n$-type constraints.}
Given $\bg\in\Z^m$, we start by characterising vectors $\ket{\psi}$ 
belonging to $\ker{\a^{\bg^+} - \balpha^{\bg}\a^{\bg^-}}$ in terms of their decomposition 
in the Fock basis denoted $\{\ket{\bn}\}_{\bn \in \N^m}$. 

\begin{proposition} \label{prop:step}
   Let $\ket{\psi} \in \H$, and write 
   $\ket{\psi}=\sum_{\bn \in \N^m} \psi_{\bn} \ket{\bn}$ its decomposition in the Fock basis.
   Then, 
   \begin{equation}
      \ket{\psi} \in \ker{\a^{\bg^+} - \balpha^{\bg}\a^{\bg^-}}
      \Leftrightarrow
      \left\{\begin{aligned}
         &\ket{\psi} \in \cD(\a^{\bg^+} - \balpha^{\bg}\a^{\bg^-})
         \quad \text{and}\quad \\
         &\forall \bn \in \N^m, \sqrt{(\bn+\bg^+)!}
         \psi_{\bn+\bg^+}=\balpha^{\bg}\sqrt{(\bn+\bg^-)!}\psi_{\bn+\bg^-}.
      \end{aligned}\right.
   \end{equation}
\end{proposition}

\begin{proof}
   By definition, $\ket{\psi} \in \ker{\a^{\bg^+} - \balpha^{\bg}\a^{\bg^-}}$
   if and only if
   \begin{equation}
      (\a^{\bg^+} - \balpha^{\bg}\a^{\bg^-})\ket{\psi} = 
      \sum_{\bn \in \N^m} 
      \left(\sqrt{\frac{(\bn+\bg^+)!}{\bn!}}\psi_{\bn+\bg^+}-\balpha^{\bg}\sqrt{\frac{(\bn+\bg^-)!}{\bn!}}
      \psi_{\bn+\bg^-}\right)\ket{\bn}=0.
   \end{equation}
   the result follows from the fact that the Fock basis is orthonormal.
\end{proof} 

To generalise this proposition to an intersection of such kernels, we need to introduce the following equivalence relation over $\N^m$.

\begin{definition}
   For an integer matrix $G : \Z^r \rightarrow \Z^m$, 
   we define the equivalence relation $\sim_G$ over 
   $\N^m$ by $\bn \sim_G \bn'$ iff $\bn'-\bn\in\im{G}$.
\end{definition}
This is a well-defined equivalence relation because 
$\im{G}$ is an additive subgroup of $\Z^m$.
\cref{prop:step} suggests that for any vector 
$\ket{\psi}\in\ker{\a^{\bg^+} - \balpha^{\bg}\a^{\bg^-}}$ and for $\bn,\bn'\in\N^m$ 
such that $\bn\sim_{\bg} \bn'$, the Fock components $\psi_{\bn}$ and 
$\psi_{\bn'}$ are related by some proportionality factor. 
\cref{fig:lattice} provides some intuition on how this 
equivalence relation comes into play when $\ket{\psi}$ belongs to two kernels.

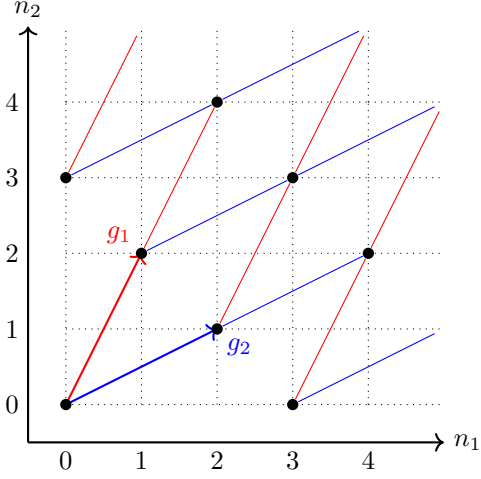
\begin{figure} 
    \centering
    \begin{tikzpicture}[scale=2]
        \begin{scope}[shift={(0.5,0.5)}]
            \draw[thick,->] (0,0) -- (5.5,0) node[right] {$n_1$};
            \draw[thick,->] (0,0) -- (0,5.5) node[above] {$n_2$};
        
            \foreach \x in {0,1,2,3,4} {
                \node[below] at (\x+0.5,0) {\x};
            }
            \foreach \y in {0,1,2,3,4} {
                \node[left] at (0,\y+0.5) {\y};
            }
        \end{scope}

        \draw[dotted, step=1] (1,1) grid (5,5);

        \foreach \x in {1,2,3,4,5} {
            \draw[dotted, step=1] (\x,5) -- (\x,6);
        }
        \foreach \y in {1,2,3,4,5} {
            \draw[dotted, step=1] (5,\y) -- (6,\y);
        }

        \draw[thick,->,color=red] (1,1) -- (2,3) node[above left] {$g_1$};
        \draw[thick,->,color=blue] (1,1) -- (3,2) node[below right] {$g_2$};

        \node[circle,fill=black,inner sep=1.5pt] (0) at (1,1) {};
        \node[circle,fill=black,inner sep=1.5pt] (1) at (2,3) {};
        \node[circle,fill=black,inner sep=1.5pt] (2) at (3,2) {};
        \node[circle,fill=black,inner sep=1.5pt] (3) at (5,3) {};
        \node[circle,fill=black,inner sep=1.5pt] (4) at (3,5) {};
        \node[circle,fill=black,inner sep=1.5pt] (5) at (1,4) {};
        \node[circle,fill=black,inner sep=1.5pt] (6) at (4,1) {};
        \node[circle,fill=black,inner sep=1.5pt] (7) at (4,4) {};
        \node[] (8) at (2,6) {};
        \node[] (9) at (5,6) {};
        \node[] (10) at (6,5) {};
        \node[] (11) at (6,2) {};

        \draw[color=red] (1) -- (4);
        \draw[color=red] (2) -- (7);
        \draw[color=red] (6) -- (3);
        \draw[color=red] (3) -- (10);
        \draw[color=red] (7) -- (9);
        \draw[color=red] (5) -- (8);
        
        \draw[color=blue] (1) -- (7);
        \draw[color=blue] (2) -- (3);
        \draw[color=blue] (5) -- (4);
        \draw[color=blue] (4) -- (9);
        \draw[color=blue] (7) -- (10);
        \draw[color=blue] (6) -- (11);
    \end{tikzpicture}
    \caption{
      \textbf{Graphical representation of $\ker{\a^{\bg_1}-\balpha^{\bg_1}}\cap\ker{\a^{\bg_2}-\balpha^{\bg_2}}$:} 
      For a state $\ket{\psi}$ to belong to both kernels, denoting 
      $\psi_{\bn}$ the coefficients of $\ket{\psi}$ in the Fock basis,
      if $\psi_{(0,0)}\neq0$,
      then all the $\psi_{\bn'}$s where $\bn'$ corresponds to a 
      coordinate marked by a black bullet on the graph must be 
      proportional to $\psi_{(0,0)}$.}
    \label{fig:lattice}
\end{figure}

Motivated by \cref{prop:step}, we define the function 
\begin{equation}
   \begin{aligned}
      \Gamma : \N^m\times \N^m &\rightarrow \C \\
      \bn,\bn'&\mapsto \sqrt{\frac{\bn!}{\bn'!}}\balpha^{\bn'-\bn}
   \end{aligned}
\end{equation}
and remark that $\frac{1}{\Gamma(\bn,\bn')}=\Gamma(\bn',\bn)$ and 
$\Gamma(\bn,\bn')\Gamma(\bn',\bn'')=\Gamma(\bn,\bn'')$.

Then, define the subspace 
\begin{equation}
   \cK=\left\{\ket{\psi}\in\bigcap_{\bg\in\im{G}}\cD(\a^{\bg^+} - \balpha^{\bg}\a^{\bg^-}) :
   \forall \bn,\bn'\in\N^m, \bn\sim_G \bn' \Rightarrow \psi_{\bn'}
   =\Gamma(\bn,\bn')\psi_{\bn}\right\}.
\end{equation}
Remark that as $\balpha$ does not have any zero entry, $\cK\neq0$.

\begin{proposition}
   $\cK = \bigcap_{\bg\in\im{G}}\ker{\a^{\bg^+} - \balpha^{\bg}\a^{\bg^-}}$.
\end{proposition}

\begin{proof}
   Let $\ket{\psi}\in\bigcap_{\bg\in\im{G}}\ker{\a^{\bg^+} - \balpha^{\bg}\a^{\bg^-}}$
   and $\bn,\bn'\in\N^m$ such that $\bn\sim_G \bn'$. Consider the vector $\bg=\bn'-\bn\in\im{G}$. 
   $\ket{\psi}$ belongs in particular to 
   $\ker{\a^{\bg^+} - \balpha^{\bg}\a^{\bg^-}}$. So, by \cref{prop:step}, 
   $\forall \bl \in \N^m, \sqrt{(\bl+\bg^+)!}\psi_{\bl+\bg^+}=\balpha^{\bg}\sqrt{(\bl+\bg^-)!}\psi_{\bl+\bg^-}$.
   Because $\bn'=\bn+\bg$ and  $\supp{\bg^+}\cap\supp{\bg^-}=\varnothing$, 
   $\bn-\bg^-$ belongs to $\N^m$. So, taking $\bl=\bn-\bg^-$, we get 
   \begin{equation}
      \sqrt{\bn'!}\psi_{\bn'}=\balpha^{\bg}\sqrt{\bn!}\psi_{\bn} \Rightarrow \psi_{\bn'}
      =\Gamma(\bn,\bn')\psi_{\bn}.
   \end{equation}

   Now, let $\ket{\psi}\in\cK$ and consider its decomposition in 
   the Fock basis $\sum_{\bn \in \N^m} \psi_{\bn} \ket{\bn}$. 
   For any $\bg\in \im{G}$, 
   \begin{equation}
      (\a^{\bg^+} - \balpha^{\bg}\a^{\bg^-})\ket{\psi}=
      \sum_{\bn \in \N^m} \left(\sqrt{\frac{(\bn+\bg^+)!}{\bn!}}
      \psi_{\bn+\bg^+}-\balpha^{\bg}\sqrt{\frac{(\bn+\bg^-)!}{\bn!}}\psi_{\bn+\bg^-}\right)\ket{\bn}.
   \end{equation}
   Because for all $\bn\in\N^m, \bn+\bg^+\sim_G\bn+\bg^-$, 
   $\psi_{\bn+\bg^+}=\Gamma(\bn+\bg^-,\bn+\bg^+)\psi_{\bn+\bg^-}=\sqrt{\frac{(\bn+\bg^-)!}{(\bn+\bg^+)!}}\balpha^{\bg}\psi_{\bn+\bg^-}$ 
   implying that $(\a^{\bg^+} - \balpha^{\bg}\a^{\bg^-})\ket{\psi}=0$. 
\end{proof}

The definition of $\cK$ suggests a natural orthonormal basis:
for a point $\bn\in\N^m$, if $\psi_{\bn}=0$ and $\ket{\psi}\in\cK$, 
then for all the points $\bk$ in the 
equivalence class $[\bn]\in\N^m/\!\!\sim_G$, $\psi_{\bk}=0$. 
So, for each $[\bn]\in\N^m/\!\!\sim_G$, define the projector
$\Pi_{[\bn]}:=\sum_{k\in[\bn]}\ket{\bk}\bra{\bk}$ 
and the normalised projected coherent states 
$\ket{\balpha_{[\bn]}}:=\frac{1}{A_{[\bn]}}\Pi_{[\bn]}\ket{\balpha}$
where $A_{[\bn]}:=\|\Pi_{[\bn]}\ket{\balpha}\|$ and 
$\ket{\balpha}:=e^{-\left|\balpha\right|^2/2}\sum_{\bn\in\N^m}\frac{\balpha^{\bn}}{\sqrt{\bn!}}\ket{\bn}$.

\begin{proposition} \label{prop:basis}
   The family of vectors $\left\{\ket{\balpha_{[\bn]}} : [\bn] \in \N^m/\!\!\sim_G\right\}$ 
   forms an orthonormal basis of $\cK$.
\end{proposition}

\begin{proof}
   It is clear that $\ket{\balpha_{[\bn]}} \in \cK$ and $\|\ket{\balpha_{[\bn]}}\|=1$ 
   for all equivalence class $[\bn]$. 
   The orthogonality comes from the fact that
   for two different equivalence classes $[\bn],[\bn'] \in \N^m/\!\!\sim_G$, 
   $\Pi_{[\bn]}\Pi_{[\bn']}=\indicator{\left\{\n \in [\bn]\cap[\bn']\right\}}=0$
   as $[\bn]\cap[\bn']=\varnothing$. In other words, the support of $\ket{\balpha_{[\bn]}}$ 
   and $\ket{\balpha_{[\bn']}}$ in the Fock basis do not overlap. 
   Finally, consider $\ket{\phi} \in \H$ such that   
   \begin{equation}
      \forall \bn,\bn'\in\N^m, \bn\sim_G \bn' \Rightarrow \phi_{\bn'}
      =\Gamma(\bn,\bn')\phi_{\bn},
   \end{equation}
   and define 
   \begin{equation}
      \ket{\tilde{\phi}} =
      e^{|\balpha|^2/2} \sum_{[\bn] \in \N^m/\!\sim_G}
      A_{[\bn]}\frac{\sqrt{\bn!}}{\balpha^{\bn}}\phi_{\bn} \ket{\balpha_{[\bn]}}
      \footnote{Strictly speaking, one should first verify that $\ket{\tilde{\phi}}$ defines a genuine vector, i.e. that it has finite norm. For the sake of clarity, we postpone this verification: the Fock-basis decomposition of $\ket{\tilde{\phi}}$ coincides with that of a finite-norm vector, which retroactively justifies that $\ket{\tilde{\phi}}$ is well defined.}. 
   \end{equation}
   This expression does not depend on the choice of representatives:
   for two different representatives $\bn,\bn'\in[\bn]$, as $\ket{\phi}\in\cK$,
   $\frac{\sqrt{\bn!}}{\balpha^{\bn}}\phi_{\bn} = \frac{\sqrt{\bn!}}{\balpha^{\bn}}\Gamma(\bn',\bn)\phi_{\bn'} = \frac{\sqrt{\bn'!}}{\balpha^{\bn'}}\phi_{\bn'}$.
   We show that $\ket{\phi}$ and $\ket{\tilde{\phi}}$ have the same 
   decomposition in the Fock basis: for all $\bk \in \N^m$, $\bk$ belongs 
   to a unique equivalence class, namely $[\bk]$. Because,
   \begin{equation}
      \langle \bk|\balpha_{[\bn]}\rangle=
      \frac{1}{A_{[\bn]}}\bra{\bk}\Pi_{[\bn]}\ket{\balpha}
      =\left\{\begin{aligned}
         &0 &\text{if } \bk\notin[\bn] \\
         &\frac{\balpha^{\bk}}{A_{[\bk]}\sqrt{\bk!}}e^{-|\balpha|^2/2} &\text{otherwise}
      \end{aligned}\right.
   \end{equation}
   we conclude that
   \begin{equation}
      \langle \bk| \tilde{\phi}\rangle = e^{|\balpha|^2/2}\sum_{[\bn] \in \N^m/\!\sim_G}
      A_{[\bn]}\frac{\sqrt{\bn!}}{\balpha^{\bn}}\phi_{\bn}\langle \bk|\balpha_{[\bn]}\rangle=\phi_{\bk}.
   \end{equation}
\end{proof}

\begin{remark}
   With \cref{prop:basis} in hand, we can simplify the
   definition of $\cK$ to get rid of the domain condition.
   We clearly have 
   \begin{equation}
      \cK\subset\left\{\ket{\psi}\in\H :
      \forall \bn,\bn'\in\N^m, \bn\sim_G \bn' \Rightarrow \psi_{\bn'}
      =\Gamma(\bn,\bn')\psi_{\bn}\right\}.
   \end{equation}
   In the proof of \cref{prop:basis}, we decomposed any 
   $\ket{\phi}\in\H$ satisfying    
   \begin{equation}
      \forall \bn,\bn'\in\N^m, \bn\sim_G \bn' \Rightarrow \phi_{\bn'}
      =\Gamma(\bn,\bn')\phi_{\bn},
   \end{equation}
   in the basis 
   $\left\{\ket{\balpha_{[\bn]}} : [\bn] \in \N^m/\!\!\sim_G\right\}\subset\cK$,
   thus proving the other inclusion. Hence,
   \begin{equation}
      \cK=\left\{\ket{\psi}\in\H :
      \forall \bn,\bn'\in\N^m, \bn\sim_G \bn' \Rightarrow \psi_{\bn'}
      =\Gamma(\bn,\bn')\psi_{\bn}\right\}.
   \end{equation}
\end{remark}

\paragraph{Adding the $\n$-type constraints.}
\cref{prop:basis} provides a basis of $\cK$.
It remains to project it onto 
$\bigcap_{i\in\llbracket1,s\rrbracket}\ker{\bh_i\cdot\n - \Delta_i}$. 

\begin{theorem}\label{theorem:basis}
   If $\cT(G,H, \balpha, \bDelta)\neq0$,
   let $ \bn_0\in\Z^m$ such that $H\bn_0=\bDelta$. Then, the family of vectors 
   $\left\{\ket{\balpha_{[\bn]}} : [\bn] \in (\bn_0+\ker{H})\cap\N^m/\!\!\sim_G\right\}$ forms an 
   orthonormal basis of the code.
\end{theorem}

\begin{proof}
   We can straightforwardly write the orthogonal projector $\Pi_{H}$ onto 
   $\bigcap_{i\in\llbracket1,s\rrbracket}\ker{\bh_i\cdot\n - \Delta_i}$ defined as
   \begin{equation}
      \Pi_{H}\ket{\bn} := \indicator{\{H\bn=\bDelta\}}\ket{\bn}
   \end{equation}
   for all $\bn\in\N^m$.

   As \cref{prop:basis} provides an orthonormal basis of $\cK$, we get a generating set of 
   $\cT(G,H, \balpha, \bDelta)$
   by projecting each vector in the family. For each $[\bn] \in \N^m/\!\!\sim_G$,
   \begin{equation}
      \Pi_{H}\ket{\balpha_{[\bn]}}
      =\frac{e^{-|\balpha|^2/2}}{A_{[\bn]}}
      \sum_{\bn'\in[\bn]}\frac{\balpha^{\bn'}}{\sqrt{\bn'!}}
      \indicator{\{H\bn'=\bDelta\}}\ket{\bn'}
      =\indicator{\{H\bn=\bDelta\}}\ket{\balpha_{[\bn]}}
   \end{equation}
   because if $\bn'\in[\bn], H\bn'=H\bn+H(\bn'-\bn)=H\bn$ as $\bn'-\bn\in\im{G}$ and $\im{G}\subset \ker{H}$.
   So, $\left\{\ket{\balpha_{[\bn]}} : \forall [\bn] \in \N^m/\!\!\sim_G, H\bn=\bDelta\right\}$ is an 
   orthonormal generating set, hence an orthonormal basis.
   
   We remarked earlier that by definition, $\cK\neq0$. So, the Tiger code is trivial (reduced to the null vector) whenever 
   $\bigcap_{i\in\llbracket1,s\rrbracket}\ker{\bh_i\cdot\n - \Delta_i}=0$,
   meaning that there is no integer vector $\bn\in\N^m$ satisfying $H\bn=\bDelta$.
   When the code is not trivial, there exists $\bn_0\in\Z^m$ such that $H\bn_0=\bDelta$,
   and we get $\{\bn\in\N^m:H\bn=\bDelta\}=(\bn_0+\ker{H})\cap\N^m$.
\end{proof}

In this work, we assume that $\cT(G,H, \balpha, \bDelta)\neq0$,
meaning that there always exists $\bn_0\in\Z^m$ such that $H\bn_0=\bDelta$.
This allows to get rid of many trivial examples such as $H=\begin{pmatrix}
   2
\end{pmatrix}$ and $\bDelta=\begin{pmatrix}
   1
\end{pmatrix}$.

\subsection{Finite generation of the code}\label{sec:finite_gen}

So far, the definition of Tiger codes (\cref{def:TC}) involves an infinite 
intersection of kernels. As discussed in the introduction, one expects that 
the codespace could be stabilised via a suitable Lindblad dynamics. 

However, such a construction would require infinitely many dissipative 
processes, which is not physically realistic. To bridge the gap between the 
abstract \cref{def:TC} and a practical implementation via reservoir 
engineering, we show that the infinite intersection can in fact be reduced to 
a finite one. More precisely, we prove that it suffices to consider a finite 
set of constraints associated with suitably chosen elements in $\im{G}$.

To that end, we generalise the notion of Markov basis introduced in 
\cite{MarkovBasesDiacoSturm}. Due to the proximity with the original 
definition of Markov bases, we choose to keep the same name.

\begin{definition}\label{def:Markov_basis}
   Given a matrix $A\in\Z^{m\times r}$, a Markov basis of $A$ is a finite set 
   of integer vectors $\{\bb_1,\ldots,\bb_M\}\subset \im{A}$ such that for 
   every pair of vectors $\bu, \bv\in\N^m$ for which $\bu-\bv\in\im{A}$, 
   there exists a sequence of indices $i_1,\ldots,i_N\in\llbracket1,M\rrbracket$ 
   such that
   \begin{equation}
      \bu + \bb_{i_1} + . . . + \bb_{i_N} = \bv
   \end{equation}
   where $\bu\to \bu+\bb_{i_1}\to \bu+\bb_{i_1}+\bb_{i_2}\to\ldots\to \bv$ is a path of $\N^m$, 
   namely for every $k\in\llbracket 1,N\rrbracket$, $\bu+\sum_{j=1}^k\bb_{i_j}\in\N^m$.
\end{definition}

We explain in Appendix \ref{app:image_and_kernel} why this definition is a 
generalisation of the usual concept as introduced in \cite{MarkovBasesDiacoSturm}. 

Let $\K$ be a field and consider the polynomial ideal 
\begin{equation}
   \I_A:=\langle \bx^{\bu}-\bx^{\bv} : \bu,\bv\in\N^m, \bu-\bv\in\im{A}\rangle \subset \K[x_1,\ldots,x_m].
\end{equation}
The following theorem characterises Markov bases in terms of ideals.
It is also a straightforward generalisation of a standard theorem on Markov bases 
\cite{MarkovBasesDiacoSturm}.

\begin{theorem}\label{theorem:Markov_ideal}
   A finite set of vectors $\{\pm\bb_1\ldots,\pm\bb_M\}\subset\im{A}$ is a Markov basis 
   of $A\in\Z^{m\times r}$ if and only if 
   \begin{equation}
      \I_A=\langle \bx^{\bb_i^+}-\bx^{\bb_i^-} : i\in\llbracket1,M\rrbracket\rangle.
   \end{equation}
\end{theorem}

\begin{proof}
   Assume $\{\pm\bb_1,\ldots,\pm\bb_M\}\subset \im{A}$ is a Markov basis of $A$. \\
   Let $\I'=\langle \bx^{\bb_i^+}-\bx^{\bb_i^-} : i\in\llbracket1,M\rrbracket\rangle$. 
   It is clear that $\I'\subset\I_A$. For the other inclusion, take $\bu,\bv\in\N^m$ such that $\bu-\bv\in\im{A}$.
   There exists a path of $\N^m$ such that $\bu + \bb_{i_1} + . . . + \bb_{i_N} = \bv$. So,
   \begin{equation}
      \begin{aligned}
         \bx^{\bu}-\bx^{\bv} 
         &= (\bx^{\bu}-\bx^{\bu+\bb_{i_1}})+(\bx^{\bu+\bb_{i_1}}-\bx^{\bu+\bb_{i_1}+\bb_{i_2}})+\ldots+(\bx^{\bu+\bb_{i_1}+\ldots+\bb_{i_{N-1}}}-\bx^{\bv}) \\
         &= \bx^{\bu-\bb_{i_1}^-}(\bx^{\bb_{i_1}^-}-\bx^{\bb_{i_1}^+})+(\bx^{\bu+\bb_{i_1}}-\bx^{\bu+\bb_{i_1}+\bb_{i_2}})+\ldots+(\bx^{\bu+\bb_{i_1}+\ldots+\bb_{i_{N-1}}}-\bx^{\bv}) \\
         &=\bx^{\bu-\bb_{i_1}^-}(\bx^{\bb_{i_1}^-}-\bx^{\bb_{i_1}^+})
         +\sum_{k=2}^{N-1} \bx^{\bu+\sum_{j=1}^{k-1}\bb_{i_j}-\bb_{i_k}^-}\left(\bx^{\bb_{i_k}^-}-\bx^{\bb_{i_k}^+}\right) \in \I'
      \end{aligned}
   \end{equation}
   where the last equality is justified by the fact that 
   $\forall \bn\in\N^m, \bn+\bb^+-\bb^-\in\N^m \Rightarrow \bn-\bb^-\in\N^m$
   as $\supp{\bb^+}\cap\supp{\bb^-}=\varnothing$.

   Now, assume that $\I'=\I_A$ and consider $\bu,\bv\in\N^m$ such 
   that $\bu-\bv\in\im{A}$. Then, $\bx^{\bu}-\bx^{\bv}\in\I_A=\I'$. 
   So, $\bx^{\bu}-\bx^{\bv}$ can be decomposed as
   \begin{equation}\label{eq:decomp_ideal}
      \begin{aligned}
         \bx^{\bu}-\bx^{\bv} &= \sum_{i=1}^{M}p_i(\bx)\left(\bx^{\bb_i^+}-\bx^{\bb_i^-}\right) \\
         &=\sum_{i=1}^{M} \sum_{\bj\in\N^n}\mu_{\bj}^{(i)}
         \left(\bx^{\bb_i^++\bj}-\bx^{\bb_i^-+\bj}\right),
      \end{aligned}
   \end{equation}
   where $p_i(\bx)=\sum_{\bj\in\N^n}\mu_{\bj}^{(i)}\bx^{\bj}$ for a given sequence of polynomials 
   $\{p_i\}_{1\leq i\leq M}$. 

   Up to a factorisation, we assume without loss of generality that each term 
   $\left(\bx^{\bb_i^++\bj}-\bx^{\bb_i^-+\bj}\right)$ appears only once in the sum 
   in \cref{eq:decomp_ideal}.
   Hence, because the sum is equal to $\bx^{\bu}-\bx^{\bv}$, there exists 
   $(i_1,\bj_1)\in\llbracket1,M\rrbracket\times\N^m$
   such that one of the term of the sum contains $\bx^{\bu}$, i.e., $\bu=\bb_{i_1}^++\bj_1$
   and $\mu_{\bj_1}^{(i_1)}=1$. Thus, $\bb_{i_1}^-+\bj_1=\bu-\bb_{i_1}$ and \cref{eq:decomp_ideal} 
   then becomes
   \begin{equation}
      \begin{aligned}
         &\bx^{\bu}-\bx^{\bv} =(\bx^{\bu}-\bx^{\bu-\bb_{i_1}})+\sum_{i,j\neq i_1,j_1}\mu_{\bj}^{(i)}\left(\bx^{\bb_i^++\bj}-\bx^{\bb_i^-+\bj}\right)\\
         \Leftrightarrow\quad&\bx^{\bu-\bb_{i_1}}-\bx^{\bv} =\sum_{i,j\neq i_1,j_1}\mu_j^{(i)}\left(\bx^{\bb_i^++\bj}-\bx^{\bb_i^-+\bj}\right).
      \end{aligned}
   \end{equation}
   As $\bu\to \bu-\bb_{i_1}$ is a path of $\N^m$, iterating this process until there is no 
   term left in the sum gives a path of $\N^m$ from $\bu$ to $\bv$.
\end{proof}

Because the ring $\K[x_1,\ldots,x_m]$ is Noetherian, by the Hilbert basis theorem, this proposition guarantees that any matrix $A$ admits
a finite Markov basis. In our case, for computations, the simplest choice
is $\K=\F_2$, the finite field with two elements.

In particular, there is a known generating set of $\I_A$ obtained via the Graver basis 
\cite{sturmfelsgröbner,Santos_2003,Hoekstra2013}. 
Again, it is a generalisation of the original concept already existing in the literature, 
but we choose to keep the same terminology.

\begin{definition}\label{def:graver_basis}
   Given a matrix $A\in\Z^{m\times r}$, the Graver basis of $A$ denoted $\cG_A$ is 
   the set of minimal elements of $\im{A}\setminus\{0\}$ according to the 
   partial order over $\Z^m$ defined as $\ba\preceq \bb$ iff for all $i\in\llbracket1,m\rrbracket$
   $0\leq a_i\leq b_i$ or $b_i\leq a_i\leq 0$.
\end{definition}

The partial order defined in \cref{def:graver_basis} is illustrated in \cref{fig:graver_basis}. 
As an example, considering $A=\begin{pmatrix}
   2\\-2
\end{pmatrix}$, as $\im{A}$ intersects nontrivially with only two orthants
of $\Z^2$, the Graver basis contains only two elements which are 
$\begin{pmatrix}
   2\\-2
\end{pmatrix}$ and $\begin{pmatrix}
   -2\\2
\end{pmatrix}$.

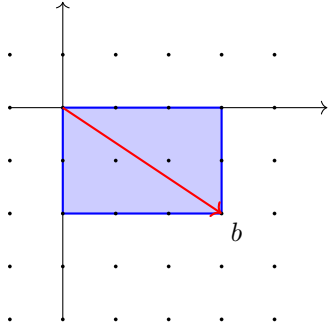
\begin{figure} 
   \centering
   \begin{tikzpicture}[scale=1.4]
      \draw[->] (-1,0) -- (5,0);
      \draw[->] (0,-4) -- (0,2);

      \coordinate (b) at (3,-2);

      \fill[blue!20] (0,0) rectangle (b);

      \draw[blue, thick] (0,0) rectangle (b);

      \foreach \x in {-1,0,1,2,3,4}
      \foreach \y in {-4,-3,-2,-1,0,1}
         \fill (\x,\y) circle (1pt);

      \draw[->, thick, red] (0,0) -- (b);

      \node[below right] at (b) {$b$};
   \end{tikzpicture}
    \caption{
      \textbf{Graphical representation of the partial order in \cref{def:graver_basis}:} 
      Let $\bb=\begin{pmatrix}
         3 & -2
      \end{pmatrix}^\top\in\Z^2$. The blue shaded area represents all the vectors of $\Z^2$
      less than or equal to $\bb$.}
    \label{fig:graver_basis}
\end{figure}

\begin{proposition}
   For any $A\in\Z^{n\times r}$, $\cG_A$ is finite and 
   $\I_A=\langle \bx^{\bu^+}-\bx^{\bu^-}:u\in\cG_A\rangle$.
\end{proposition} 

\begin{proof}
   For each orthant of $\Z^m$---namely, $\sigma_1\N\times\ldots\times\sigma_m\N$ 
   where $\sigma_i\in\{+,-\}$---the partial order $\preceq$ coincides with the partial
   product order of $\N^m$. In this respect, $\im{A}$ restricted to 
   an orthant is a subset of $\N^m$ and so 
   admits a finite set of least elements. As there is only a finite number 
   of orthants, there is a finite number of least elements in $\im{A}$.

   Now, let $\I_{\cG_A}=\langle \bx^{\bu^+}-\bx^{\bu^-}:\bu\in\cG_A\rangle$. 
   $\I_{\cG_A}\subset\I_A$ is clear as $\cG_A\subset\im{A}$.
   By contradiction, assume there exist $\bu,\bv\in\N^m$ such that $\bu-\bv\in\im{A}$ 
   and $\bx^{\bu}-\bx^{\bv}\notin\I_{\cG_A}$. Consider a least element $\bu-\bv$ of
   \begin{equation}
      \left(\left\{\bu-\bv\in\im{A}:
      \bx^{\bu}-\bx^{\bv}\notin\I_{\cG_A}\right\},\preceq\right).
   \end{equation}
   Because $\bu-\bv$ is not in the Graver basis, 
   it means that there exists a Graver basis element $\ba$ such that 
   $\ba\prec \bu-\bv$, which means $\bw^+:=\bu-\ba^+ \geq0$ and $\bw^-:=\bv-\ba^-\geq0$ 
   componentwise and $\bw^+,\bw^-\neq0$.
   Hence, 
   \begin{equation}
      \bx^{\bu}-\bx^{\bv}=\bx^{\bw^+}\bx^{\ba^+}-\bx^{\bw^-}\bx^{\ba^-}
      =\bx^{\bw^+}(\bx^{\ba^+}-\bx^{\ba^-}) + \bx^{\ba^-}(\bx^{\bw^+}-\bx^{\bw^-}).
   \end{equation}
   Because $\bx^{\bw^+}(\bx^{\ba^+}-\bx^{\ba^-})\in\I_{\cG_A}$, we get that 
   $\bx^{\bw^++\ba^-}-\bx^{\bw^-+\ba^-}\notin\I_{\cG_A}$.
   But as 
   \begin{equation}
      \bw^+-\bw^-=\bu-\bv-\ba \in\im{A} \qquad \text{and} \qquad
      \bw^+-\bw^-\prec \bu-\bv,
   \end{equation}
   this contradicts the minimality of $\bu-\bv$.
   Hence, $\left\{\bu-\bv\in\im{A}:\bx^{\bu}-\bx^{\bv}\notin\I_{\cG_A}\right\}=\varnothing$ 
   which proves that $\I_A\subset\I_{\cG_A}$. 
\end{proof}

By \cref{theorem:Markov_ideal}, the Graver basis is a Markov
basis. It may therefore be used directly. Alternatively, 
one may work with the associated polynomial ideal and 
compute a Gröbner basis for a chosen monomial order, 
which can yield a Markov basis with fewer elements. 

We can now prove that any Tiger code can be generated 
as the intersection of a finite number of kernels.

\begin{proposition}\label{prop:finite_generation}
   Let $\{\bb_1,\ldots,\bb_M\}\subset\im{G}$ be a Markov basis of $G$. 
   Then, 
   \begin{equation}
      \cK = \bigcap_{i\in\llbracket1,M\rrbracket}\ker{\a^{\bb_i^+}-\balpha^{\bb_i}\a^{\bb_i^-}}.
   \end{equation}
\end{proposition}

\begin{proof}
   Because $\cK = \bigcap_{\bg\in\im{G}}\ker{\a^{\bg^+} - \balpha^{\bg}\a^{\bg^-}}\subset\bigcap_{i\in\llbracket1,M\rrbracket}\ker{\a^{\bb_i^+}-\balpha^{\bb_i}\a^{\bb_i^-}}$, 
   one inclusion is clear. 
   Let $\ket{\psi}\in\bigcap_{i\in\llbracket1,M\rrbracket}\ker{\a^{\bb_i^+}-\balpha^{\bb_i}\a^{\bb_i^-}}$ 
   and $\bn,\bn'\in\N^m$ such that $\bn\sim_G\bn'$.
   Consider the path $\bn'\to \bn' + \bb_{i_1}\to\ldots\to \bn' + \bb_{i_1}+\ldots+ \bb_{i_N} = \bn$ in $\N^m$. 
   As \break $\bn'\sim_G \bn' + \bb_{i_1} \sim_G \ldots \sim_G \bn' + \bb_{i_1} + \ldots + \bb_{i_N}$, 
   using \cref{prop:step}, we have
   \begin{equation}
      \begin{aligned}
         \psi_{\bn'}&=\Gamma(\bn' + \bb_{i_1},\bn')\psi_{\bn'+\bb_{i_1}}\\
         &=\ldots \\
         &=\Gamma(\bn'+\bb_{i_1} + \ldots + \bb_{i_N},\bn'+ \bb_{i_1} + \ldots + \bb_{i_{N-1}})\ldots\Gamma(\bn' + \bb_{i_1},\bn')\psi_{\bn' + \bb_{i_1} + \ldots + \bb_{i_N}} \\
         &=\Gamma(\bn,\bn')\psi_{\bn}
      \end{aligned}
   \end{equation}
   where we used the fact that $\Gamma(\bn,\bn')\Gamma(\bn',\bn'')=\Gamma(\bn,\bn'')$. 
\end{proof}

Note that the previous proposition does not take into account 
the $\n$-type constraints. In fact, when taken into account, it
is possible for some Tiger code to be generated by fewer
$\a$-type constraints as shown in the next example.

\begin{example}\label{ex:binomial_code}
   The two-mode binomial code \cite{Michael_2016} is a bosonic code
   defined by an integer $\Delta\in\N$ as 
   \begin{equation}
      \cB_\Delta=\ker{\hn_1+\hn_2-\Delta}\cap\ker{\ha_1^2-\ha_2^2}.
   \end{equation} 
   This corresponds to the Tiger code where $H=\begin{pmatrix}
      1&1
   \end{pmatrix}$, $G=\begin{pmatrix}
      2\\
      -2
   \end{pmatrix}$, $\bDelta=\Delta$ and $\balpha=\begin{pmatrix}
      \alpha\\
      \alpha
   \end{pmatrix}$ for some $\alpha\in\C$ such that $\alpha\neq0$.
   Here, 
   \begin{equation}
      \left\{\begin{pmatrix}
      2\\
      -2
   \end{pmatrix},\begin{pmatrix}
      -2\\
      2
   \end{pmatrix}\right\}
   \end{equation}
   clearly forms a Markov basis of $G$. As such, according to \cref{prop:finite_generation},
   \begin{equation}
      \bigcap_{\bg\in\im{G}}\ker{\a^{\bg^+}-\balpha^{\bg}\a^{\bg^-}}=
      \ker{\ha_1^2-\ha_2^2}\cap\ker{\ha_2^2-\ha_1^2}
      =\ker{\ha_1^2-\ha_2^2}.
   \end{equation}

   In the particular case where $\Delta=1$, then 
   $\ker{\hn_1+\hn_2-\Delta}=\Span{\ket{0,1},\ket{1,0}}$ and so
   $\ker{\hn_1+\hn_2-\Delta}\subset\ker{\hat{a}_1^2-\hat{a}_2^2}$.
   Hence, 
   \begin{equation}
      \cB_\Delta=\Span{\ket{0,1},\ket{1,0}}=\ker{\hn_1+\hn_2-\Delta}.
   \end{equation} 
   Thus, for $\Delta=1$, the $\a$-type constraint is superfluous.
\end{example}

\section{Homology and the codespace}\label{sec:homology}

According to \cref{theorem:basis}, the algebraic structure of a Tiger 
code is dictated by the set \break $(\bn_0+\ker{H})\cap\N^m/\!\!\sim_G$. 
Due to its resemblance with the \emph{homology group} \break
$H_1(C_\bullet)=\ker{H}\!/\!\im{G}$ of the chain complex 
$C_\bullet:\Z^r\xrightarrow{G}\Z^m\xrightarrow{H}\Z^s$, 
we investigate in this section in which case the algebraic structure of 
our set of interest can be captured by that of 
the homology group---which is arguably easier to study.

For readability, denote $K_H^{\bDelta}:=(\bn_0+\ker{H})\cap\N^m$.
First, we introduce an embedding of $K_H^{\bDelta}/\!\!\sim_G$ into 
$H_1(C_\bullet)$:
\begin{equation}\label{eq:iota}
   \begin{aligned}
      \iota : K_H^{\bDelta}/\!\!\sim_G & \rightarrow H_1(C_\bullet) \\
      [\bn]&\mapsto [\bn- \bn_0].
   \end{aligned}
\end{equation}
While this map is injective, it turns out to be surjective as well in some cases.
Below, we present a sufficient condition for the map $\iota$ to be surjective.

\begin{definition}
   A Tiger code $\cT(G,H, \balpha, \bDelta)$
   is \emph{positive} if there exists a vector $\bg\in\im{G}$ such that $g>0$ 
   componentwise.
   A Tiger code that is not positive is called a \emph{non-positive} Tiger code.
\end{definition}

\begin{proposition}\label{prop:iota}
   $\iota$ is always injective, and is bijective if the Tiger code 
   is positive.
\end{proposition}

\begin{proof}
   Consider two elements $[\bn],[\bn']\in K_H^{\bDelta}/\!\!\sim_G$ 
   such that $\iota([\bn])=\iota([\bn'])$. Then, 
   \break $\bn-\bn'=\bn-\bn_0-(\bn'-\bn_0)\in\im{G}$, 
   which proves that $\bn\sim_G\bn'$ and $[\bn]=[\bn']$.
   For surjectivity, consider $[\bn]\in H_1(C_\bullet)$. As the Tiger code is positive, 
   there exists $\bg\in\im{G}\subset\ker{H}$ such that $\bg>0$. So, for $\lambda\in\N$ large enough, 
   $\bn_0+\bn+\lambda \bg \in K_H^{\bDelta}$.
   Thus, $\iota([\bn+\bn_0+\lambda \bg])=[\bn+\lambda \bg]=[\bn]$.
\end{proof}

Note that this proposition gives a sufficient condition which is not a necessary 
one. The non-positive Tiger code in the following example makes $\iota$ surjective.

\begin{example}\label{ex:binomial_code2}
   Consider the two-mode binomial code introduced in \cref{ex:binomial_code}. This Tiger code 
   is non-positive. However, for $\Delta>0$,
   \begin{equation}
      K_H^{\bDelta}=\{\bn\in\N^2: H\bn=\bDelta\}=\left\{\begin{pmatrix}
         i\\
         \Delta-i
      \end{pmatrix}:i\in\llbracket0,\Delta\rrbracket\right\}
   \end{equation}
   and thus 
   \begin{equation}
      K_H^{\bDelta}/\!\!\sim_G=\left\{\left[\begin{pmatrix}
         0\\
         \Delta
      \end{pmatrix}\right],\left[\begin{pmatrix}
         1\\
         \Delta-1
      \end{pmatrix}\right]\right\}.
   \end{equation}
   On \cref{fig:binomial_code}, the equivalence class 
   $\left[\begin{pmatrix}
         0\\
         \Delta
      \end{pmatrix}\right]$ is represented by the red points while 
   the equivalence class of $\left[\begin{pmatrix}
         1\\
         \Delta-1
      \end{pmatrix}\right]$ is represented by the blue points.
   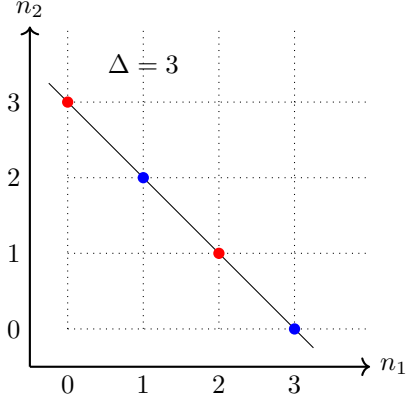
\begin{figure} 
      \centering
      \begin{tikzpicture}[scale=2]
        \begin{scope}[shift={(0.5,0.5)}]
            \draw[thick,->] (0,0) -- (4.5,0) node[right] {$n_1$};
            \draw[thick,->] (0,0) -- (0,4.5) node[above] {$n_2$};
        
            \foreach \x in {0,1,2,3} {
                \node[below] at (\x+0.5,0) {\x};
            }
            \foreach \y in {0,1,2,3} {
                \node[left] at (0,\y+0.5) {\y};
            }
        \end{scope}

        \draw[dotted, step=1] (1,1) grid (4,4);

        \foreach \x in {1,2,3,4} {
            \draw[dotted, step=1] (\x,4) -- (\x,5);
        }
        \foreach \y in {1,2,3,4} {
            \draw[dotted, step=1] (4,\y) -- (5,\y);
        }

        \draw[color=black] (4.25,0.75) -- (0.75,4.25);

        \node[circle,fill=red,inner sep=1.5pt] (0) at (1,4) {};
        \node[circle,fill=blue,inner sep=1.5pt] (1) at (2,3) {};
        \node[circle,fill=red,inner sep=1.5pt] (2) at (3,2) {};
        \node[circle,fill=blue,inner sep=1.5pt] (3) at (4,1) {};
        \node[] (4) at (2,4.5) {$\Delta=3$};

    \end{tikzpicture}
      \caption{\textbf{Graphical representation of the two-mode binomial code for $\Delta=3$:} 
      Due to the constraint $H\bn=n_1+n_2=\Delta$, the points in $\{\bn\in\N^2:H\bn=\Delta\}$
      belong to the black line. Two points on this line distant by a multiple 
      of $\begin{pmatrix}2 & -2\end{pmatrix}^\top$ belong to the same equivalence class.
      Hence, the points are partitioned into two sets---red and blue---corresponding to the two
      equivalence classes.}
      \label{fig:binomial_code}
   \end{figure}
   Meanwhile, 
   \begin{equation}
      H_1(C_\bullet)=\text{ker}\begin{pmatrix}
         1 & 1
      \end{pmatrix}/\text{im}\begin{pmatrix}
         2\\ -2
      \end{pmatrix}=\{\left[\begin{pmatrix}
         0\\0
      \end{pmatrix}\right],\left[\begin{pmatrix}
         1\\-1
      \end{pmatrix}\right]\}.
   \end{equation}
   Hence, 
   \begin{equation}
      \iota\left(\left[\begin{pmatrix}
         0\\0
      \end{pmatrix}\right]\right) = \left[\begin{pmatrix}
         0\\ \Delta
      \end{pmatrix}\right] \quad \text{and} \quad 
      \iota\left(\left[\begin{pmatrix}
         1\\-1
      \end{pmatrix}\right]\right)=\left[\begin{pmatrix}
         1\\\Delta-1
      \end{pmatrix}\right]
   \end{equation}
   thus proving that $\iota$ is indeed bijective.
\end{example}

Combined with \cref{theorem:basis}, the injection $\iota$ of \cref{eq:iota} 
allows for picturing the codespace of $\cT(G,H,\balpha,\bDelta)$ as a subspace of 
$L^2\left(H_1(C_\bullet)\right)$.

\begin{theorem} \label{theorem:pi}
   Consider the bounded linear map 
   \begin{equation}
      \begin{aligned}
         \pi : L^2\left(H_1(C_\bullet)\right) & \rightarrow \cT(G,H, \balpha, \bDelta) \\
         f &\mapsto 
         \sum_{[\bn]\in K_H^{\bDelta}/\!\sim_G}f\circ\iota([\bn])
         \ket{\balpha_{[\bn]}}.
      \end{aligned}
   \end{equation}
   Then, $\cT(G,H, \balpha, \bDelta)\cong L^2\left(H_1(C_\bullet)\right)/\ker{\pi}$.
   Using the isomorphism $H_1(C_\bullet)\cong \Z^f\times \Z_{d_1}\times\ldots\times\Z_{d_t}$ (\cref{theorem:H1isom}), 
   we deduce that 
   \begin{equation}
      \cT(G,H, \balpha, \bDelta)\cong L^2(\Z^f)\otimes \C^{d_1}\otimes\ldots\otimes \C^{d_t}/\ker{\pi}
   \end{equation}
   where $d_{i}$ divides $d_{i+1}$. In particular, for positive Tiger codes, $\ker{\pi}=0$.
\end{theorem}

\begin{proof}
   We prove that $\pi$ is bounded and then prove surjectivity.
   For any $f\in L^2\left(H_1(C_\bullet)\right)$,
   \begin{equation}
      \|\pi(f)\|^2=
      \sum_{[\bn]\in K_H^{\bDelta}/\!\sim_G}|f\circ\iota([\bn])|^2=
      \sum_{[\bx]\in \iota\left(K_H^{\bDelta}/\!\!\sim_G\right)}|f\left([\bx]\right)|^2
      \leq \|f\|^2
   \end{equation}
   as $\iota\left(K_H^{\bDelta}/\!\!\sim_G\right)\subseteq H_1(C_\bullet)$.

   Then, because $\left\{\ket{\balpha_{[\bn]}} : [\bn] \in K_H^{\bDelta}/\!\!\sim_G\right\}$
   forms an orthonormal basis of the codespace (\cref{theorem:basis}) 
   any vector $\ket{\psi}\in\cT(G,H, \balpha, \bDelta)$ can be decomposed as 
   \begin{equation}
      \ket{\psi}
      =\sum_{[\bn] \in K_H^{\bDelta}/\!\sim_G}\psi_{[\bn]}\ket{\balpha_{[\bn]}}.
   \end{equation}
   So, as $\iota$ is injective by \cref{prop:iota}, taking 
   \begin{equation}
      f:[\bx]\mapsto\left\{\begin{aligned}
         &\psi_{\iota^{-1}([\bx])} & \text{if } [\bx]\in \iota(K_H^{\bDelta}/\!\!\sim_G)\\
         &0 & \text{otherwise},
      \end{aligned}\right.
   \end{equation}
   $f$ belongs to $L^2\left(H_1(C_\bullet)\right)$ and
   we get $\pi(f)=\ket{\psi}$.
\end{proof}

According to \cref{theorem:pi}, whenever $\ker{\pi}=0$, the
Tiger code contains $f$ rotors---meaning $f$ copies of $L^2(\Z)$---and qudits of dimension $d_1,\ldots,d_t$, respectively.
$\ker{\pi}$ happens to be trivial when $\iota\left(K_H^{\bDelta}/\!\!\sim_G\right)=H_1(C_\bullet)$,
in which case $\pi$ is an isometry: it preserves the norm and has a bounded inverse.

As a direct consequence of \cref{prop:iota}, we know that being positive 
is a \emph{sufficient} but not \emph{necessary} condition
for a Tiger code to claim that $\ker{\pi}=0$. 

To conclude this section, we explicitly describe the construction 
of the isomorphism \break 
$H_1(C_\bullet)\cong \Z^f\times \Z_{d_1}\times\ldots\times\Z_{d_t}$
aforementioned in \cref{theorem:pi}.

\begin{theorem}\label{theorem:H1isom}
   There exist $f,d_1,\ldots,d_t\in\Z$ such that
   \begin{equation}
      H_1(C_\bullet)\cong \Z^f\times \Z_{d_1}\times\ldots\times\Z_{d_t}
   \end{equation}
   where for all $i\in\llbracket1,t-1\rrbracket$, $d_i$ divides $d_{i+1}$
   and $f+t\leq m$.
\end{theorem}

\begin{proof}
   Consider the Smith normal form of $G$ written:
   \begin{equation}
      LGC=D
   \end{equation}
   where $L\in\Z^{m\times m}$ and $C\in\Z^{r\times r}$ are invertible and 
   \begin{equation}
      D = \diag{m\times r}{d_1,\ldots,d_t}
      = \begin{pmatrix}
         d_1 & 0   & \cdots & 0 & 0 \\
         0   & d_2 & \ddots & \vdots & \vdots \\
         \vdots & \ddots & \ddots & 0 & \vdots \\
         0   & \cdots & 0 & d_t & 0 \\
         \vdots &        &        & \vdots & \vdots \\
         0   & \cdots & \cdots & 0 & 0
         \end{pmatrix}
   \end{equation}
   where $d_i$ divides $d_{i+1}$.
   Denoting $\{\bl_{d_1},\ldots,\bl_{d_t},\bl_1, \ldots,\bl_{f'}\}$ the columns of $L^{-1}$
   such that $m=t+f'$, they form a basis of $\Z^m$. Because $GC=L^{-1}D$ and $\im{GC}=\im{G}$,
   we have 
   \begin{equation}
      \begin{aligned}
         \Z^m/\im{G} &= \Z_{d_1}\bl_{d_1}\oplus \ldots \oplus \Z_{d_t}\bl_{d_t} \oplus \Z\bl_1 \oplus \ldots\oplus\Z\bl_{f'} \\
         &\cong \Z_{d_1} \times \ldots \times \Z_{d_t}\times \Z^{f'}.
      \end{aligned}
   \end{equation} 
   Now restricting to $\ker{H}$; note that  
   the first $t$ columns $\{\bl_{d_1},\ldots,\bl_{d_t}\}$ of $L^{-1}$ belong to $\ker{H}$.
   Indeed, denoting $\{\be_i\}_{1\leq i\leq m}$ the canonical basis of $\R^m$,
   \begin{equation}
      H\bl_{d_i}=HL^{-1}D\frac{\be_i}{d_i}=\frac{1}{d_i}HGC\be_i=0.
   \end{equation}
   Hence, $\ker{H}\!/\!\im{G}$ has the same torsion 
   as $\Z^m/\im{G}$. It remains to complete $\{\bl_{d_1},\ldots,\bl_{d_t}\}$ to a basis
   of $\ker{H}$ with elements generated by
   the last columns $\{\bl_1, \ldots,\bl_{f'}\}$ of $L^{-1}$ to get the free part of 
   $H_1(C_\bullet)$. 
   Up to a change of basis, we assume without loss of generality that the first columns 
   $\{\bl_1, \ldots,\bl_{f}\}$ with $f\leq f'$ complete the basis, then
   \begin{equation}
      \begin{aligned}
         H_1(C_\bullet) &= \Z_{d_1}\bl_{d_1}\oplus \ldots \oplus \Z_{d_t}\bl_{d_t} \oplus \Z\bl_1 \oplus \ldots\oplus\Z\bl_{f} \\
         &\cong \Z_{d_1} \times \ldots \times \Z_{d_t}\times \Z^f.
      \end{aligned}
   \end{equation}
   So, denoting $\tilde{L}\in\Z^{(t+f)\times m}$ the matrix containing the 
   $t+f$ first rows of $L$, its right inverse is 
   $\tilde{L}^{-1}=\left[\bl_{d_1}| \ldots | \bl_{d_t} |\bl_1 | \ldots| \bl_f\right]$, 
   meaning the matrix containing the $t+f$ first columns of $L^{-1}$ and
   \begin{equation}\label{eq:eta}
      \begin{aligned}
         \eta : \Z_{d_1} \times \ldots \times \Z_{d_t}\times \Z^f & \rightarrow H_1(C_\bullet) \\
         \bgamma=(\gamma_{d_1},\ldots,\gamma_{d_t}, \gamma_1, \ldots,\gamma_f)^\top & \mapsto 
         [\tilde{L}^{-1}\bgamma]=\left[\gamma_{d_1}\bl_{d_1}+\ldots+\gamma_{d_t}\bl_{d_t} + \gamma_1\bl_1+ \ldots+\gamma_f\bl_f\right]
      \end{aligned}
   \end{equation}
   is an isomorphism.
\end{proof}

\section{Coherent states description of the code}\label{sec:Fourier_coherent}

Having introduced an orthonormal basis of the codespace and used 
it to obtain a homological description of the code, we now turn 
to a different description based on coherent states.

For instance, the cat code---that corresponds to $\ker{\hat{a}^2-\alpha^2}$---can 
be described by the span of either the coherent states belonging 
to the codespace or the \emph{even} and \emph{odd} cats:
\begin{equation}
   \{\ket{\alpha},\ket{-\alpha}\} \qquad \text{or} \qquad
   \left\{\sum_{n\in\N}\frac{\alpha^{2n}}{\sqrt{(2n)!}}\ket{2n},
   \sum_{n\in\N}\frac{\alpha^{2n+1}}{\sqrt{(2n+1)!}}\ket{2n+1}\right\}.
\end{equation}
So far, the Hilbert basis constructed in \cref{sec:code_construction} plays, 
for general Tiger codes, the role of the even/odd basis of the cat code. 
In the cat-code case, this basis is related by a Fourier transform 
to the coherent-state basis
\begin{equation}
   \left\{
   \frac{e^{|\alpha|^2/2}}{2}(\ket{\alpha}+\ket{-\alpha}),
   \frac{e^{|\alpha|^2/2}}{2}(\ket{\alpha}-\ket{-\alpha})
   \right\}.
\end{equation}

This example suggests that the codespace of certain Tiger codes may 
be described as the closure of the span of the coherent states 
it contains. However, as pointed out in the introduction, 
this property fails for general polynomial constraints. Indeed, 
for polynomials $P_1,\ldots,P_r\in\C[X_1,\ldots,X_m]$, the inclusion
\begin{equation}
\overline{\Span{\ket{\balpha}:P_i(\balpha)=0,\forall i\in\llbracket1,r\rrbracket}}
\subsetneq\bigcap_{i=1}^r\ker{P_i(\a)}
\end{equation}
is strict in general. For instance, for $P=X^2$, one has
\begin{equation}
   \Span{\ket{0}}\subsetneq\ker{\hat{a}^2}.
\end{equation}
By contrast, when the polynomials arise from the $\a$-type constraints 
of a Tiger code, this inclusion becomes an equality. Thus, the 
codespace can be described entirely in terms of the coherent states 
it contains.

To prove this non-trivial property of Tiger codes, we introduce 
a Fourier transform on the codespace, 
inherited directly from the Fourier transform on the homology 
group of the corresponding Tiger code.

\subsection{Pontryagin duality and Fourier transform}\label{sec:Pontryagin}

In this section, our focus is to define a Fourier transform between the homology and cohomology groups of a given
chain complex. We start by introducing the Pontryagin dual which allows for defining a Fourier transform between
a locally compact abelian group and its dual \cite{folland2015course}.

\paragraph{Pontryagin duality} For any locally compact abelian (LCA) group $\cG$, the Pontryagin dual $\hat{\cG}$
corresponds to the LCA group of \emph{continuous} homomorphisms from 
$\cG$ to $\T:=\{z \in \C : |z|=1\}$
\begin{equation}
    \hat{\cG} := \rm{Hom}(\cG,\T).
\end{equation}
In this work, we will mainly work with the dual of some known groups, namely
$\hat{\Z}\cong \T$, $\hat{\T}\cong\Z$ and $\hat{\Z}_d\cong\frac{2\pi}{d}\Z_d$ (Theorem 4.6 \cite{folland2015course}).
For two LCA groups $\cG_1$ and $\cG_2$, we can consider the 
product topology to form a new LCA group $\cG_1\times\cG_2$.
In particular, we have:

\begin{proposition}\label{prop:product_hom}
   (Proposition 4.7 \cite{folland2015course}) $\widehat{\cG_1\times\cG_2}\cong \widehat{\cG_1}\times \widehat{\cG_2}$.
\end{proposition}

So, the Pontryagin dual of the group $\Z_{d_1}\times\ldots\times\Z_{d_t}\times\Z^f$ 
is isomorphic to \break
$\frac{2\pi}{d_1}\Z_{d_1}\times\ldots\times\frac{2\pi}{d_t}\Z_{d_t}\times [0,2\pi)^f$ 
via the LCA group isomorphism
\begin{equation}\label{eq:upsilon}
   \begin{aligned}
      \upsilon : \frac{2\pi}{d_1}\Z_{d_1}\times\ldots\times\frac{2\pi}{d_t}\Z_{d_t}\times[0,2\pi)^f &\rightarrow \rm{Hom}(\Z_{d_1}\times\ldots\times\Z_{d_t}\times\Z^f,\T) \\
      \bphi & \mapsto (\bx\mapsto e^{i\bphi\cdot\bx}).
   \end{aligned}
\end{equation}
Finally, given a map $A:\cG_1 \rightarrow \cG_2$ between two 
LCA groups, $A^*:=\rm{Hom}(A,\T)$ is defined as 
\begin{equation}
   \begin{aligned}
      A^* : \widehat{\cG_2} & \rightarrow \widehat{\cG_1} \\
      \varphi & \mapsto \varphi\circ A.
   \end{aligned}
\end{equation}
In other words, $\rm{Hom}(-,\T)$ is a contravariant functor.

\paragraph{Cohomology} Given the chain complex $C_\bullet:\Z^r\xrightarrow{G}\Z^m\xrightarrow{H}\Z^s$, 
applying the functor $\rm{Hom}(-,\T)$ to each of the maps and groups yields its associated cochain complex
\begin{equation}
   \widehat{\Z^r}\xleftarrow{G^*}\widehat{\Z^m}\xleftarrow{H^*}\widehat{\Z^s}.
\end{equation}
We define the cohomology to be $H^1(C_\bullet)=\ker{G^*}/\im{H^*}$. It turns out that the cohomology 
is isomorphic to the Pontryagin dual of the homology group.

\begin{theorem}\label{theorem:UCT}
   The map $h:H^1(C_\bullet)\rightarrow \widehat{H_1(C_\bullet)}, [\varphi]\mapsto ([x]\mapsto \varphi(x))$
   is a well defined group isomorphism.
\end{theorem}

\begin{proof}
   This is a consequence of both the \emph{Universal Coefficient Theorem} \cite{Hatcher:478079} and the fact 
   that $\T$ is a divisible group.
\end{proof}

Although the definition of the cohomology group is abstract, 
it admits the following more operational description. 
Consider the quotient group
\begin{equation}
   \kerpi{G^\top}/\impi{H^\top}:=
   \{\bphi\in[0,2\pi)^m : G^\top\bphi=0 \text{ mod } 2\pi\}
   /\{H^\top\bphi \text{ mod } 2\pi : \bphi\in[0,2\pi)^s\}.
\end{equation}
The proposition below shows that this group is isomorphic 
to the cohomology group and that a generating set can be 
read off directly from the Smith normal form of $G$.
Let 
\begin{equation}\label{eq:tau}
   \begin{aligned}
      \tau : \kerpi{G^\top}/\impi{H^\top} & \to H^1(C_\bullet) \\
      [\bphi] & \mapsto [\bx\mapsto e^{i\bphi\cdot\bx}].
   \end{aligned}
\end{equation}
It is a well-defined morphism: if $\bphi\in\kerpi{G^\top}$, then 
$G^*(\bx\mapsto e^{i\bphi\cdot\bx}) = (\bx\mapsto e^{i\bphi^\top G\bx})=1$
which proves that $\bx\mapsto e^{i\bphi\cdot\bx}\in\ker{G^*}$. 
And if $\bphi=H^\top\bpsi\in\impi{H^\top}$, 
then \break 
$(\bx\mapsto e^{i(H^\top\bpsi)\cdot\bx})=H^*(\bx\mapsto e^{i\bpsi\cdot\bx})\in \im{H^*}$.

\begin{proposition}
   All the maps in the sequence 
   \begin{equation}
      \kerpi{G^\top}/\impi{H^\top} \xrightarrow{\tau} H^1(C_\bullet) 
      \xrightarrow{h} \widehat{H_1(C_\bullet)} 
      \xrightarrow{\upsilon^{-1}\circ\eta^*} 
      \frac{2\pi}{d_1}\Z_{d_1}\times\ldots\times\frac{2\pi}{d_t}\Z_{d_t}\times [0,2\pi)^f
   \end{equation}
   are group isomorphisms.
\end{proposition}

\begin{proof}
   The last map $\upsilon^{-1}\circ\eta^*$ is an isomorphism as 
   $\eta$ in \cref{eq:eta} 
   and $\upsilon$ in \cref{eq:upsilon} are both isomorphisms and that 
   the Pontryagin dual of an isomorphism is an isomorphism.
   \cref{theorem:UCT} guarantees that $h$ is an isomorphism as well. 
   So, it remains to prove that $\tau$ is an isomorphism. 

   It is surjective: for any $[\varphi]\in H^1(C_\bullet)$ as 
   $\varphi \in \widehat{\Z^m}$, there exists 
   $\bphi\in[0,2\pi)^m$ such that $\forall \bx\in\Z^m,\varphi(\bx)=e^{i\bphi\cdot\bx}$. 
   Because $G^*\varphi=0$, for any $\bx\in\Z^m$,
   $e^{i\bphi\cdot G\bx}=1$ implying that \break $\bphi\in\kerpi{G^\top}$. 
   Hence, $\tau([\bphi])=[\varphi]$.

   It is injective as well as if $(\bx\mapsto e^{i\bphi\cdot\bx})=H^*\varphi'$ 
   for some $\varphi'\in\widehat{\Z^s}$,
   there exists $\bphi'\in[0,2\pi)^s$ such that 
   $\forall \by\in\Z^s,\varphi'(\by)=e^{i\bphi'\cdot\by}$ and 
   $e^{i\bphi\cdot\bx}=e^{i\bphi'\cdot H\bx} \Rightarrow e^{i(\bphi-H^\top\bphi')\cdot\bx}=1$
   for all $\bx\in\Z^m$, meaning that $\bphi=H^\top\bphi'\in\impi{H^\top}$.
\end{proof}

Thus, we can construct a generating set of $\kerpi{G^\top}/\impi{H^\top}$ by 
tracking \break 
$[\bphi]\in \kerpi{G^\top}/\impi{H^\top}$ through the chain of isomorphisms:
\begin{equation}\label{eq:sequence_isom}
   [\bphi]\xmapsto{\quad\tau\quad}[\bx\mapsto e^{i\bphi\cdot\bx}]
   \xmapsto{\quad h\quad}([\bx]\mapsto e^{i\bphi\cdot\bx})
   \xmapsto{\quad \eta^*\quad} (\bz\mapsto e^{i\bphi\cdot\tilde{L}^{-1}\bz})
   \xmapsto{\quad\upsilon^{-1}\quad} (\tilde{L}^{-1})^\top\bphi.
\end{equation}
So, $\upsilon^{-1}\circ\eta^*\circ h\circ\tau$ is an isomorphism with inverse 
$\bphi\mapsto[\tilde{L}^\top\bphi]$, where $\tilde{L}$ is defined in the proof 
of \cref{theorem:H1isom}. 
Hence, denoting $\{\bw_{d_1},\ldots,\bw_{d_t},\bw_1\ldots,\bw_f\}$
the columns of $\tilde{L}^\top$, we have
\begin{equation}\label{eq:basis_ker_2pi}
   \kerpi{G^\top}/\impi{H^\top}= \frac{2\pi}{d_1}\Z_{d_1}\bw_{d_1}\oplus\ldots
   \oplus\frac{2\pi}{d_t}\Z_{d_t}\bw_{d_t}\oplus [0,2\pi)\bw_1 \oplus\ldots\oplus[0,2\pi)\bw_f.
\end{equation}
In the following, we call $\vartheta:=\upsilon^{-1}\circ\eta^*\circ h$
\begin{equation}\label{eq:vartheta}
   \vartheta:
   H^1(C_\bullet)\to\frac{2\pi}{d_1}\Z_{d_1}\times\ldots\times\frac{2\pi}{d_t}\Z_{d_t}\times [0,2\pi)^f.
\end{equation}

\paragraph{Fourier transform}
We define the Fourier transform between $L^2$ functions over $H_1(C_\bullet)$ and $H^1(C_\bullet)$ by carrying the Fourier transform 
between $\cA:=\frac{2\pi}{d_1}\Z_{d_1}\times\ldots\times\frac{2\pi}{d_t}\Z_{d_t}\times [0,2\pi)^f$ and
$\cZ:=\Z_{d_1} \times \ldots \times \Z_{d_t}\times \Z^f$ through the isomorphisms previously constructed in the 
following manner:
\begin{equation}
   \begin{tikzcd}
      {L^2(\cZ)} & {L^2(\cA)}\\
      {L^2(H_1(C_\bullet))} & {L^2(H^1(C_\bullet))}
      \arrow["\cF_{2\pi}", from=1-1, to=1-2]
      \arrow["\eta"', from=1-1, to=2-1]
      \arrow["\vartheta"', from=2-2, to=1-2]
   \end{tikzcd}
\end{equation}
To do so, we first need to define the measure space for each of these groups. The strategy is to 
define the usual measure space over $\cZ$ and $\cA$, and turn them into measure spaces for 
$H_1(C_\bullet)$ and $H^1(C_\bullet)$ via the isomorphisms $\eta$ (\cref{eq:eta}) and 
$\vartheta$ (\cref{eq:vartheta}), respectively.

On one hand, consider the usual topologies on $\Z$, $\Z_d$ and $[0,2\pi)$,
and let $\cB(\cZ)$ and $\cB(\cA)$ be the corresponding Borel
$\sigma$-algebras with the product topology on
\begin{equation}
   \cZ:=\Z_{d_1}\times\cdots\times\Z_{d_t}\times\Z^f,
   \qquad
   \cA:=\frac{2\pi}{d_1}\Z_{d_1}\times\cdots\times\frac{2\pi}{d_t}\Z_{d_t}\times[0,2\pi)^f.
\end{equation}
Let $\delta$ be the counting measure on $\cB(\Z)$, $\delta_d$ be the counting measure on $\cB(\Z_d)$ 
and $\lambda_{[0,2\pi)}$ be the Lebesgue measure on $\cB([0,2\pi))$. We equip $\cB(\cZ)$ and $\cB(\cA)$
with the product measures 
\begin{equation}
   \mu_{\cZ}:=\delta_{d_1}\otimes\cdots\otimes\delta_{d_t}\otimes\delta^{\otimes f},
   \quad \text{and} \quad 
   \mu_{\cA}
   :=
   \left(\frac1{d_1}\delta_{d_1}\right)\otimes\cdots\otimes
   \left(\frac1{d_t}\delta_{d_t}\right)\otimes
   \left(\frac{1}{2\pi}\lambda_{[0,2\pi)}\right)^{\otimes f}.
\end{equation}
With these measure spaces, the Fourier transform between 
$(\cZ,\cB(\cZ),\mu_{\cZ})$ and $(\cA,\cB(\cA),\mu_{\cA})$ is
\begin{equation} \label{eq:TF_ZT}
   \begin{aligned}
      \cF_{2\pi} : L^2\left(\cZ\right) & \to L^2\left(\cA\right) \\
      f&\mapsto \left(\bphi \mapsto \sum_{\bz\in \cZ} f(\bz)e^{-i\bphi\cdot\bz}\right)
   \end{aligned}
\end{equation}
with inverse 
\begin{equation}
   \begin{aligned}
      \cF_{2\pi}^{-1} : L^2\left(\cA\right) & \to L^2\left(\cZ\right) \\
      g&\mapsto \left(\bz \mapsto \int_{\cA} g(\bphi)e^{i\bphi\cdot\bz}d\mu_{\cA}(\bphi)\right).
   \end{aligned}
\end{equation}
These are isometries for the $L^2$ norm \cite{folland2015course}. 

On the other hand, define the two measure spaces:
\begin{itemize}
   \item $(H_1(C_\bullet),\eta(\cB(\cZ)),\mu_{H_1})$ where $\eta(\cB(\cZ))$ is the inverse image $\sigma$-algebra and $\mu_{H_1}=\mu_{\cZ}\circ\eta^{-1}$
   is the image measure, 
   \item $(H^1(C_\bullet),\vartheta^{-1}(\cB(\cA)),\mu_{H^1})$ where $\vartheta^{-1}(\cB(\cA))$ 
   is the inverse image $\sigma$-algebra and $\mu_{H^1}=\mu_{\cA}\circ\vartheta$ is the image measure, 
\end{itemize}
Note that this definition of the $\sigma$-algebras immediately makes $\eta$
and $\vartheta$ measurable. Defining 
\begin{equation}
   \begin{aligned}
      V_1 : L^2\left(H_1(C_\bullet)\right) & \to L^2\left(\cZ\right) \\
      f&\mapsto f \circ \eta^{-1}
   \end{aligned} \quad ; \quad
   \begin{aligned}
      V_2 : L^2\left(\cA\right) & \to L^2\left(H^1(C_\bullet)\right) \\
      f&\mapsto f\circ \vartheta
   \end{aligned}
\end{equation}
these two maps are isometries by the change of variable formula.
Composing these isometries with the Fourier transform in \cref{eq:TF_ZT}, 
we define a new Fourier transform.

\begin{definition}
   The Fourier transform $\cF$ between $L^2\left(H_1(C_\bullet)\right)$ and 
   $L^2\left(H^1(C_\bullet)\right)$ is defined as $\cF = V_2\circ\cF_{2\pi}\circ V_1$,
   namely,
   \begin{equation} \label{eq:TF_H}
      \begin{aligned}
         \cF : L^2\left(H_1(C_\bullet)\right) & \to L^2\left(H^1(C_\bullet)\right) \\
         f&\mapsto \left([\varphi] \mapsto \sum_{[\bn]\in H_1(C_\bullet)} f([\bn])
         \overline{\varphi(\bn)}\right).
      \end{aligned}
   \end{equation}
\end{definition}

Note that this definition does not depend on the representative
as for any $\varphi\in \ker{H^*}$ and $\bn\in\im{G}$, $\varphi(\bn)=1$.
Moreover, it inherits all the properties of the Fourier transform $\cF_{2\pi}$.
In particular, $\cF$ is an isometry for the $L^2$ norm and maps convolutions 
to multiplications.
The inverse $\cF^{-1}=V_1^{-1}\circ\cF_{2\pi}\circ V_2^{-1}$ can be written 
explicitly
\begin{equation} \label{eq:TF_H_inverse}
   \begin{aligned}
      \cF^{-1} : L^2\left(H^1(C_\bullet)\right) & \to L^2\left(H_1(C_\bullet)\right) \\
      g&\mapsto \left([\bn] \mapsto 
      \int_{H^1(C_\bullet)} g([\varphi])\varphi(\bn)d\mu_{H^1}([\varphi])\right)
   \end{aligned}
\end{equation}
based on the expression of $\cF_{2\pi}^{-1}$.

\subsection{Duality in the codespace}

The connection between the coherent states of the codespace and the Fourier transform 
can be cast as a commutative diagram
\begin{equation}\label{eq:diagram}
   \begin{tikzcd}
      {L^2\left(H_1(C_\bullet)\right)} & {L^2\left(H^1(C_\bullet)\right)} \\
      {\cT(G,H, \balpha, \bDelta)} & {\cT(G,H, \balpha, \bDelta)} 
      \arrow["\cF", from=1-1, to=1-2]
      \arrow["\pi"', from=1-1, to=2-1]
      \arrow["\rho", from=1-2, to=2-2]
      \arrow["\cN^{-1}"', from=2-1, to=2-2]
   \end{tikzcd}
\end{equation}
where
\begin{equation}\label{eq:rho}
   \left\{\begin{aligned}
      &\pi:f \mapsto \sum_{[\bn]\in K_H^{\bDelta}/\!\sim_G}
      f\circ\iota([\bn])\ket{\balpha_{[\bn]}}\\
      &\rho:g \mapsto  \int_{H^1(C_\bullet)} g([\varphi])\varphi(\n-\bn_0) 
      \Pi_{H} \ket{\balpha} d\mu_{H^1}([\varphi])
   \end{aligned}\right.
\end{equation}
and $\cN$ is a unbounded normalisation operator, diagonal in the 
$\{\ket{\balpha_{[\bn]}}:[\bn]\in K_H^{\bDelta}/\!\!\sim_G\}$
basis, defined as 
\begin{equation}
   \cN \ket{\balpha_{[\bn]}}:=
   \frac{1}{A_{[\bn]}}\ket{\balpha_{[\bn]}}
   \quad \text{and} \quad
   \cD(\cN):=\Span{\ket{\balpha_{[\bn]}}:[\bn]\in K_H^{\bDelta}/\!\!\sim_G}
\end{equation}
where we recall that $A_{[\bn]}=\|\Pi_{[\bn]}\ket{\balpha}\|$
with $\Pi_{[\bn]}=\sum_{\bn\in[\bn]}\ket{\bn}\bra{\bn}$.
It is an invertible operator with bounded inverse
\begin{equation}
   \cN^{-1} \ket{\balpha_{[\bn]}}=
   A_{[\bn]}\ket{\balpha_{[\bn]}}.
\end{equation}
and 
$\|\cN^{-1}\|=\sup_{[\bn]\in K_H^{\bDelta}/\!\!\sim_G} A_{[\bn]}\leq1$

The integral in $\rho$ is defined as a Bochner integral.
Indeed, the function \break
$[\varphi]\mapsto g([\varphi])\varphi(\n-\bn_0)\Pi_{H}\ket{\balpha}$
is Bochner-integrable as
\begin{equation}
   \begin{aligned}
      \int_{H^1(C_\bullet)} \|g([\varphi])\varphi(\n-\bn_0) \Pi_{H} \ket{\balpha}\| d\mu_{H^1}([\varphi])
      &\leq \int_{H^1(C_\bullet)} |g([\varphi])|d\mu_{H^1}([\varphi])\\
      &\leq \|g\|_2 \mu_{H^1}(H^1(C_\bullet))^{1/2}
   \end{aligned}
\end{equation}
where the last inequality is a Cauchy-Schwarz inequality.
This proves that $\cD(\rho)=L^2(H^1(C_\bullet))$ and hence $\rho$ is a bounded operator.
Furthermore, the integral does not depend on the choice of 
representative $\varphi$ of $[\varphi]$ as for any $\bn\in\N^m$,
\begin{equation}
   \varphi(\n-\bn_0)\Pi_{H}\ket{\bn}=\left\{
   \begin{aligned}
      &\varphi(\bn-\bn_0)\ket{\bn} &\text{if } H\bn=\bDelta\\
      &0&\text{otherwise}
   \end{aligned}\right.
\end{equation} 
and $\bn-\bn_0\in\ker{H}$ if $H\bn=\bDelta$.

To prove the commutativity of the diagram in \cref{eq:diagram}, we use 
the following lemma, which follows from the fact that the sum of 
the roots of unity is zero.

\begin{lemma}\label{lemma:sumofroots}
   For all $[\bn]\in H_1(C_\bullet)$,
   \begin{equation}
      \cF^{-1}([\varphi]\mapsto \varphi(0))([\bn])
      =\int_{H^1(C_\bullet)} \varphi(\bn)d\mu_{H^1}([\varphi])
      =\indicator{\{[\bn]=[0]\}}.
   \end{equation}
\end{lemma}

\begin{proof}
   Denoting $\bgamma=\eta^{-1}([\bn])=(\gamma_{d_1},\ldots,\gamma_{d_t},\gamma_1,\ldots,\gamma_f)$, we have 
   \begin{equation}
      \begin{aligned}
         \int_{H^1(C_\bullet)} \varphi(\bn)d\mu_{H^1}([\varphi]) 
         &= \int_{\cA} e^{i\bphi\cdot\bgamma}d\mu_{\cA}(\bphi) \\
         &= \frac{1}{(2\pi)^fd_1\ldots d_t}\int_{[0,2\pi)^f}e^{i(\phi_1\gamma_1+\ldots\phi_f\gamma_f)}d\phi_1\ldots d\phi_f \\
         &\qquad\qquad\times\sum_{k_1=1}^{d_1} e^{\frac{2i\pi}{d_1}k_1\gamma_{d_1}} \ldots \sum_{k_t=1}^{d_t} e^{\frac{2i\pi}{d_t}k_t\gamma_{d_t}} \nonumber \\
         &= \indicator{\{\gamma_1=0\}}\ldots\indicator{\{\gamma_f=0\}}\indicator{\{\gamma_{d_1}=0\}} \ldots \indicator{\{\gamma_{d_t}=0\}} \nonumber \\
         &=\indicator{\{\bgamma=0\}}=\indicator{\{[\bn]=[0]\}}. 
      \end{aligned}
   \end{equation}
\end{proof}

\begin{theorem}\label{theorem:diagram}
   The diagram in \cref{eq:diagram} commutes.
\end{theorem}

\begin{proof}
   We prove that $\rho\cF=\cN^{-1}\pi$ on the canonical basis
   $\{\indicator{\{[\bl]\}}:[\bl]\in H_1(C_\bullet)\}$ of
   $L^2(H_1(C_\bullet))$.
   For $[\varphi]\in H^1(C_\bullet)$,
   \begin{equation}
      \cF(\indicator{\{[\bl]\}})([\varphi])
      =\sum_{[\bn]\in H_1(C_\bullet)} \indicator{\{[\bl]\}}([\bn])\,\overline{\varphi(\bn)}
      =\overline{\varphi(\bl)}
      =\varphi(-\bl).
   \end{equation}
   Applying $\rho$ gives
   \begin{equation}
      \begin{aligned}
         \rho\left([\varphi] \mapsto \varphi(-\bl)\right) 
         &= \int_{H^1(C_\bullet)}\varphi(-\bl)\varphi(\n-\bn_0)\Pi_{H} 
         \ket{\balpha}d\mu_{H^1}([\varphi]) \\
         &= e^{-|\balpha|^2/2}\sum_{\substack{
            \bn \in \N^m \\
            H\bn=\bDelta
         }} \frac{\balpha^{\bn}}{\sqrt{\bn!}} \int_{H^1(C_\bullet)}\varphi(\bn-\bn_0-\bl)d\mu_{H^1}([\varphi]) \ket{\bn} \\
         &= e^{-|\balpha|^2/2}\sum_{\substack{
            \bn \in \N^m \\
            H\bn=\bDelta
         }} \frac{\balpha^{\bn}}{\sqrt{\bn!}} 
         \indicator{\{\bn \in [\bl+\bn_0]\}} \ket{\bn}
      \end{aligned}
   \end{equation}
   where the last equality follows from \cref{lemma:sumofroots}.
   Note that $\indicator{\{\bn \in [\bl+\bn_0]\}}$ is always 
   zero if \break $[\bl+\bn_0]\cap\N^m=\varnothing$, which is equivalent 
   to $[\bl]\notin\iota(K_H^{\bDelta}/\!\!\sim_G)$ where $\iota$ is defined 
   in \cref{eq:iota}. Therefore,
   \begin{equation}
      \rho\cF\left(\indicator{\{[\bl]\}}\right) 
      = \left\{\begin{aligned}
         &0 \quad &\text{if } [\bl]\notin\iota(K_H^{\bDelta}/\!\!\sim_G), \\
         &A_{\iota^{-1}([\bl])}
         \ket{\balpha_{\iota^{-1}([\bl])}} \quad &\text{otherwise.}
      \end{aligned}\right.
   \end{equation}
   On the other hand,
   \begin{equation}
      \pi(\indicator{\{[\bl]\}})=
      \begin{cases}
         0 & [\bl]\notin \iota(K_H^{\bDelta}/\!\!\sim_G)\\
         \ket{\balpha_{\iota^{-1}([\bl])}}
         & [\bl]\in \iota(K_H^{\bDelta}/\!\!\sim_G),
      \end{cases}
   \end{equation}
   so applying $\cN^{-1}$ yields the same expression.
   Since $\pi$, $\rho$, $\cF$ and $\cN^{-1}$ are bounded operators, the
   equality extends by linearity and continuity to all of
   $L^2(H_1(C_\bullet))$.
\end{proof}

This theorem allows to deduce the main result mentioned in the introduction.
First, we can prove that, up to a constant, the states in the integral 
in the definition of $\rho$ (\cref{eq:rho}) correspond 
to projected coherent states, i.e.
\begin{equation}
   \{\varphi(\n-\bn_0)\Pi_H\ket{\balpha}: \varphi\in H^1(C_\bullet)\}
   \subset
   \{e^{i\theta}\Pi_H\ket{\bbeta}:\theta\in[0,2\pi),\forall \bg\in\im{G},
   \bbeta^{\bg^+}=\balpha^{\bg}\bbeta^{\bg^-}\}.
\end{equation}
Indeed, using the isomorphism $\tau$ introduced in \cref{eq:tau}, we can 
find $\bphi\in\kerpi{G^\top}/\impi{H^\top}$ such that 
$[\varphi]=[x\mapsto e^{i\bphi\cdot\bx}]$. So, as 
$[\varphi(\n-\bn_0),\Pi_H]=0$,
\begin{equation}
   \varphi(\n-\bn_0)\Pi_H\ket{\balpha}=
   e^{-i\bphi\cdot\bn_0}\Pi_H\ket{e^{i\phi_1}\alpha_1,\ldots,e^{i\phi_m}\alpha_m}.
\end{equation}
As for all $\bg\in\im{G}, e^{i\bphi\cdot\bg}=1$, 
the coherent state $\ket{e^{i\phi_1}\alpha_1,\ldots,e^{i\phi_m}\alpha_m}$ satisfies 
\begin{equation}
   e^{i\bphi\cdot\bg^+}\balpha^{\bg^+}-\balpha^{\bg}e^{i\bphi\cdot\bg^-}\balpha^{\bg^-}
   =e^{i\bphi\cdot\bg^+}(\balpha^{\bg^+}-\balpha^{\bg}\balpha^{\bg^-})=0.
\end{equation}
Moreover, as $\pi$, $\cF^{-1}$ and $\cN^{-1}$
are surjective (\cref{theorem:pi}), so is $\rho=\cN^{-1}\pi\cF^{-1}$.
So, any vector in $\cT(G,H,\balpha,\bDelta)$ can be written as a 
Bochner integral over the states \break
$\{\varphi(\n-\bn_0)\Pi_{H}\ket{\balpha}: \varphi\in H^1(C_\bullet)\}$.
Because for all $g\in L^2(H^1(C_\bullet))$, 
\begin{equation}
   \int_{H^1(C_\bullet)} g([\varphi])\varphi(\n-\bn_0) 
      \Pi_{H} \ket{\balpha} d\mu_{H^1}([\varphi])
   \in\overline{\Span{\varphi(\n-\bn_0)\Pi_{H}\ket{\balpha}: \varphi\in H^1(C_\bullet)}}
\end{equation}
This implies that 
$\{\varphi(\n-\bn_0)\Pi_{H}\ket{\balpha}: \varphi\in H^1(C_\bullet)\}$
is a total set in $\cT(G,H,\balpha,\bDelta)$, thus proving 
\begin{equation}
   \cT(G,H, \balpha, \bDelta)=\overline{t(G,H,\balpha,\bDelta)}
\end{equation}
where $t(G,H,\balpha,\bDelta):=\overline{\Span{\varphi(\n-\bn_0)\Pi_{H}\ket{\balpha}: \varphi\in H^1(C_\bullet)}}$ 
is defined in \cref{eq:original_def}. 
This implies that the following set is also total in the codespace
\begin{equation}
   \Span{\Pi_{H}\ket{\bbeta}:\forall \bg\in\im{G},
   \bbeta^{\bg^+}-\balpha^{\bg}\bbeta^{\bg^-}=0}
\end{equation}
since
\begin{equation}
   \overline{t(G,H,\balpha,\bDelta)}\subset
   \overline{\Span{\Pi_{H}\ket{\bbeta}:\forall \bg\in\im{G},
   \bbeta^{\bg^+}-\balpha^{\bg}\bbeta^{\bg^-}=0}}
   \subset \cT(G,H, \balpha, \bDelta).
\end{equation}

Hence, for Tiger codes, projected coherent states belonging to 
the codespace are enough to fully characterise the code.
Besides, this draws a connection between the algebraic variety
induced by the polynomial ideal 
$\{\bx^{\bg^+}-\balpha^{\bg}\bx^{\bg^-}=0 : \bg\in\im{G}\}$ 
and the codespace.

Importantly, note that if at least one of the components of $\balpha$
is zero, the coherent states in the code are no longer enough to characterise the 
codespace. For instance, consider the cat code with $\balpha=0$ defined 
as $\cT(0,(2),0,0)=\ker{\hat{a}^2}=\Span{\ket{0},\ket{1}}$.
Then, the only coherent state in the code is $\ket{0}$ and 
$\overline{\Span{\ket{0}}}\neq\cT(0,(2),0,0)$.

\subsection{X and Z bases}

\Cref{theorem:diagram} motivates the introduction of two dual bases of the code.
Following \cite{xu2024lettingtigercagebosonic} conventions, we introduce $X$ and $Z$ bases.

\begin{definition}
   The $X$ and $Z$ bases of the Tiger code $\cT(G,H, \balpha, \bDelta)$
   are 
   \begin{equation}
      \begin{aligned}
         &\mathcal{B}_X:=\left\{\ket{\balpha_{[\bn]}}:=
         \frac{1}{A_{[\bn]}} \Pi_{[\bn]}\ket{\balpha}: [\bn]\in K_H^{\bDelta}/\!\!\sim_G\right\}\\
         &\mathcal{B}_Z:=\left\{\ket{\omega_{[\varphi]}}:=\varphi(\n-\bn_0)\Pi_{H} \ket{\balpha} : 
         [\varphi]\in H^1(C_\bullet)\right\},
      \end{aligned}
   \end{equation}
   respectively.
\end{definition}

To be precise, while the $X$ basis is an orthonormal basis of 
the code, the $Z$ basis is only a total set in the code, 
consisting of unnormalised vectors. Uniqueness of the 
decomposition in the $Z$ basis is equivalent to requiring 
$\rho=\cN^{-1}\pi\cF^{-1}$ to be injective. 
By \cref{theorem:pi,theorem:diagram}, this holds if and 
only if $\ker\pi=0$.

The two bases are linked by the Fourier transform
\begin{equation}\label{eq:XZ_bases_relations}
   \left\{\begin{aligned}
      &\ket{\balpha_{[\bn]}}=\frac{1}{A_{[\bn]}}
      \int_{H^1(C_\bullet)}\overline{\varphi(\bn-\bn_0)}\ket{\omega_{[\varphi]}}d\mu_{H^1}(\varphi) \\
      &\ket{\omega_{[\varphi]}}=
      \sum_{[\bn]\in K_H^{\bDelta}/\!\sim_G} \varphi(\bn-\bn_0)A_{[\bn]}\ket{\balpha_{[\bn]}}.
   \end{aligned}\right.
\end{equation}
The first equality corresponds to 
$\rho\cF\left(\indicator{\{\iota([\bn])\}}\right)=\cN^{-1}\pi\left(\indicator{\{\iota([\bn])\}}\right)$,
which results from \cref{theorem:diagram}.
The second is more subtle to obtain due to the fact that $\cN$ is 
an unbounded operator.
Indeed,
\begin{equation}
   \sum_{[\bn]\in K_H^{\bDelta}/\!\sim_G}
   \varphi(\iota([\bn]))A_{[\bn]}\ket{\balpha_{[\bn]}}
\end{equation}
does not belong to the domain of $\cN$ since 
$\sum_{[\bn]\in K_H^{\bDelta}/\!\sim_G}|\varphi(\iota([\bn]))|^2=\infty$ 
whenever the codespace is infinite-dimensional.
Nevertheless, the second equality can be understood 
formally as follows: one applies a generalised Fourier 
transform to the non-square-integrable function
\begin{equation}
   [\bn]\mapsto \varphi(\bn)\notin L^2(H_1(C_\bullet)),
\end{equation}
obtaining a Dirac distribution; applying $\rho$ to this 
distribution then gives $\ket{\omega_{[\varphi]}}$.
Formalising this argument would require extending the 
Fourier transform considered in this work to the setting 
of tempered distributions, and proving an analogue of the 
commutative diagram in \cref{eq:diagram} in that setting. 
For simplicity, we instead give a direct proof:
\begin{equation}
   \begin{aligned}
      \sum_{[\bn]\in K_H^{\bDelta}/\!\sim_G} \varphi(\bn-\bn_0)A_{[\bn]}\ket{\balpha_{[\bn]}}
      &=e^{-\frac{|\balpha|^2}{2}}\sum_{[\bn]\in K_H^{\bDelta}/\!\sim_G} \varphi(\bn-\bn_0)
      \sum_{\bk\in[\bn]}\frac{\balpha^{\bk}}{\sqrt{\bk!}}\ket{\bk}\\
      &=e^{-\frac{|\balpha|^2}{2}}\sum_{[\bn]\in K_H^{\bDelta}/\!\sim_G} \sum_{\bk\in[n]}
      \varphi(\bk-\bn_0)\frac{\balpha^{\bk}}{\sqrt{\bk!}}\ket{\bk}\\
      &=e^{-\frac{|\balpha|^2}{2}}\sum_{\substack{
            \bk\in\N^m \\
            H\bk=\bDelta 
         }}\varphi(\bk-\bn_0)\frac{\balpha^{\bk}}{\sqrt{\bk!}}\ket{\bk}\\
      &=\ket{\omega_{[\varphi]}}
   \end{aligned}
\end{equation}
where the second equality comes from the fact that 
$\forall \bg\in\im{G},\varphi(\bg)=1$ since $\varphi\in\ker{G^*}$.

\section{Non-linear constraints}\label{sec:non_linear}

This section introduces a generalisation of Tiger codes. 
One limitation of the algebraic structure considered so 
far is that the $\n$-type constraints do not affect the 
qudit degrees of freedom of the code. More precisely, 
when the map $\pi$ of \cref{theorem:pi} is bijective, the 
Tiger codes $\cT(G,H,\balpha,\bDelta)$ and $\cT(G,0,\balpha,0)$ 
host the same qudits. Algebraically, this reflects the fact 
that $\ker H/\im G$ and $\Z^m/\im G$ have the same torsion 
subgroup; see the proof of \cref{theorem:H1isom}.

However, several bosonic codes in the literature use the qudits 
hosted by an underlying Tiger code to encode a single, 
better-protected qudit. This is the case for the 
\emph{repetition cat code} \cite{Guillaud_2019}: its $\a$-type 
constraints,
$\{\ha_i^2-\alpha_i^2:1\leq i\leq m\}$,
are Tiger-code constraints, whereas its $\n$-type constraints,
$\{(-1)^{\hn_i+\hn_{i+1}}-1:1\leq i\leq m-1\}$,
are not. Similarly, the \emph{4-legged cat code} 
\cite{Mirrahimi_2014} has the $\a$-type constraint 
$\ha^4-\alpha^4$ together with the non-Tiger $\n$-type 
constraint $(-1)^{\hn}-1$. 
This is also the case of the two-mode bosonic Fourier code~\cite{Lev26}
with $\a$-type constraints $\ha_1^4-\alpha^4$ and $\ha_1^2+\ha_2^2$ and 
a non-linear $\n$-type constraint $1+(-1)^{\hn_1+\hn_2}$.
The generalisation introduced
below captures all these examples by allowing nonlinear $\n$-type 
constraints.

Consider the Tiger code
\begin{equation}
   \cT(G,0,\balpha,0)
   =\bigcap_{\bg\in\im{G}}\ker{\a^{\bg^+} - \balpha^{\bg}\a^{\bg^-}} .
\end{equation}

\begin{definition}
   Given $G\in\Z^{m \times r}$ and $\balpha\in\C^m$ such that $\balpha$
   has no zero entry, consider a finite family of kernel $G$-periodic functions 
   $\{h_\gamma:\Z^m\to\C\}_{1\leq\gamma\leq s}$, meaning that
   \begin{equation}\label{eq:condition_h_gamma}
      \forall \bn\in\Z^m,\forall \bg\in\im{G}, h_\gamma(\bn)=0
      \Rightarrow h_\gamma(\bn+\bg)=0.
   \end{equation}
   An extended Tiger code is defined as 
   \begin{equation}
      \cT\left(G, \balpha, \{h_\gamma\}_{1\leq\gamma\leq s}\right)
      :=\bigcap_{\bg\in\im{G}}\ker{\a^{\bg^+} - \balpha^{\bg}\a^{\bg^-}} 
      \cap\bigcap_{\gamma\in\llbracket1,s\rrbracket} \ker{h_\gamma(\n)}.
   \end{equation}
\end{definition}

To recover the usual definition of Tiger codes (\cref{def:TC}), we take \break
$\{h_\gamma\}=\{\bn\mapsto \bh_i\cdot\bn-\Delta_i\}_{1\leq i\leq s}$
where $\bh_i$ and $\Delta_i$ are the $i$-th row of $H$ and $\bDelta$, respectively.
Note that the intersection between
\begin{equation}
   \bigcap_{g\in\im G}\ker{\a^{g^+}-\balpha^g\a^{g^-}}
   \quad\text{and}\quad
   \bigcap_{\gamma\in\llbracket1,s\rrbracket} \ker{h_\gamma(\n)}
\end{equation}
need not be trivial even when the functions $h_\gamma$ 
are not kernel $G$-periodic. The exact condition for 
this intersection to be nontrivial is that there exist 
$\bn\in\N^m$ such that
\begin{equation}
\bn+\im G \subset 
\bigcap_{\gamma\in\llbracket1,s\rrbracket} \ker{h_\gamma(\n)}.
\end{equation}
Equivalently, there must exist a non-negative integer vector 
whose entire equivalence class modulo $\im G$ lies in the 
common zero set of the functions $h_\gamma$.
For simplicity, we impose the stronger condition that the 
$h_\gamma$ are kernel $G$-periodic. The following example 
illustrates how to write the Fourier code \cite{Lev26} 
as an extended Tiger code.

\begin{example}
   The four-dimensional two-mode bosonic Fourier code of~\cite{Lev26} is the
   extended Tiger code $\mathcal T(G_{\mathrm F},(\alpha,i\alpha)^\top,\{h_{\mathrm F}\})$, where
   $G_{\mathrm F}=\bigl(\begin{smallmatrix}4&2\\0&-2\end{smallmatrix}\bigr)$ and
   $h_{\mathrm F}:\bn\mapsto1+(-1)^{n_1+n_2}$. Its $\a$-type constraints are generated by
   $\ha_1^4-\alpha^4$ and $\ha_1^2+\hat a_2^2$, 
   while $h_{\mathrm F}$ is kernel $G_F$-periodic and
   selects the odd-total-parity sector.
\end{example}

The remainder of this section is dedicated to proving that 
all the properties of Tiger codes generalise to their 
extended versions.

\paragraph{Finite generation.}
First, according to \cref{prop:finite_generation}, the intersection \break
$\bigcap_{\bg\in\im{G}}\ker{\a^{\bg^+} - \balpha^{\bg}\a^{\bg^-}}$ 
can be reduced to a finite intersection. So, likewise to Tiger codes, as 
$\{h_\gamma\}_{1\leq\gamma\leq s}$ is a finite family, extended Tiger codes are 
generated by a finite number of constraints.

\paragraph{Fourier transform and coherent states.}
Here, we aim at generalising \cref{theorem:diagram}.
To do so, we remark that any extended Tiger code can be 
written as a Tiger code with no $\n$-type constraints 
projected on a space satisfying the non-linear $\n$-type constraints, namely
\begin{equation}\label{eq:ext_and_vanilla}
   \cT\left(G, \balpha, \{h_\gamma\}\right)
   =\cT(G,0, \balpha,0)
   \cap\bigcap_{\gamma\in\llbracket1,s\rrbracket} \ker{h_\gamma(\n)}.
\end{equation}
Because \cref{theorem:diagram} applies to $\cT(G,0, \balpha,0)$, 
it suffices to describe how the diagram in \cref{eq:diagram} 
is transformed when taking into account the extra non-linear constraints.
To streamline the notations, we introduce 
\begin{equation}
   K_{\gamma}:=\bigcap_{\gamma\in\llbracket1,s\rrbracket} \ker{h_\gamma}
   \subset \Z^m.
\end{equation}

First, when $H=0$, note that we can take $\bn_0=0$. 
Then, since all the $h_\gamma$ are kernel $G$-periodic,
the indicator function 
$\indicator{K_\gamma}:\Z^m\to\{0,1\}$ 
descends to the equivalence classes 
\begin{equation}
   \indicator{K_\gamma}([\bn])
   :=\indicator{K_\gamma}(\bn),
   \qquad \forall [n]\in\Z^m.
\end{equation}
Finally, we define the analogue of the projector $\Pi_H$ for 
non-linear constraints that is
\begin{equation}
   \Pi_{\{h_\gamma\}}:=
   \sum_{\bn\in\cap_{\gamma=1}^s \ker{h_\gamma}\cap\N^m}
   \ket{\bn}\bra{\bn}=\indicator{K_\gamma}(\n).
\end{equation}
The following diagram is the appropriate analogue of \cref{eq:diagram}
in the extended setting:
\begin{equation}\label{eq:diagram_ext}
   \begin{tikzcd}
      {\{\indicator{K_\gamma}\times f:f\in L^2\left(\Z^m/\im{G}\right)\}} 
      & {\{C(g): g\in L^2\left(\ker{G^*}\right)\}} \\
      {\cT\left(G, \balpha, \{h_\gamma\}\right)} & {\cT\left(G, \balpha, \{h_\gamma\}\right)} 
      \arrow["\cF", from=1-1, to=1-2]
      \arrow["\pi"', from=1-1, to=2-1]
      \arrow["\tilde{\rho}", from=1-2, to=2-2]
      \arrow["\cN^{-1}"', from=2-1, to=2-2]
   \end{tikzcd}
\end{equation}
where, 
\begin{equation}\label{eq:ext_maps}
   \left\{\begin{aligned}
      &\pi:f \mapsto \sum_{[\bn]\in \N^m/\!\sim_G}
      f\circ\iota([\bn])\ket{\balpha_{[\bn]}}\\
      &\tilde{\rho}:g \mapsto  \int_{\ker{G^*}} g(\varphi)
      \Pi_{\{h_\gamma\}}\varphi(\n) \ket{\balpha} 
      d\mu_{H^1}(\varphi).\\
      &\begin{aligned}
         C:L^2\left(\ker{G^*}\right)&\to L^2\left(\ker{G^*}\right)\\
         g&\mapsto \cF\left(\indicator{K_\gamma}
         \times\cF^{-1}(g)\right).
      \end{aligned}
   \end{aligned}\right.
\end{equation}
Before proving that the diagram in \cref{eq:diagram_ext} commutes, 
we record two useful observations about this construction 
and its relation to the previous diagram in \cref{eq:diagram}.
First, the maps $\tilde{\rho}$ and $C$ admit the 
following concrete descriptions. One has 
\begin{equation}
   \tilde{\rho}=\Pi_{\{h_\gamma\}}\rho,
\end{equation}
where $\rho$ is the map defined in \cref{eq:rho} in the case $H=0$.
Moreover, if
$\indicator{K_\gamma}\in L^2\left(\Z^m/\im G\right)$,
then $C$ is a convolution operator \cite{folland2015course}:
\begin{equation}
   C:g\longmapsto \cF\left(\indicator{K_\gamma}\right)*g .
\end{equation}
Second, in the case of linear constraints---namely
$\{h_\gamma\}=\{n\mapsto h_i(\bn)-\Delta_i\}_{1\leq i\leq s}$---the 
upper-left space in \cref{eq:diagram_ext} identifies 
to that of \cref{eq:diagram} since 
\begin{equation}
   \{\indicator{K_\gamma}\!\times\! f:f\in L^2\left(\Z^m/\im G\right)\}
   \cong L^2(H_1(C_\bullet)).
\end{equation}

\begin{theorem}\label{theorem:ext_diagram}
   The diagram in \cref{eq:diagram_ext} commutes.   
\end{theorem}

\begin{proof}
   First, as pointed out in \cref{eq:ext_and_vanilla}, 
   we have
   \begin{equation}
      \cT\left(G, \balpha, \{h_\gamma\}\right)
      =\Pi_{\{h_\gamma\}}
      \cT(G,0,\balpha,0).
   \end{equation}
   Then, define the bounded multiplication operator
   \begin{equation}
      \begin{aligned}
         M : L^2(\Z^m/\im{G})&\to L^2(\Z^m/\im{G}) \\
         f&\mapsto \indicator{K_\gamma}\times f
      \end{aligned}
   \end{equation}
   which is the orthogonal projection onto 
   $\{\indicator{K_\gamma}\times f:f\in L^2\left(\Z^m/\im{G}\right)\}$.
   Remark that, as defined in \cref{eq:ext_maps},
   $C=\cF M\cF^{-1}$.
   Hence, because $M$ is an orthogonal projector and 
   that $\cF$ is an isometry, $C$ is also an orthogonal projector.

   We aim at proving two points: 1) the codomain of the maps in \cref{eq:diagram_ext}
   are correct and, 2) the diagram commutes.
   Both points can be proven using the commutativity of the following diagram
   \begin{equation}
      \begin{tikzcd}
         {L^2(\Z^m/\im{G})} & & & {L^2(\ker{G^*})} \\
         & {L^2(\Z^m/\im{G})} & {L^2(\ker{G^*})} & \\
         & {\cT(G,0,\balpha,0)} & {\cT(G,0,\balpha,0)} & \\ 
         {\cT(G,0,\balpha,0)} & & & {\cT(G,0,\balpha,0)}\\ 
         \arrow["\cF", from=1-1, to=1-4]
         \arrow["\cF", from=2-2, to=2-3]
         \arrow["\pi"', from=1-1, to=4-1]
         \arrow["\pi"', from=2-2, to=3-2]
         \arrow["\rho", from=1-4, to=4-4]
         \arrow["\rho", from=2-3, to=3-3]
         \arrow["\cN^{-1}"', from=3-2, to=3-3]
         \arrow["\cN^{-1}"', from=4-1, to=4-4]
         \arrow["M"', from=1-1, to=2-2]
         \arrow["C", from=1-4, to=2-3]
         \arrow["\Pi_{\{h_\gamma\}}", from=4-1, to=3-2]
         \arrow["\Pi_{\{h_\gamma\}}"', from=4-4, to=3-3]
      \end{tikzcd}
   \end{equation}
   First, assume that this diagram commutes. 
   Both domains of $\pi$ and $\tilde{\rho}$ in 
   \cref{eq:diagram_ext} can be expressed as the image 
   of $M$ and $C$, respectively:
   \begin{equation}\label{eq:domain_as_image}
      \left\{\begin{aligned}
         &\im{M}
         =\{\indicator{K_\gamma}\times f:f\in L^2\left(\Z^m/\im{G}\right)\}, \\
         &\im{C} 
         =\{C(g): g\in L^2\left(\ker{G^*}\right)\}.
      \end{aligned}\right.
   \end{equation}
   Furthermore, the commutation relation
   $\pi M=\Pi_{\{h_\gamma\}}\pi$
   proves that 
   \begin{equation}
      \im{\pi M}\subset
      \Pi_{\{h_\gamma\}}\cT(G,0,\balpha,0)=\cT(G,\balpha,\{h_\gamma\}).
   \end{equation}
   Similarly, as 
   $\rho C =\Pi_{\{h_\gamma\}}\rho$, we get 
   $\im{\rho C}\subset\cT(G, \balpha, \{h_\gamma\})$. 
   So, combining these inclusions with the descriptions 
   of the domains in \cref{eq:domain_as_image}, this proves that 
   both codomains of $\pi$ and $\tilde{\rho}$ in \cref{eq:diagram_ext}
   are correct.
   Finally, 
   $\cF M=C \cF$
   and 
   $\cN^{-1}\Pi_{\{h_\gamma\}}=\Pi_{\{h_\gamma\}}\cN^{-1}$
   guarantee in the same manner that the codomains
   of $\cF$ and $\cN^{-1}$ in \cref{eq:diagram_ext} are correct.

   Point 2) results from a diagram chasing:
   \begin{equation}
      \tilde{\rho}\cF M
      =\Pi_{\{h_\gamma\}}\rho\cF M
      =\Pi_{\{h_\gamma\}}\cN^{-1}\pi M
      =\cN^{-1}\Pi_{\{h_\gamma\}}\pi M
      =\cN^{-1}\pi MM
      =\cN^{-1}\pi M
   \end{equation}
   where we used the fact that $M$ is 
   a projector.
   Hence, for any  
   $f\in\{\indicator{K_\gamma}\times f:f\in L^2\left(\Z^m/\im{G}\right)\}$
   we get 
   \begin{equation}
      \tilde{\rho}\cF f=\tilde{\rho}\cF Mf = \cN^{-1}\pi Mf=\cN^{-1}\pi f.
   \end{equation}

   Let us now prove the commutativity of the diagram in \cref{eq:diagram_ext}.
   The intuition is the following: the inner square has to be understood as the projection 
   of the outer square onto the space satisfying the 
   non-linear $\n$-type constraints.
   
   First, remark that the inner and outer squares both correspond to the diagram in \cref{eq:diagram} 
   in the particular case where $H=0$. According to \cref{theorem:diagram}, they both commute.
   It remains to prove that 
   \begin{equation}
      \left\{\begin{aligned}
         &C \cF
         =\cF M \\
         &\Pi_{\{h_\gamma\}}\cN^{-1} 
         =\cN^{-1} \Pi_{\{h_\gamma\}} \\
         &\pi M 
         = \Pi_{\{h_\gamma\}} \pi \\
         &\rho C  
         =\Pi_{\{h_\gamma\}}\rho.
      \end{aligned}\right.
   \end{equation}
   The first relation is clear since $C=\cF M\cF^{-1}$. For the second, 
   consider any basis vector \break 
   $\ket{\balpha_{[\bn]}}\in \cT(G,0,\balpha,0)$ 
   for some $[\bn]\in\N^m/\!\!\sim_G$.
   Because the $h_\gamma$'s are kernel $G$-periodic, \break
   $\Pi_{\{h_\gamma\}}\ket{\balpha_{[\bn]}}=\indicator{K_\gamma}(\bn)\ket{\balpha_{[\bn]}}$.
   Hence,
   \begin{equation}
      \cN^{-1}\Pi_{\{h_\gamma\}}\ket{\balpha_{[\bn]}}
      =\indicator{K_\gamma}(\bn)A_{[\bn]}\ket{\balpha_{[\bn]}}
      =\Pi_{\{h_\gamma\}}\cN^{-1}\ket{\balpha_{[\bn]}}.
   \end{equation}
   By linearity and continuity, we get the second commutation relation.
   For the third one, consider an element $\indicator{\{[\bl]\}}$ of the 
   orthonormal basis $\{\indicator{\{[\bn]\}}:[\bn]\in \Z^m/\im{G}\}$ of $L^2(\Z^m/\im{G})$.
   Then,
   \begin{equation}
      \pi M\indicator{\{[\bl]\}}
      =\pi \left(\indicator{K_\gamma}\indicator{\{[\bl]\}}\right)
      =\sum_{[\bn]\in\N^m/\!\sim_G} 
      \indicator{K_\gamma}(\iota([\bn]))
      \indicator{\{[\bl]\}}(\iota([\bn]))\ket{\balpha_{[\bn]}}.
   \end{equation}
   If $[\bl]\notin\im{\iota}$, $\indicator{\{[\bl]\}}([\bn])=0$
   for all $[\bn]\in\N^m/\!\!\sim_G$ and $\pi\indicator{\{[\bl]\}}=0$. Otherwise, 
   as $\iota$ is injective by \cref{prop:iota},
   \begin{equation}
      \pi M\indicator{\{[\bl]\}}
      =\indicator{K_\gamma}([\bl])\ket{\balpha_{\iota^{-1}([\bl])}}
      =\Pi_{\{h_\gamma\}}\pi\indicator{\{[\bl]\}}.
   \end{equation}
   Similarly, we conclude by linearity and continuity to obtain 
   the third commutation relation.
   Finally, using all the commutation relation we have proven so far,
   we deduce the last one:
   \begin{equation}
      \begin{aligned}
         \rho C  
         &=\rho \cF M \cF^{-1}\\
         &=\cN^{-1}\pi M \cF^{-1}\\
         &=\cN^{-1}\Pi_{\{h_\gamma\}}\pi \cF^{-1}\\
         &=\Pi_{\{h_\gamma\}}\cN^{-1}\pi \cF^{-1}\\
         &=\Pi_{\{h_\gamma\}}\rho.
      \end{aligned}
   \end{equation}
\end{proof}

\paragraph{$X$ and $Z$ bases.}
Again, we can define an $X$ and $Z$ basis in the same manner as before.

\begin{definition}
   Let the $X$ and $Z$ bases be
   \begin{equation}
      \left\{\begin{aligned}
         &\cB_X:=\left\{\ket{\balpha_{[\bn]}}=\frac{\Pi_{[\bn]}\ket{\balpha}}{A_{[\bn]}} 
         : [\bn]\in K_\gamma\cap\N^m/\im{G}\right\}\\
         &\cB_Z:=\left\{\ket{\omega_{\varphi}}=\varphi(\n) \Pi_{\{h_\gamma\}}\ket{\balpha} : \varphi\in \ker{G^*}\right\}
      \end{aligned}\right.
   \end{equation}
   respectively.
\end{definition}

As previously, while the $X$ basis is an orthonormal basis of $\cT\left(G, \balpha, \{h_\gamma\}\right)$,
the $Z$ basis is only a total set of $\cT\left(G,\balpha,\{h_\gamma\}\right)$.
The two bases are linked by the Fourier transform
\begin{equation}
   \left\{\begin{aligned}
      &\ket{\balpha_{[\bn]}}=\frac{1}{A_{[\bn]}}
      \int_{\ker{G^*}}\!\overline{\varphi(\bn)}\ket{\omega_{\varphi}}d\mu_{H^1}(\varphi) \\
      &\ket{\omega_{\varphi}}=
      \sum_{[\bn]\in K_\gamma\cap\N^m\!/\!\sim_G} \varphi(\bn)A_{[\bn]}\ket{\balpha_{[\bn]}}
   \end{aligned}\right.
\end{equation}
where both equalities are obtained by considering the equalities 
in \cref{eq:XZ_bases_relations} in the case $H=0$ and then 
projecting with $\Pi_{\{h_\gamma\}}$.

\section{Logical operators}\label{sec:logop}

Having described the structure of Tiger codes and introduced 
two equivalent descriptions of the codespace, through the $X$ 
and $Z$ bases, we now turn to the question of how computations 
can be performed on the encoded data. This requires 
characterising the logical operators of any Tiger code.

\begin{definition}\label{def:logop}
   A logical operator is an operator $L$ preserving the codespace, meaning 
   \begin{equation}
      L(\cT(G,H, \balpha, \bDelta))\subset \cT(G,H, \balpha, \bDelta).
   \end{equation}
   Two logical operators $L$ and $L'$ are said to be equivalent if
   they have the same action on the codespace, 
   i.e., 
   \begin{equation}
      L\!\restriction_{\cT(G,H, \balpha, \bDelta)}=
      L'\!\restriction_{\cT(G,H, \balpha, \bDelta)}.
   \end{equation}
   A logical operator is said to be \emph{trivial} if 
   it is equivalent to the identity.
\end{definition}

In this section, we assume that the code under consideration 
satisfies $\ker\pi=0$, i.e. that the codespace hosts only 
rotors and qudits; see \cref{theorem:pi}. 

This assumption 
has two important consequences that simplify the computations 
in what follows. First, the 
embedding $\iota$ defined in \cref{eq:iota} becomes a 
bijection by \cref{prop:iota}, thereby endowing the set 
$K_H^{\bDelta}$ with a group structure.
Second, by \cref{theorem:diagram}, the maps $\pi$ and 
$\cF$ in \cref{eq:diagram} become isometric isomorphisms,
while $\cN^{-1}$ and $\rho = \cN^{-1}\pi\cF^{-1}$ are
bounded and injective with dense range. They are surjective---hence
isomorphisms---if and only if the codespace is finite-dimensional: in the
infinite-dimensional case, $\cN$ is unbounded implying that its domain 
is a strict subspace of the codespace. In
particular, the injectivity of $\rho$ guarantees that expansions along the
$Z$-basis $\{\ket{\omega_{[\varphi]}}\}$ are unique whenever they exist.

Although this section focuses on Tiger codes, 
some of the logical operators derived for them 
generalise straightforwardly to their extended 
counterparts. We now give a sufficient condition 
for a logical operator of a Tiger code to remain a 
logical operator of the corresponding extended Tiger code.

\begin{proposition}
   Let $\cT(G, \balpha, \{h_\gamma\})$ be an extended 
   Tiger code and $L$ be a logical operator of 
   $\cT(G,0,\balpha,0)$.
   If $[L,\Pi_{\{h_\gamma\}}]=0$, then, 
   $L$ is a logical operator of 
   $\cT(G, \balpha, \{h_\gamma\})$.
\end{proposition}

\begin{proof}
   This is a consequence of the fact that 
   $\cT(G, \balpha, \{h_\gamma\})=\Pi_{\{h_\gamma\}}\cT(G,0,\balpha,0)$.
\end{proof}

In \cref{sec:Pauli,sec:Zrotations}, we present logical operators
that are functions of the number operators $\n$ 
for arbitrary Tiger codes. 
As such, they commute with $\Pi_{\{h_\gamma\}}$, making these  
valid logical operators for extended Tiger codes as well.
We first study logical Pauli operators in \cref{sec:Pauli}
before turning to Z-rotations that implement Clifford and 
non-Clifford logical operators in \cref{sec:Zrotations}.

\subsection{Pauli logical operators}\label{sec:Pauli}

The following definition is motivated by the 
definition of the Heisenberg group for $L^2$ functions 
over a LCA group (Theorem 4.1 \cite{Prasad_2011}).

\begin{definition}\label{def:Pauli}
   The family of logical $X$ Pauli operators is the family
   \begin{equation}
      \cL_X:=\left\{\left[X_{[\bn]}\right]:[\bn]\in H_1(C_\bullet)\right\}
   \end{equation}
   of equivalence classes of logical operators indexed by the 
   homology group $H_1(C_\bullet)$. For each \break
   $[\bn']\in H_1(C_\bullet)$, any representative
   $X_{[\bn']}\in [X_{[\bn']}]$
   is required to be bounded on the codespace and to act on 
   the X basis by translation:
   \begin{equation}\label{eq:XPauli_def}
      X_{[\bn']}\ket{\balpha_{\iota^{-1}([\bn])}}
      =\ket{\balpha_{\iota^{-1}([\bn-\bn'])}},
      \qquad \forall[\bn]\in H_1(C_\bullet).
   \end{equation}

   The logical $Z$ Pauli operators form a family
   \begin{equation}
      \cL_Z:=\left\{\left[Z_{[\varphi]}\right]:[\varphi]\in H^1(C_\bullet)\right\}
   \end{equation}
   of equivalence classes of logical operators indexed 
   by the cohomology group $H^1(C_\bullet)$. For each 
   $[\varphi]\in H^1(C_\bullet)$, any representative 
   $Z_{[\varphi]}\in[Z_{[\varphi]}]$ is bounded on the codespace 
   and act on the X basis by adding a phase
   \begin{equation}\label{eq:ZPauli_def}
      Z_{[\varphi]}\ket{\balpha_{\iota^{-1}([\bn])}}
      =\varphi(\bn)\ket{\balpha_{\iota^{-1}([\bn])}},
      \qquad \forall[\bn]\in H_1(C_\bullet)
   \end{equation}
   where the value $\varphi(\bn)\in\T$ is independent of
   the chosen representatives of $[\varphi]$ and $[\bn]$.
\end{definition}

Because $\{\ket{\balpha_{\iota^{-1}([\bn])}}:[\bn]\in H_1(C_\bullet)\}$ 
form a basis of $\cT(G,H, \balpha, \bDelta)$ (\cref{theorem:basis})
and that the representatives are bounded on the codespace, X and 
Z logical Pauli operator representatives are unitaries on the codespace and
both \cref{eq:XPauli_def} and \cref{eq:ZPauli_def} uniquely 
define their action on the codespace.
With this definition, we retrieve the canonical commutation relation
on the codespace
\begin{equation}
   X_{[\bn]}Z_{[\varphi]}=\varphi(\bn)Z_{[\varphi]}X_{[\bn]}.
\end{equation}
Outside of $\cT(G,H,\balpha,\bDelta)$, this relation does 
not necessarily hold. Similarly, while a representative of 
a logical Pauli operator is unitary on the codespace,
it does not need to be unitary on the whole Hilbert space.

\begin{remark}
   Alternatively, we could have defined the Pauli logical operators
   using the isomorphism \cref{theorem:H1isom} 
   and the standard definition for Pauli operators 
   over qudits \cite{Sarkar_2024} and rotors 
   \cite{Vuillot_2024}. Both definitions coincide.
\end{remark}

Finding explicit representatives for each $X$ and $Z$ logical 
Pauli operators is straightforward. For instance, a naive approach 
consists in defining for any $[\bn']\in H_1(C_\bullet)$
\begin{equation}
   \sum_{[\bn]\in H_1(C_\bullet)} 
   \ket{\balpha_{\iota^{-1}([\bn-\bn'])}}\bra{\balpha_{\iota^{-1}([\bn])}}
\end{equation}
which is a representative of $[X_{[\bn']}]$.
However, even though it is theoretically valid,
this representative does not highlight any practical 
implementation of such logical operator.
That is why, in what follows, we narrow our focus down to certain 
representatives of logical operators. 
We investigate $X$ and $Z$ logical 
Pauli operator representatives that belong to 
\begin{equation}\label{eq:Pauli_representatives}
   \{P(\a^\dag,\a) : P\in\C[x_1,\ldots,x_m,y_1,\ldots,y_m]\} 
   \quad \text{and} \quad
   \{e^{i(\bphi\cdot\n+\theta)} : \bphi\in[0,2\pi)^m, \theta\in[0,2\pi)\}
\end{equation}
respectively. 

\paragraph{$Z$ logical Pauli operators.}
For the logical $Z$ Pauli operators, we show that each 
equivalence class $[Z_{[\varphi]}]$ admits a representative 
given by a physical rotation; see \cref{eq:Pauli_representatives}.

\begin{proposition}\label{prop:ZPaulilog}
   For all $[\varphi]\in H^1(C_\bullet)$, consider a representative 
   $\varphi\in [\varphi]$. Then, 
   \begin{equation}
      \varphi(\n-\bn_0)\in
      \{e^{i(\bphi\cdot\n+\theta)} : 
      \bphi\in[0,2\pi)^m, \theta\in[0,2\pi)\}
   \end{equation}
   and $\varphi(\n-\bn_0)\in[Z_{[\varphi]}]$.
\end{proposition}

\begin{proof}
   First, as $\varphi:\Z^m\to\T$, $\varphi(\n-n_0)$ is bounded 
   on the whole Hilbert space, so in particular on the codespace. 
   Furthermore, since $\tau$ introduce in \cref{eq:tau} is an 
   isomorphism,  \break
   $\varphi(\n-\bn_0)=e^{i\tau^{-1}([\varphi])\cdot(\n-\bn_0)}$ 
   where $\tau^{-1}([\varphi])\in[0,2\pi)^m$ (\cref{eq:sequence_isom}).
   So,
   \begin{equation}
      \varphi(\n-\bn_0)\in \{e^{i(\bphi\cdot\n+\theta)} : \bphi\in[0,2\pi)^m, \theta\in[0,2\pi)\}
   \end{equation} 
   for all $[\varphi]\in H^1(C_\bullet)$.
   Finally, for all $[\bn]\in H_1(C_\bullet)$
   \begin{equation}
      \varphi(\n-\bn_0)\ket{\balpha_{\iota^{-1}([\bn])}} 
      =\frac{e^{-|\balpha|^2/2}}{A_{[\bn]}}
      \sum_{\bk\in\iota^{-1}([\bn])}\frac{\balpha^{\bk}}{\sqrt{\bk!}}\varphi(\bk-\bn_0)\ket{\bk}
      =\varphi(\bn)\ket{\balpha_{\iota^{-1}([\bn])}}
   \end{equation}
   as $\varphi(\bg)=1$ for all $\bg\in\im{G}$. 
   So, $\varphi(\n-\bn_0)\in[Z_{[\varphi]}]$ and its action 
   on the codespace does not depend on 
   the choice of the representative $\varphi$ as $\bn-\bn_0\in\ker{H}$. 
   Hence, all the representatives 
   $\varphi\in[\varphi]$ yield equivalent logical operators.
\end{proof}

\paragraph{$X$ logical Pauli operators.}
While logical Z Pauli operators always admit representatives 
given by physical rotations, and hence can be implemented 
through Hamiltonian dynamics, it is less clear whether 
logical X Pauli operators admit comparably simple representatives.
In fact, finding a polynomial $P$ such that $P(\a^\dag,\a)$
is a $X$ logical Pauli operator representative turns out 
to be difficult in general. 
For this reason, we introduce two simplifying hypotheses: 
1) we study \emph{approximate} logical $X$ Pauli operators, and 
2)we assume that the logical space is finite-dimensional. 
Both hypotheses are motivated and discussed below.
\begin{definition}
   The approximate $X$ logical Pauli operators form a family
   \begin{equation}
      \cL_{\tilde{X}}:=\{[\tilde{X}_{[\bn]}]:[\bn]\in H_1(C_\bullet)\}.
   \end{equation}
   For each $[\bn]\in H_1(C_\bullet)$, any 
   representative $\tilde{X}_{[\bn]}\in[\tilde{X}_{[\bn]}]$
   must satisfy
   \begin{equation}\label{eq:appXlogical_def}
      \tilde{X}_{[\bn]}\ket{\omega_{[\varphi]}}
      =\varphi(\bn)\ket{\omega_{[\varphi]}},
      \qquad \forall[\varphi]\in H^1(C_\bullet).
   \end{equation}
   and $\tilde{X}_{[\bn]}$ must be bounded on the codespace.
\end{definition}

Recall that the family $\{\ket{\omega_{[\varphi]}}:[\varphi]\in H^1(C_\bullet)\}$ 
is a total set of the codespace by \cref{theorem:diagram}.
Therefore, similarly to \cref{def:Pauli}, once 
$\tilde{X}_{[\bn]}$ is required to be bounded 
on the codespace, its action on this family uniquely determines 
its action on the entire codespace.
But contrary to \cref{def:Pauli}, approximate $X$ logical
Pauli operators are \emph{not} unitary on the codespace as the 
$\ket{\omega_{[\varphi]}}$'s are not orthonormal.

For instance, the approximate $X$ logical Pauli operator of the 
cat code $\ker{\ha^2-\alpha^2}$ is simply $\frac{\ha}{\alpha}$ because
\begin{equation}
   \frac{\ha}{\alpha}\ket{\alpha}=\ket{\alpha} 
   \quad \text{and} \quad
   \frac{\ha}{\alpha}\ket{-\alpha}=-\ket{-\alpha}.
\end{equation}

For the same implementation reasons aforementioned, 
we are investigating for representatives of approximate 
$X$ logical Pauli operators that are still of the form
\begin{equation}
   \{P(\a^\dag,\a) : P\in\C[x_1,\ldots,x_m,y_1,\ldots,y_m]\}.
\end{equation}
Because no simple general boundedness statement on polynomials in 
$\a$ and $\a^\dag$ is available for arbitrary
infinite dimensional Tiger codes,
for the remainder of this subsection we restrict 
our attention to Tiger codes with finite-dimensional codespaces.

\begin{proposition}\label{prop:logicalPauliX}
   For positive Tiger codes with finite dimensional codespaces, 
   for all $[\bl]\in H_1(C_\bullet)$,
   $[\tilde{X}_{[\bl]}]=\left[\frac{\a^{\bl}}{\balpha^{\bl}}\right]$
   where $\bl\in[\bl]$ is any representative satisfying $\bl\geq0$ componentwise.
\end{proposition}

\begin{proof}
   Let $[\bl]\in H_1(C_\bullet)$.
   Since the code is positive, there exists $\bg\in\im{G}$ with $\bg>0$
   componentwise.
   Hence one can choose $\lambda\in\N$ large enough so that
   $\bl+\lambda \bg\ge 0$.
   Replacing $\bl$ by the equivalent representative $\bl+\lambda \bg$, we may
   without loss of generality assume that $\bl\ge0$.

   For any $[\varphi]\in H^1(C_\bullet)$, using the equality 
   $\Pi_{H}=\indicator{\{H\n=\bDelta\}}=\indicator{\{H(\n+\bl)=\bDelta\}}$, 
   we get
   \begin{equation}
      \begin{aligned}
         \frac{\a^{\bl}}{\balpha^{\bl}}\ket{\omega_{[\varphi]}}
         &=\frac{\a^{\bl}}{\balpha^{\bl}}\varphi(\n-\bn_0)\Pi_{H}\ket{\balpha}\\
         &=\varphi(\n-\bn_0+\bl)\indicator{\{H(\n+\bl)=\bDelta\}}\frac{\a^{\bl}}{\balpha^{\bl}}\ket{\balpha}\\
         &=\varphi(\bl)\ket{\omega_{[\varphi]}}.
      \end{aligned}
   \end{equation}
   as $\a^{\bl}\ket{\balpha}=\balpha^{\bl}\ket{\balpha}$.
   Thus $\frac{\a^{\bl}}{\balpha^{\bl}}$ acts on the Z-family exactly as the
   approximate logical X operator labelled by $[\bl]$.
   Since $\varphi(\bl)$ depends only on the homology class $[\bl]$, different
   non-negative representatives give equivalent logical operators.
   Finally, as the codespace is finite-dimensional,
   $\a^{\bl}$ is clearly bounded on the codespace.
\end{proof}

However, for non-positive Tiger codes, finding such representatives
of approximate X logical Pauli operators remains challenging. 
Here, we provide an algorithmic method to find \emph{some of} the 
the representatives in the form of a polynomial $P(\a^\dag,\a)$ 
in $\a$ and $\a^\dag$---if they exist---following 
the approach of \cite{xu2024lettingtigercagebosonic}.
To narrow down the space of polynomials we are looking for, 
we start by noticing that each monomial $\bx^{\bi}\by^{\bj}$ 
in $P(\bx,\by)$ should satisfy $\bi-\bj\in\ker{H}$.

\begin{proposition}
   Let $P(\a^\dag,\a)=\sum_{\substack{\bi,\bj\\ \bi-\bj\notin\ker{H}}}\lambda_{\bi,\bj}\a^{\dag\bi}\a^{\bj}$
   be a logical operator representative. 
   Then, $P(\a^\dag,\a)\!\restriction_{\cT(G,H,\balpha,\bDelta)}=0$.
\end{proposition}

\begin{proof}
   Assume that $P(\a^\dag,\a)$ is a logical operator.
   Then, for any $[\bn]\in K_H^{\bDelta}/\!\!\sim_G$, 
   if \break $P(\a^\dag,\a)\ket{\balpha_{[\bn]}}\neq0$, then
   \begin{equation}
      P(\a^\dag,\a)\ket{\balpha_{[\bn]}} \notin 
      \overline{\Span{\ket{\bk}:\bk\in\N^m,H\bk=\bDelta}}.
   \end{equation}
   But, by \cref{def:logop}, 
   $P(\a^\dag,\a)\ket{\balpha_{[\bn]}}\in\cT(G,H,\balpha,\bDelta)\subset \overline{\Span{\ket{k}:k\in\N^m,Hk=\bDelta}}$,
   which is absurd.
\end{proof}

Thus, by linearity, we shall only consider polynomials 
$P(\a^\dag,\a)$ of the following form 
\begin{equation}
   P(\a^\dag,\a) = 
   \sum_{\substack{\bi,\bj\\ \bi-\bj\in\ker{H}}}\lambda_{\bi,\bj}\a^{\dag\bi}\a^{\bj}.
\end{equation}
The algorithm presented in \cite{xu2024lettingtigercagebosonic} rely 
on the following computation.
For a monomial $\a^{\dag\bi}\a^{\bj}$ such that 
$\bj-\bi\in[\bn]\subset\ker{H}$, 
for all $[\varphi]\in H^1(C_\bullet)$,
\begin{equation}\label{eq:appXintuition}
   \begin{aligned}
      \a^{\dag\bi}\a^{\bj}\ket{\omega_{[\varphi]}}
      &=\a^{\dag\bi}\a^{\bj}\varphi(\n-\bn_0)\Pi_{H}\ket{\balpha}\\
      &=\varphi(\n-\bn_0-(\bi-\bj))\indicator{\{H(\n-(\bi-\bj))=\bDelta\}}\a^{\dag\bi}\a^{\bj}\ket{\balpha}\\
      &=\varphi(\bj-\bi)\varphi(\n-\bn_0)\Pi_{H}\a^{\dag\bi}\a^{\bi}\balpha^{\bj-\bi}\ket{\balpha}\\
      &=\varphi(\bn)\varphi(\n-\bn_0)\Pi_{H}\bi!\binom{\n}{\bi}\balpha^{\bj-\bi}\ket{\balpha}\\
      &=\varphi(\bn)\balpha^{\bj-\bi}\bi!\binom{\n}{\bi}\ket{\omega_{[\varphi]}}.
   \end{aligned}
\end{equation}
So, if there exists a linear combination $\sum_{\bi}\lambda_{\bi}\binom{\n}{\bi}$ 
such that this operator has a trivial action on the codespace, then 
$P(\a^\dag,\a)=\sum_{\bi}\frac{\lambda_{\bi}}{\balpha^{\bj-\bi}\bi!}\binom{\n}{\bi}$ is a 
representative of the approximate X logical Pauli operator $[\tilde{X}_{[\bn]}]$.
Thus, the problem boils down to finding 
a linear combination $\sum_{\bi}\lambda_{\bi}\binom{\n}{\bi}$ 
with trivial action on the codespace. This can be done using 
polynomial ideal reduction methods such as 
Gröbner basis algorithms.
The next lemma and proposition formalise this construction.

\begin{lemma}
   For any $[\bn]\in H_1(C_\bullet)$, the set
   \begin{equation}
      \I_{[\bn]}:=\left\{\sum_{\substack{\bi,\bj\\ \bj-\bi\in[\bn]}}
      \mu_{\bi,\bj}\binom{\bx}{\bi}:\mu_{\bi,\bj}\in\C \text{ with finite support}\right\}
   \end{equation}
   is a polynomial ideal in $\C[x_1,\ldots,x_m]$.
\end{lemma}

\begin{proof}
   Stability by addition of $\I_{[\bn]}$ is clear. To prove stability
   by multiplication of any polynomial,
   it is sufficient to prove that given 
   $\binom{\bx}{\bi}\in\I_{[\bn]}$
   and any monomial $\bx^k$ for 
   $\bk\in\N^m,\bx^{\bk}\binom{\bx}{\bi}\in\I_{[\bn]}$.
   For any integer $n\in\N$, define the polynomials
   \begin{equation}
      Q_p^{(q)}(x):=\left\{\begin{aligned}
         &1 &\text{if }p=0 \\
         &\frac{x-q}{q+1}\ldots \frac{x-(q+p-1)}{q+p} & \text{otherwise}
      \end{aligned}\right.
   \end{equation}
   in $\C[x]$. They form a basis of $\C[x]$ and $\binom{x}{q}Q_p^{(q)}(x)=\binom{x}{q+p}$.
   Thus, we can write
   \begin{equation}
      \begin{aligned}
         \bx^{\bk}\binom{\bx}{\bi}&=x_1^{k_1}\binom{x_1}{i_1}\ldots x_m^{k_m}\binom{x_m}{i_m}\\
         &=\sum_{p_1}\lambda_{p_1}^{(1)}Q_{p_1}^{(i_1)}(x)\binom{x_1}{i_1}\ldots
         \sum_{p_m}\lambda_{p_m}^{(m)}Q_{p_m}^{(i_m)}(x)\binom{x_m}{i_m}\\
         &=\sum_{p_1}\lambda_{p_1}^{(1)}\binom{x_1}{i_1+p_1}\ldots
         \sum_{p_m}\lambda_{p_m}^{(m)}\binom{x_m}{i_m+p_m}\\
         &=\sum_{p_1,\ldots,p_m}\lambda_{p_1}^{(1)}\ldots \lambda_{p_m}^{(m)}\binom{\bx}{\bi+\bp}
      \end{aligned}
   \end{equation}
   where the $\lambda_{p_1}^{(1)},\ldots,\lambda_{p_m}^{(m)}$'s 
   are the coefficients of $x_1^{k_1},\ldots,x_m^{k_m}$ in the basis \break
   $\{Q_{p_1}^{(i_1)}(x)\}_{p_1\in\N},\ldots,\{Q_{p_m}^{(i_m)}(x)\}_{p_m\in\N}$, respectively.
   Finally, as $\binom{\bx}{\bi}\in\I_{[\bn]}$, there exists $\bj\in\N^m$
   such that $\bj-\bi\in[\bn]$. Hence, as $\bj+\bp-(\bi+\bp)\in[\bn]$ for all 
   $\bp\in\N^m$, $\binom{\bx}{\bi+\bp}\in\I_{[\bn]}$ and 
   thus $\bx^{\bk}\binom{\bx}{\bi}\in\I_{[\bn]}$.
\end{proof}

\begin{proposition}\label{prop:findingXlog}
   Consider the polynomial ideal  
   \begin{equation}
      \J:=\langle\{\bh_k\cdot\bx-\Delta_k:k\in\llbracket1,s\rrbracket\}\rangle
   \end{equation}
   of $\C[x_1,\ldots,x_m]$, where $\bh_k$ is the $k$-th row of $H$.
   For $[\bn]\in H_1(C_\bullet)$, if $1\in\I_{[\bn]}+\J$, then there exists a representative of 
   $[\tilde{X}_{[\bn]}]$ taking the form of a polynomial in $\a$ and $\a^\dag$.
\end{proposition}

\begin{proof}
   For a monomial $\a^{\dag\bi}\a^{\bj}$ such that $\bj-\bi\in[\bn]\subset\ker{H}$, 
   for all $[\varphi]\in H^1(C_\bullet)$, \cref{eq:appXintuition} gives 
   \begin{equation}
      \a^{\dag\bi}\a^{\bj}\ket{\omega_{[\varphi]}}=
      \varphi(\bn)\balpha^{\bj-\bi}\bi!\binom{\n}{\bi}\ket{\omega_{[\varphi]}}
   \end{equation}
   So, if $1\in\I_{[\bn]}+\J$, there exists a decomposition
   \begin{equation}
      1=\sum_{\substack{\bi,\bj\\ \bj-\bi\in[\bn]}}\mu_{\bi,\bj}
      \binom{\bx}{\bi}+\sum_{k=1}^{s}Q_k(\bx)(\bh_k\cdot\bx-\Delta_k)
   \end{equation}
   and, taking the polynomial 
   \begin{equation}
      P(\a^\dag,\a) = 
      \sum_{\substack{\bi,\bj\\ \bj-\bi\in[n]}}
      \frac{\mu_{\bi,\bj}}{\balpha^{\bj-\bi}\bi!}\a^{\dag\bi}\a^{\bj},
   \end{equation}
   we get 
   \begin{equation}
      \begin{aligned}
         P(\a^\dag,\a)\ket{\omega_{[\varphi]}}
         &=\sum_{\substack{\bi,\bj\\ \bj-\bi\in[\bn]}}\mu_{\bi,\bj}\varphi(\bn)\binom{\n}{\bi}\ket{\omega_{[\varphi]}}\\
         &=\varphi(\bn)\left(1-\sum_{k=1}^{s}Q_k(\n)(\bh_k\cdot\n-\Delta_k)\right)\ket{\omega_{[\varphi]}}\\
         &=\varphi(\bn)\ket{\omega_{[\varphi]}}
      \end{aligned}
   \end{equation}
   as $\ket{\omega_{[\varphi]}}\in\cT(G,H,\balpha,\bDelta)\subset \bigcap_{k\in\llbracket1,s\rrbracket}\ker{\bh_k\cdot\n-\Delta_k}$.
   Hence, $P(\a^\dag,\a)$ is a representative of $[\tilde{X}_{[\bn]}]$.
\end{proof}

Note that the above proposition only gives a sufficient condition on representatives of 
$X$ logical Pauli operators taking the form of a polynomial in $\a$ and $\a^\dag$.

\subsection{Z rotations: non-Clifford gates}\label{sec:Zrotations}

In this section, we look for logical operators admitting 
representatives of the form $e^{2i\pi P(\n)}$, where 
$P\in\R[x_1,\ldots,x_m]$. By \cref{prop:ZPaulilog}, 
every logical $Z$ Pauli operator admits such a representative, 
but these operators do not exhaust all possibilities. 

Importantly, since $e^{2i\pi P(\n)}$ is bounded, its restriction to 
the codespace is also bounded. Consequently, throughout 
this section it suffices to prove identities on an 
orthonormal basis of the codespace; they then extend to 
the whole codespace by continuity. This consideration is important 
whenever the codespace is infinite-dimensional.

The condition that $e^{2i\pi P(\n)}$ preserve the code 
can be recast as a condition on the polynomial $P$.

\begin{proposition}\label{prop:periodicity_of_P}
   $e^{2i\pi P(\n)}$ preserves the codespace if and only if
   \begin{equation} \label{eq:P_logical}
      \forall \bn\in K_H^{\bDelta}, \forall \bg \in \im{G},
      \quad \bn+\bg\in\N^m\Rightarrow P(\bn+\bg)-P(\bn) \in \Z.
   \end{equation}
\end{proposition}

\begin{proof}
   For any $[\bn]\in K_H^{\bDelta}/\!\!\sim_G$,
   \begin{equation}
      e^{2i\pi P(\n)} \ket{\balpha_{[\bn]}} 
      = \frac{e^{-|\balpha|^2/2}}{A_{[\bn]}}
      \sum_{\bk\in[\bn]} \frac{\balpha^{\bk}}{\sqrt{\bk!}} 
      e^{2i\pi P(\bk)}\ket{\bk}.
   \end{equation}
   If $P$ satisfies \cref{eq:P_logical}, then 
   $e^{2i\pi P(\n)} \ket{\balpha_{[\bn]}}=e^{2i\pi P(\bn)}\ket{\balpha_{[\bn]}}$, 
   and thus $e^{2i\pi P(\n)}$ preserves the code.

   If $e^{2i\pi P(\n)}$ preserves the codespace, then 
   $e^{2i\pi P(\n)}\ket{\balpha_{[\bn]}}\in\cT(G,H,\balpha,\bDelta)$. 
   Remark that \break $e^{2i\pi P(\n)} \ket{\balpha_{[\bn]}}$ has the same 
   support as $\ket{\balpha_{[\bn]}}$ in the Fock basis.
   So, as the $\ket{\balpha_{[\bn]}}$'s form an orthonormal basis 
   of $\cT(G,H,\balpha,\bDelta)$ and that they have disjoint 
   support in the Fock basis, this implies that 
   $e^{2i\pi P(\n)} \ket{\balpha_{[\bn]}}$ must be proportional to 
   $\ket{\balpha_{[\bn]}}$. 
   Hence, by uniqueness of the decomposition in the Fock basis,
   \begin{equation}
      \forall \bk \in [\bn], e^{2i\pi P(\bk)}=e^{2i\pi P(\bn)}.
   \end{equation}
   As this condition holds for all $[\bn]\in K_H^{\bDelta}/\!\!\sim_G$, we get
   the desired characterisation.
\end{proof}

Remark that instead of $\im{G}$ in \cref{eq:P_logical}, one could have 
used a Markov basis $\{\bb_1\ldots,\bb_M\}$ of $G$; see \cref{def:Markov_basis}. 
In fact, for all $\bn\in K_H^{\bDelta}$,
\begin{equation}
   \begin{aligned}
      &\left(\forall\bg\in\im{G}, \bn+\bg\in\N^m\Rightarrow 
      P(\bn+\bg)-P(\bn) \in\Z\right)\\
      \Leftrightarrow 
      &\quad\left(\forall i\in\llbracket1,M\rrbracket, 
      \bn+\bb_i\in\N^m\Rightarrow P(\bn+\bb_i)-P(\bn) \in \Z\right).
   \end{aligned}
\end{equation} 
While the left to right implication is clear, the right to left one 
results from the fact that there exists a path 
$\bn\to \bn+\bb_{i_1}\to\ldots\to \bn+\bb_{i_1}+\ldots+\bb_{i_N}=\bn+\bg$
in $\N^m$ and 
\begin{equation}
   \begin{aligned}
      P(\bn+\bg)-P(\bn)=&\left(P(\bn+\bg)-P(\bn+\bb_{i_1}+\ldots+\bb_{i_{N-1}})\right)+\ldots\\
      &+\left(P(\bn+\bb_{i_1}+\bb_{i_2})-P(\bn+\bb_{i_1})\right)+\left(P(\bn+\bb_{i_1})-P(\bn)\right)\in\Z.
   \end{aligned}
\end{equation}

For the remainder of this section, we mainly focus on positive 
Tiger codes, for which the polynomials satisfying 
\cref{eq:P_logical} can be partially characterised. 
In \cref{theorem:Zrotations}, under a technical assumption 
on the Smith decomposition of the matrix $G$ and assuming 
the absence of logical rotors, we obtain a structural necessary 
condition on admissible polynomials. In the logical one-qubit case, 
this condition becomes rigid enough to determine the  
polynomial of minimal degree required to implement a prescribed logical 
$Z$-rotation; see \cref{prop:qubitZrotation}. We return to 
non-positive codes at the end of the section, in 
\cref{ex:binomial_code3}.

The proof of \cref{theorem:Zrotations} proceeds in two steps. 
We first study univariate polynomials $P\in\R[x]$ such that, 
for a fixed $a\in\N$, the finite difference $P(x+a)-P(x)$ takes 
integer values on the integers. This is the purpose of 
\cref{lemma:integer_valued_polynomial,prop:mono_P_characterisation}. 
We then spell out the two assumptions on the Tiger code and 
their consequences for the proof of \cref{theorem:Zrotations} 
in \cref{lemma:infinite_many_constraints,lemma:fixed_last_components}.

Given a positive integer $a\in\N$, we introduce the difference
operator $\delta_a$ defined over $\R[x]$ as 
\begin{equation}
   \begin{aligned}
      \delta_a:\R[x]&\to\R[x]\\
      P&\mapsto P(x+a)-P(x).
   \end{aligned}
\end{equation}
By default, $\delta:=\delta_1$.

\begin{lemma}\label{lemma:integer_valued_polynomial}
   For any integer $p\in \Z$, 
   $\left\{\binom{x}{n}\right\}_{n\in\N}$ forms a generating set of the group
   \begin{equation}
      \left\{R\in\R[x]: R(\pm\N_{\geq p})\subset\Z\right\}.
   \end{equation}
\end{lemma}

\begin{proof}
   First, let us prove that $\left\{R\in\R[x]: R(\N)\subset\Z\right\}$
   is generated by $\left\{\binom{x}{n}\right\}_{n\in\N}$. 
   The inclusion from right to left is obvious. 
   In the other direction, consider a polynomial 
   $R\in\left\{R\in\R[x]: R(\N)\subset\Z\right\}$.
   Using Jordan's identity \cite{Jordan}
   \begin{equation}
      R = \sum_{n=0}^{deg(R)}\delta^nR(0) \binom{x}{n},
   \end{equation}
   since $R(\N)\subset\Z$, $\delta^nR(0) \in \Z$, hence 
   $R \in \left\langle\left\{\binom{x}{n}\right\}_{n\in\N}\right\rangle_{\Z}$.

   Now, given $p\in\Z$ and $Q\in\R[x]$, if $Q(\pm\N_{\geq p})\subset\Z$, then \break
   $P(x):=Q(\pm(x+p)) \in \left\{R\in\R[x]: R(\N)\subset\Z\right\}$. So, 
   $P\in \left\langle\left\{\binom{\pm(x-p)}{n}\right\}_{n\in\N}\right\rangle_{\Z}=\left\langle\left\{\binom{x}{n}\right\}_{n\in\N}\right\rangle_{\Z}$.
\end{proof}

\begin{proposition}\label{prop:mono_P_characterisation}
   Let $a\in\N^*$. Define the family of univariate polynomials
   \begin{equation}
      \forall n\in\N, P_{n+1}^{(a)}(x):=\sum_{k=0}^n \sum_{j=0}^k (-1)^{k-j}\binom{k}{j}\binom{aj}{n}\binom{x/a}{k+1}
   \end{equation} 
   with $P_0^{(a)}(x)=1$.
   For all $p\in\Z$, we have 
   \begin{equation}
      \left\{P\in\R[x]: \delta_aP(\pm\N_{\geq p})\subset\Z\right\}
      =\R+\left\langle \left\{P_n^{(a)}\right\}_{n\geq1}\right\rangle_{\Z}.
   \end{equation}
\end{proposition}

Remark that $P_{n}^{(a)}$ is a polynomial of degree $n$. 
In particular, the first three non-constant polynomials are
\begin{equation}\label{eq:3firstPa}
   \begin{aligned}
      &P_1^{(a)}(x)=\frac{x}{a},\\
      &P_2^{(a)}(x)=\frac{1}{2a}x^2-\frac{1}{2}x,\\
      &P_3^{(a)}(x)=\frac{1}{6a}x^3 - \frac{1+a}{4a}x^2 + \frac{a+3}{12}x.
   \end{aligned}
\end{equation}

\begin{proof}
   If $P$ satisfies $\delta_aP(\pm\N_{\geq p})\subset\Z$, according to \cref{lemma:integer_valued_polynomial},
   we have $\delta_aP \in \left\langle\left\{\binom{x}{n}\right\}_{n\in\N}\right\rangle_{\Z}$.
   The question is now how to retrieve $P$ from $\delta_aP$. To do so, we use the following 
   identity 
   \begin{equation}
      \forall n\in\N, \delta_a\binom{x}{n}_{\! a}=\left\{\begin{aligned}
         &a\binom{x}{n-1}_{\! a} & \text{if } n\geq1 \\
         & 0 & \text{otherwise}
      \end{aligned}\right.
   \end{equation} 
   where $\binom{x}{n}_{\! a}:=a^n\binom{x/a}{n}=\frac{x(x-a)\ldots(x-a(n-1))}{n!}$.
   From this identity, we deduce that $\delta_a$ is surjective but not injective. 
   In particular, $\ker{\delta_a}=\R$. 

   So, following Jordan's work \cite{Jordan}, we define the 
   linear operator $\delta_a^{-1}$ for all $n\in\N$ as 
   \begin{equation}
      \delta_a^{-1}\binom{x}{n}_{\! a}:=\frac{1}{a}\binom{x}{n+1}_{\! a}.
   \end{equation}
   This way, for any $P\in\R[x], \delta_a^{-1}\delta_aP=P-P(0)$.
   Because $\delta_aP \in \left\langle\left\{\binom{x}{n}\right\}_{n\in\N}\right\rangle_{\Z}$,
   $P(0)\in\R$ and because $\delta_a^{-1}$ is linear, we have
   \begin{equation}
      P\in \R+\left\langle\left\{\delta_a^{-1}\binom{x}{n}\right\}_{n\in\N}\right\rangle_{\Z}.
   \end{equation}
   It remains to make $\delta_a^{-1}\binom{x}{n}$ explicit.
   Using Jordan's identity, we can decompose $\binom{x}{n}$ in the 
   $\left\{\binom{x}{n}_a\right\}_{n\in\N}$ basis to get 
   \begin{equation}
      \binom{x}{n}=\sum_{k=0}^n \frac{\delta_a^k\binom{x}{n}(0)}{a^k}\binom{x}{k}_a.
   \end{equation}
   Consider $E_a:P(x)\mapsto P(x+a)$. using the operator equality 
   $\delta_a^k=(E_a-1)^k=\sum_{j=0}^k \binom{k}{j}(-1)^{k-j}E_a^j$,
   we obtain
   \begin{equation}
      \binom{x}{n}=\sum_{k=0}^n 
      \frac{1}{a^k}\sum_{j=0}^k (-1)^{k-j}\binom{k}{j}
      \binom{aj}{n}\binom{x}{k}_a.
   \end{equation}
   Hence, 
   \begin{equation}
      \delta_a^{-1}\binom{x}{n} = P_{n+1}^{(a)}(x) 
      = \sum_{k=0}^n \frac{1}{a^{k+1}}\sum_{j=0}^k (-1)^{k-j}\binom{k}{j}\binom{aj}{n}\binom{x}{k+1}_a.
   \end{equation}
\end{proof}

We now state the two technical assumptions on the 
code that are needed to prove \cref{theorem:Zrotations} 
and derive their main consequences.

\begin{lemma}\label{lemma:infinite_many_constraints}
   Consider a positive Tiger code $\cT(G,H, \balpha, \bDelta)$.
   Let 
   \begin{equation}
      LGC = \diag{m\times r}{d_1,\ldots,d_t}=D
   \end{equation} 
   be a Smith decomposition of $G$. Assume that each column of $GC$ is either non-negative or 
   non-positive, and denote by 
   $\sigma_i$ the sign of the $i$-th column. Then, there exists 
   $\bn_*\in\bn_0+\ker H$ such that
   \begin{equation}
      L\bn_* + 
      \begin{pmatrix}
         \sigma_1 \N\\
         \vdots \\
         \sigma_t \N\\
         0_{m-t}
      \end{pmatrix}\subset L(K_H^{\bDelta}).
   \end{equation}
\end{lemma}

\begin{proof}
   Let $\{\be_i\}_{1\leq i\leq m}$ be the canonical basis of $\R^m$.
   Take $\bn_*\in\bn_0+\ker H$ large enough so that 
   $\bn_*\geq \sum_{i=1}^t(d_i-1)(L^{-1}\be_i)^-$ componentwise.
   Importantly, such an $\bn_*$ exists as $G$ is positive and 
   $\im G \subset \ker H$.
   
   On one hand, by hypothesis, we can consider $\sigma_1\bg_1,\ldots,\sigma_r\bg_r$ the 
   columns of $GC$ where each $\bg_i\geq0$. 
   So, for all $p_i\in\N$, $p_i\bg_i \in \N^m\cap\im G \subset \N^m\cap\ker H$ and thus 
   $\sum_{i=1}^rp_i\bg_i \in \N^m\cap\ker H$.

   On the other hand, the $t$ first columns 
   $\{\bl_{d_1},\ldots,\bl_{d_t}\}$ of $L^{-1}$
   belong to $\ker{H}$ as
   \begin{equation}
      H\bl_{d_i}=HL^{-1}D\frac{\be_i}{d_i}=\frac{1}{d_i}HGC\be_i=0.
   \end{equation}
   So, taking for all $i\in\llbracket1,t\rrbracket$, $q_i\in\llbracket 0,d_i-1 \rrbracket$,
   then $\bn_*+\sum_{i=1}^tq_iL^{-1}\be_i \in\bn_0+\ker H$ and, by definition of $\bn_*$,
   we have that $\bn_*+\sum_{i=1}^tq_iL^{-1}\be_i$ is non-negative.
   Hence, for any family $(p_i)_{1\leq i\leq t}$ in $\N$, we get that 
   \begin{equation}
      \bn_*+\sum_{i=1}^t(q_iL^{-1}\be_i+p_i\bg_i)\in K_H^{\bDelta}.
   \end{equation}
   Finally, using $L\bg_i=\sigma_iLGC\be_i=\sigma_id_i\be_i$ for $i\in\llbracket1,t\rrbracket$, we get
   \begin{equation}
      L(\bn_*+\sum_{i=1}^t(q_iL^{-1}\be_i+p_i\bg_i)) 
      = L\bn_* + \sum_{i=1}^t(q_i+\sigma_ip_id_i)\be_i\in L(K_H^{\bDelta}).
   \end{equation}
   We conclude with
   \begin{equation}
      \begin{pmatrix}
         \sigma_1 \N\\
         \vdots \\
         \sigma_t \N\\
         0_{m-t}
      \end{pmatrix}\subseteq
      \left\{\sum_{i=1}^t(q_i+\sigma_ip_id_i)\be_i : p_i\in\N, 
      q_i \in \llbracket0,d_i-1\rrbracket\right\}
   \end{equation}
   where the equality holds if $\sigma_i=+1$ for all $i\in\llbracket1,t\rrbracket$.
\end{proof}

\begin{lemma}\label{lemma:fixed_last_components}
   Consider a Tiger code $\cT(G,H, \balpha, \bDelta)$ isomorphic to 
   $L^2(H_1(C_\bullet))$ (\cref{theorem:pi}), and let 
   \begin{equation}
      LGC = \diag{m\times r}{d_1,\ldots,d_t}=D
   \end{equation} 
   be a Smith decomposition of $G$. 
   Denote $f=m-t$ and decompose $L$ in a block form 
   \begin{equation}
      L=\begin{pmatrix}
         L_t \\
         L_f
      \end{pmatrix}
   \end{equation} 
   where $L_t$ (resp. $L_f$) are the $t$ first (resp. $f$ last) rows of $L$.
   If the code does not host any logical rotor,
   then $\ker{L_f}=\ker{H}$.
\end{lemma}

\begin{proof}
   First, we prove that $\ker{L_f}\subset\ker{H}$. Let $\by\in\ker{L_f}$
   and consider $\by'=L\by$. Because $\by\in\ker{L_f}$, the last 
   $f$ entries of $\by'$ are zero. Define
   \begin{equation}
      \bx'=\begin{pmatrix}
         \frac{y'_1}{d_1}&\ldots&\frac{y'_t}{d_t}&0&\ldots&0
      \end{pmatrix}^\top\in\Q^r,
   \end{equation}
   such that $\by'=D\bx'$. This way, using $G=L^{-1}DC^{-1}$, we get  
   \begin{equation}
      \by=L^{-1}\by'=L^{-1}D\bx'=GC\bx'.
   \end{equation}
   Hence, taking $\bx=C\bx'\in\Q^r$, we deduce that $\by=G\bx$.
   Finally, because $HG=0_{s\times r}$ in $\Z$, the equality holds also over $\Q$
   and $\by\in\ker{H}$. So, $\ker{L_f}\subset\ker{H}$, and 
   taking the Pontryagin dual, we infer that
   \begin{equation}
      \impi{H^\top}\subset\impi{L_f^\top}.
   \end{equation}
   Now, if this inclusion was strict, there would exist $\bphi\in\impi{L_f^\top}\setminus\impi{H^\top}$.
   Let $\bphi'\in[0,2\pi)^f$ such that $\bphi=L_f^\top\bphi'$.
   Because $LG=DC^{-1}$, we have $\bphi^\top G=\bphi'^\top L_fG=0$.
   So, for any angle $\theta\in[0,2\pi),G^\top\theta\bphi=0$,
   meaning that the group $\kerpi{G^\top}/\impi{H^\top}$ is not discrete. But as the code is 
   isomorphic to $L^2(H_1(C_\bullet))$ which is itself isomorphic 
   to $L^2(H^1(C_\bullet))$ and that no logical rotor belongs to the code,
   $H^1(C_\bullet)$ is supposed to be a finite group according to 
   \cref{eq:basis_ker_2pi}, which is absurd. So, $\impi{H^\top}=\impi{L_f^\top}$ and 
   \begin{equation}
      \impi{H^\top}=\impi{L_f^\top} \quad\Leftrightarrow\quad \im{H^*}=\im{L_f^*}
      \quad\Leftrightarrow\quad \ker{H}=\ker{L_f},
   \end{equation}
   by Pontryagin duality.
\end{proof}

We can now record the main theorem of this section. 
For a subclass of positive Tiger codes, it 
provides a necessary condition on the form of the multivariate 
polynomial $P\in\R[X_1,\ldots,X_n]$ for $e^{iP(\n)}$ to be a 
logical operator.
For completeness, we first state the theorem in broad generality 
before specialising it to the logical one-qubit case where 
the logical action of $e^{iP(\n)}$ can be inferred.

\begin{theorem}\label{theorem:Zrotations}
   Consider a positive Tiger code $\cT(G,H,\balpha,\bDelta)$ over $m$ modes, and let
   \begin{equation}
      LGC=\diag{m\times r}{d_1,\ldots,d_t}=D
   \end{equation}
   be a Smith decomposition of $G$. Denote by
   $\{\bw_{d_1},\ldots,\bw_{d_t},\bw_1,\ldots,\bw_{m-t}\}$
   the rows of $L$. Assume that:
   \begin{enumerate}
   \item the code $\cT(G,H,\balpha,\bDelta)$ does not host any logical rotors. In this case, the rows
   \break $\frac{2\pi}{d_1}\bw_{d_1},\ldots,\frac{2\pi}{d_t}\bw_{d_t}$ form a generating set of
   \begin{equation}
      \kerpi{G^\top}/\impi{H^\top}\cong H^1(C_\bullet);
   \end{equation}
   see \cref{theorem:pi,eq:basis_ker_2pi},
   \item each column of $GC$ is either componentwise non-negative or componentwise non-positive.
   \end{enumerate}
   Let $P\in\R[x_1,\ldots,x_m]$ satisfy \cref{eq:P_logical}. Write $P$ in the form
   \begin{equation}\label{eq:explicit_P}
      P(\bx)
      =
      \sum_{k_1=0}^{\deg_1}\cdots\sum_{k_t=0}^{\deg_t}
      f_{k_1,\ldots,k_t}(\bw_1\cdot\bx,\ldots,\bw_{m-t}\cdot\bx)
      \prod_{j=1}^t
      P_{k_j}^{(d_j)}(\bw_{d_j}\cdot\bx),
   \end{equation}
   where $\deg_1,\ldots,\deg_t$ are the partial degrees of 
   the polynomial $\bx\mapsto P(L^{-1}\bx)$.

   Then there exist $p_1,\ldots,p_t\in\N$ and signs
   $\sigma_1,\ldots,\sigma_t\in\{-1,1\}$ such that the following holds. For every
   $i\in\llbracket1,t\rrbracket$, every $k_i\geq 1$, and every choice of integers
   \begin{equation}
      \tilde n_j\in \sigma_j\N_{\geq p_j},
      \qquad j\in\llbracket1,t\rrbracket\setminus{i},
   \end{equation}
   the following coefficient is an integer:
   \begin{equation}\label{eq:int_coef}
      \sum_{k_1,\ldots,k_{i-1},k_{i+1},\ldots,k_t}
      f_{k_1,\ldots,k_t}(\bw_1\cdot\bn_0,\ldots,\bw_{m-t}\cdot\bn_0)
      \prod_{j\neq i} P_{k_j}^{(d_j)}(\tilde n_j)
      \in\Z .
   \end{equation}
   In other words, the coefficient of $P_{k_i}^{(d_i)}(\bw_{d_i}\cdot\bx)$ 
   takes integer values over an infinite number of integer coordinates.
\end{theorem}

\begin{proof}
   Let $P$ satisfy \cref{eq:P_logical}. Consider the polynomial $Q(\bx):=P(L^{-1}\bx)$ 
   and denote \break $\sigma_1\bg_1,\ldots,\sigma_r\bg_r$ the 
   columns of $GC$ and $\{\be_i\}_{1\leq i\leq m}$ the canonical basis of $\R^m$.
   As the code is positive, te second hypothesis guarantees that 
   $\{\pm\sigma_1\bg_1,\ldots,\pm\sigma_r\bg_r\}$ form a Markov basis of $G$; 
   see \cref{def:Markov_basis}.
   So, according to \cref{prop:periodicity_of_P}, \cref{eq:P_logical} is equivalent to
   \begin{equation}
      \forall i \in \llbracket 1,r \rrbracket,
      \forall \bn \in K_H^{\bDelta}, P(\bn+\bg_i)-P(\bn) \in \Z.
   \end{equation}
   Using $L\sigma_i\bg_i=LGC\be_i=d_i\be_i$ for $i\in\llbracket1,t\rrbracket$ and 
   $L\sigma_i\bg_i=0$ for $i\in\llbracket t+1,r\rrbracket$,
   this is equivalent to demanding $Q$ to satisfy
   \begin{equation}
      \forall i\in\llbracket1,t\rrbracket,
      \forall \bn \in K_H^{\bDelta}, Q(L\bn+\sigma_id_i\be_i)-Q(L\bn) \in \Z.
   \end{equation}
   According to \cref{lemma:infinite_many_constraints}, 
   there exists $\bn_*\in(\bn_0+\ker{H})$ such that 
   \begin{equation}
      L\bn_* + 
      \begin{pmatrix}
         \sigma_1 \N\\
         \vdots \\
         \sigma_t \N\\
         0_{m-t}
      \end{pmatrix}\subset L(K_H^{\bDelta})).
   \end{equation}
   So, in particular, $Q$ satisfies
   \begin{equation}\label{eq:constraint_on_Q}
      \forall i\in\llbracket1,t\rrbracket,\forall \tilde{\bn} \in L\bn_* + 
      \begin{pmatrix}
         \sigma_1 \N\\
         \vdots \\
         \sigma_t \N\\
         0_{m-t}
      \end{pmatrix}, Q(\tilde{\bn}+\sigma_id_i\be_i)-Q(\tilde{\bn}) \in \Z.
   \end{equation}
   According to \cref{lemma:fixed_last_components}, as the code does
   not host any logical rotor, $\ker{L_f}=\ker{H}$. So, 
   \begin{equation}
      \forall\bn\in K_H^{\bDelta}=(\bn_0+\ker{H})\cap\N^m,
      L\bn=\begin{pmatrix}
         L_t\bn \\
         L_f\bn_0
      \end{pmatrix}.
   \end{equation}
   Hence, the last $m-t=f$ entries of $L\bn$ are fixed equals to 
   $L_f\bn_0=\begin{pmatrix}
      \bw_1\cdot\bn_0 & \ldots & \bw_{m-t}\cdot\bn_0
   \end{pmatrix}^\top$.
   In particular, for $\tilde{\bn}$ in \cref{eq:constraint_on_Q}, 
   $\tilde{n}_{t+1},\ldots,\tilde{n}_m=\bw_1\cdot\bn_0, \ldots, \bw_{m-t}\cdot\bn_0$.

   Because $\{P_k^{(d)}(x)\}_{k\in\N}$ is an echeloned family of polynomials,
   it forms a basis of $\R[x]$. 
   So, $Q$ can be decomposed as
   \begin{equation}
      Q(\bx) = \sum_{k_1=0}^{deg_1(Q)}\ldots\sum_{k_t=0}^{deg_t(Q)} 
      f_{k_1,\ldots,k_t}(x_{t+1},\ldots,x_m)P_{k_1}^{(d_1)}(x_1)\ldots P_{k_t}^{(d_t)}(x_t)
   \end{equation}
   where $f_{k_1,\ldots,k_t} \in \R[x_{t+1},\ldots,x_m]$.
   Here, we recover \cref{eq:explicit_P} with $P(\bx)=Q(L^{-1}\bx)$ 
   and where $deg_1,\ldots,deg_t=deg_1(Q),\ldots,deg_t(Q)$.
   For $i\in\llbracket1,t\rrbracket$ and 
   $\tilde{\bn}\in L\bn_* + 
      \begin{pmatrix}
         \sigma_1 \N&
         \ldots &
         \sigma_t \N&
         0_{m-t}
      \end{pmatrix}^\top$,
   consider the univariate polynomial
   \begin{equation}
      Q_i(x_i):=Q(\tilde n_1,\ldots,x_i,\ldots,\tilde n_t,
         \bw_1\cdot\bn_0,\ldots,\bw_f\cdot\bn_0)
         \nonumber\\
      =\sum_{k_i=0}^{\deg_i(Q)}
      A_{k_i}(\tilde \bn)P_{k_i}^{(d_i)}(x_i),
   \end{equation}
   with 
   \begin{equation}
      A_{k_i}(\tilde \bn)
      :=
      \sum_{k_1,\ldots,k_{i-1},k_{i+1},\ldots,k_t}
      f_{k_1,\ldots,k_t}
      (\bw_1\cdot\bn_0,\ldots,\bw_f\cdot\bn_0) \\
      \prod_{j\neq i} P_{k_j}^{(d_j)}(\tilde n_j).
   \end{equation}
   According to \cref{eq:constraint_on_Q}, $Q_i(x_i)$ takes 
   integer values over $\{(L\bn_*)_i+\sigma_ik:k\in\N\}$.
   As 
   \begin{equation}
      \sigma_i\N_{\geq \sigma_i|(L\bn_*)_i|}\subsetneq\{(L\bn_*)_i+\sigma_ik:k\in\N\}
   \end{equation}
   $Q_i(x_i)$ belongs to 
   $\left\{P(x)\in\R[x]: \delta_{d_i}P\left(\sigma_i\N_{\geq \sigma_i|(L\bn_*)_i|}\right)\subset\Z\right\}$.
   So, for $i\in\llbracket1,t\rrbracket$, denoting \break
   $p_i:=\sigma_i|(L\bn_*)_i|$,
   according to \cref{prop:mono_P_characterisation}, for all 
   $k_i\in\llbracket1,deg_i(Q)\rrbracket$ and for all \break
   $\tilde{n}_1,\ldots,\tilde{n}_t\in \sigma_1\N_{\geq p_1},\ldots,\sigma_t\N_{\geq p_t}$,
   $A_{k_i}(\tilde \bn)\in\Z$.
\end{proof}

\Cref{theorem:Zrotations} only proves a necessary condition. 
In practice, it can be proven that it is a 
sufficient condition as well. Using \cref{lemma:integer_valued_polynomial}
and \cref{prop:mono_P_characterisation}, satisfying \cref{eq:int_coef} for 
all $i$ and for infinitely many consecutive integers is sufficient to 
prove that any $P$ satisfying \cref{eq:int_coef} also satisfies 
\cref{eq:P_logical}.

To do so, we would push further the 
characterisation of the $f_{k_1,\ldots,k_t}$'s in the following manner.
For simplicity, assume $i=1$ and consider a fixed $k_1\geq1$. We have
\begin{equation}
   \sum_{k_2,\ldots,k_t} 
   f_{k_1,\ldots,k_t}(\bw_2\cdot\bn_0,\ldots,\bw_{m-t}\cdot\bn_0)
   P_{k_2}^{d_2}(\tilde{n}_2)\ldots P_{k_t}^{d_t}(\tilde{n}_t)\in\Z.
\end{equation}
If we decompose $P_{k_2}^{(d_2)}$ in the $\left\{\binom{x_2}{p}\right\}_{p\in\N}$ basis,
because $\tilde{n}_2,\ldots,\tilde{n}_t$ take value in 
$\sigma_2\N_{\geq p_2},\ldots,\sigma_t\N_{\geq p_t}$, respectively, 
the coefficients of this decomposition should all be integers,
according to \cref{lemma:integer_valued_polynomial}.
We can inductively iterate the process to eventually reach a full 
characterisation of \break
$f_{k_1,\ldots,k_t}(x_{t+1},\ldots,x_m)$. 

We chose not to do so 
as the resulting characterisation turns out to be cumbersome and not enlightening
in the general case. However, in the logical one-qubit case, \cref{theorem:Zrotations} 
leads to an explicit necessary and sufficient characterisation, which we now
record.

\begin{proposition}\label{prop:qubitZrotation}
   Consider a positive Tiger code satysfying the assumptions of 
   \cref{theorem:Zrotations}. Moreover,
   assume that it hosts a single logical qubit. Denoting 
   $\{w,w_1,\ldots,w_{m-1}\}$ the rows of $L$,
   $\kerpi{G^\top}/\impi{H^\top}=\pi\Z_2w$. 
   Then, $e^{2i\pi P(\n)}$ preserves the codespace if and only if 
   \begin{equation}
      P(\bx) = \sum_{k=0}^{deg_1}
      f_k(\bw_1\cdot\bx,\ldots,\bw_{m-1}\cdot\bx)
      P_{k}^{(2)}(\bw\cdot\bx),
   \end{equation}
   where for all $k\geq1, f_k(\bw_1\cdot\bn_0,\ldots,\bw_{m-1}\cdot\bn_0)\in\Z$.
   Moreover, a qubit logical $Z$ rotation of 
   angle $\frac{(-1)^{k+1}}{2^k}2\pi$ can only be produced by the degree-$d$ $P_d^{(2)}$ 
   component in the decomposition of $P$.
\end{proposition}

\begin{proof}
   Since the code encodes a single logical qubit, 
   \cref{theorem:Zrotations} guarantees that a polynomial $P$ 
   satisfying \cref{eq:P_logical} takes the following form
   \begin{equation}
      P(\bx) = \sum_{k=0}^{deg_1}
      f_k(\bw_1\cdot\bx,\ldots,\bw_{m-1}\cdot\bx)
      P_{k}^{(2)}(\bw\cdot\bx)
   \end{equation}
   where, for all $k\geq1$,
   \begin{equation}
      f_k(\bw_1\cdot\bn_0,\ldots,\bw_{m-1}\cdot\bn_0)\in\Z.
   \end{equation}
   We now check that this condition over $f_k$
   is sufficient for $P$ to satisfy \cref{eq:P_logical}. For all 
   $\bn\in K_H^{\bDelta}$ and all $\bg\in\im{G}$ such that $\bn+\bg\in\N^m$,
   because $\bg\in\ker{H}=\ker{L_f}$ (\cref{lemma:fixed_last_components}), 
   $\bw_i\cdot\bg=0$ for all $i\in\llbracket1,m-1\rrbracket$. Moreover, denoting 
   $\bg=G\br_0$ for some $\br_0\in\Z^r$, because $\bw$ is the first row of $L$,
   \begin{equation}
      \bw\cdot\bg=\be_1^\top LG\br_0=\be_1^\top DC^{-1}\br_0
      =2\be_1^\top C^{-1}\br_0=2q
   \end{equation}
   for some integer $q\in\Z$. Hence,
   \begin{equation}
      P(\bn+\bg)-P(\bn)=\sum_{k=0}^{deg_1}f_k(\bw_1\cdot\bn_0,\ldots,\bw_{m-1}\cdot\bn_0)
      \left(P_{k}^{(2)}(\bw\cdot\bn+2q) - P_{k}^{(2)}(\bw\cdot\bn)\right)
   \end{equation}
   for some integer $q\in\Z$. By \cref{prop:mono_P_characterisation}, 
   $\delta_2P_{k}^{(2)}(\Z)\subset\Z$ which implies that, since 
   $\bw\cdot\bn\in\Z$,
   \begin{equation}
      P_{k}^{(2)}(\bw\cdot\bn+2q)-P_{k}^{(2)}(\bw\cdot\bn)\in\Z
   \end{equation}
   Hence, $P(\bn+\bg)-P(\bn)\in\Z$.

   By \cref{prop:mono_P_characterisation}, $P_{k}^{(2)}(x)$ can only take two 
   different values in $\R/\Z$. Hence, it is sufficient to know $P_{k}^{(2)}(0)=0$ and 
   $P_{k}^{(2)}(1)$ to recover $e^{2i\pi P_{k}^{(2)}(n)}$ for all integer $n\in\N$.
   According to \cref{lemma:P1}, 
   \begin{equation}
      \forall k\in\N,P_{k}^{(2)}(1)=\frac{(-1)^{k+1}}{2^{k}}.
   \end{equation}
   Thus, if $deg_1\leq d$, $e^{2i\pi P(\n)}$ can at most implement 
   a Z-rotation of angle $\frac{(-1)^{d+1}}{2^d}2\pi$.
\end{proof}

According to \cite{Cui_2017}, a Z-rotation of angle $\frac{(-1)^{d+1}}{2^d}2\pi$
belongs to the $d$-th level of the Clifford hierarchy.
Therefore, this proposition makes a direct connection between 
the degree of the polynomial of the physical gate and the degree of 
the logical gate induced in 
the Clifford hierarchy \cite{Cui_2017,desilva2025cliffordhierarchyqubitqudit}.
Below, we describe how \cref{prop:qubitZrotation} applies to 
both the cat and paircat codes.

\begin{example}
   Recalling \cref{eq:3firstPa} to the case $a=2$, we get 
   \begin{equation}
      \begin{aligned}
         &P_1^{(2)}(x)=\frac{x}{2},\\
         &P_2^{(2)}(x)=\frac{1}{4}x^2-\frac{1}{2}x,\\
         &P_3^{(2)}(x)=\frac{1}{12}x^3 - \frac{3}{8}x^2 + \frac{5}{12}x.
      \end{aligned}
   \end{equation}
   For the cat code $\cT\left(G\!=\!(2),H\!=\!(0),\alpha,0\right)$, since $G$ is already 
   in Smith normal form, we have $L=(1)$ and the first row is $w=(1)$.
   Therefore, the operators 
   \begin{equation}
      e^{i\pi\hn}\quad;\qquad 
      e^{i(\frac{\pi}{2}\hn^2-\pi\hn)}\quad;\qquad
      e^{i(\frac{\pi}{6}\hn^3 - \frac{3\pi}{4}\hn^2 + \frac{5\pi}{6}\hn)}
   \end{equation}
   induce respectively a logical $Z$ gate, $S$ gate and $T$ gate,
   according to \cref{prop:qubitZrotation}.

   Now, consider the paircat code $\cT\left(G\!=\!\begin{pmatrix}
      2&2
   \end{pmatrix}^\top,H\!=\!\begin{pmatrix}
      1&-1
   \end{pmatrix},\begin{pmatrix}
      \alpha&\alpha
   \end{pmatrix}^\top,\Delta\right)$ and take $\Delta\leq0$ 
   for simplicity. In that case, Smith normal form gives
   \begin{equation}\label{eq:Lpaircat}
      L=\begin{pmatrix}
         1&0\\1&-1
      \end{pmatrix}
   \end{equation}
   such that $LG=\begin{pmatrix}
      2&0
   \end{pmatrix}^\top$. 
   The first row is $\bw=\begin{pmatrix}
      1&0
   \end{pmatrix}$, implying that $\bw\cdot\n=\hn_1$. Hence, 
   as for the cat code,
   \begin{equation}
      e^{i\pi\hn_1}\quad;\qquad 
      e^{i(\frac{\pi}{2}\hn_1^2-\pi\hn_1)}\quad;\qquad
      e^{i(\frac{\pi}{6}\hn_1^3 - \frac{3\pi}{4}\hn_1^2 + \frac{5\pi}{6}\hn_1)}
   \end{equation}
   induce respectively a logical $Z$ gate, $S$ gate and $T$ gate.
   Note that 
   \begin{equation}
      e^{i(\hn_1-\hn_2-\Delta+2)(\frac{\pi}{2}\hn_1^2-\pi\hn_1)}
   \end{equation}
   is still a logical operator where $f_1(\bx)=x_1-x_2-\Delta+2$ by
   \cref{prop:qubitZrotation}.
   However, its logical action is no longer a $S$ gate: 
   for any code state $\ket{\psi}$, 
   since $(\hn_1-\hn_2-\Delta)\ket{\psi}=0$, we get 
   \begin{equation}
      e^{i(\hn_1-\hn_2-\Delta+2)(\frac{\pi}{2}\hn_1^2-\pi\hn_1)}
      \ket{\psi}=e^{2i(\frac{\pi}{2}\hn_1^2-\pi\hn_1)}
      \ket{\psi}.
   \end{equation}
   So, for $\bl=\begin{pmatrix}
         1,1-\Delta
      \end{pmatrix}^\top$,
   \begin{equation}
      \ket{\balpha_{[\bl]}}=\frac{e^{-|\alpha|^2}}{A_{[\bl]}}
      \sum_{n\in\N}\frac{\alpha^{2(n+1)-\Delta}}{\sqrt{(n+1)!(n+1-\Delta)!}}\ket{n+1,n+1-\Delta}
   \end{equation}
   and
   \begin{equation}
      e^{2i(\frac{\pi}{2}\hn_1^2-\pi\hn_1)}\ket{\balpha_{[\bl]}}
      =e^{2i(\frac{\pi}{2}-\pi)}\ket{\balpha_{[\bl]}}
      =e^{i\pi}\ket{\balpha_{[\bl]}}.
   \end{equation}
   Hence, $e^{i(\hn_1-\hn_2-\Delta+2)(\frac{\pi}{2}\hn_1^2-\pi\hn_1)}$
   implements a logical $Z$ gate.
\end{example}

\begin{remark}
   Remark that the second row of $L$ in \cref{eq:Lpaircat} is exactly $H$.
   This fact is not surprinsing considering that 
   \cref{lemma:fixed_last_components} guarantees that $L_f$, which 
   in this case is simply the second row of $L$, has the same kernel 
   as $H$. Here, $L_f=H$ but notice that if we would have chosen 
   $H=\begin{pmatrix}
      2&-2
   \end{pmatrix}$, 
   the resulting code would have still been a paircat code
   with the same matrix $L$, but $H=2L_f$. 
\end{remark}

We do not expect to generalise such a characterisation to all Tiger codes.
In fact, \cref{theorem:Zrotations} fails for certain negative Tiger codes 
as shown in the following example. This leaves room for implementations 
of non-Clifford gates with low degree polynomials.

\begin{example}\label{ex:binomial_code3}
   Consider the two-mode binomial code introduced in \cref{ex:binomial_code}.
   It is a non-positive Tiger code.
   In particular, take $\Delta=1$. In this case, the two logical 
   states of the qubit are 
   \begin{equation}
      \ket{\overline{0}}=\ket{0,1}
      \quad \text{and}\quad 
      \ket{\overline{1}}=\ket{1,0}.
   \end{equation}
   For any $\theta\in[0,2\pi)$, the operator $e^{i\theta \hn_1}$ 
   implements the qubit gate $\begin{pmatrix}
      1&0\\ 0&e^{i\theta}
   \end{pmatrix}$. Hence, any Z-rotation can be implemented with
   $P(x_1,x_2)=\theta x_1$, a degree one polynomial. 
\end{example}

\section{Conclusion}

We have provided a rigorous algebraic treatment of Tiger codes.
Starting from the kernel definition of the codespace, we proved that the
\(\a\)-type constraints are finitely generated, constructed an explicit
orthonormal basis indexed by $K_H^{\bDelta}/\!\!\sim_G$, and showed that the
homology group $\ker H/\im G$ is the natural algebraic object governing the
logical structure of the code.
We then developed a Fourier-theoretic description of Tiger codes, together with
their dual X- and Z-type bases, and showed that the same strategy
extends to non-linear number constraints. This enlarges the
class of bosonic codes captured by the Tiger code framework.

Finally, we addressed the question of logical operators. We first
provided physical implementations of approximate $X$ and $Z$ logical Pauli operators
and then, for finite-dimensional positive Tiger codes satisfying an additional
sign-compatibility assumption on the Smith decomposition of $G$, we 
derived a structural necessary condition
for polynomial phase gates of the form $e^{2i\pi P(\n)}$ to preserve the 
codespace. In the particular case of a one-qubit codespace, 
the operator $e^{2i\pi P(\n)}$ preserves the codespace if and 
only if $P$ admits a certain decomposition. This decomposition 
reveals how the degrees of its components constrain the level 
of the induced logical gate in the Clifford hierarchy.

Several natural questions remain open.
A first one concerns the special role played by coherent states. The present
work shows that Tiger codes are remarkably well suited to a coherent-state
description: the coherent states belonging to the code generate the
codespace. It would be interesting to understand whether this phenomenon is
specific to Tiger codes, or whether it extends to a broader class of
multivariate polynomials. More precisely, one may ask whether there exists a
class of multivariate polynomials $\cS$ such that, for every family
$P_1,\ldots,P_r\in\cS$,
\begin{equation}
   \bigcap_{1\leq i\leq r}\ker{P_i(\a)}=
   \overline{\Span{\ket{\balpha}:\balpha\in\C^m \text{ s.t. }\forall i\in\llbracket1,r\rrbracket,\ P_i(\balpha)=0}}.
\end{equation}

Another direction concerns finite generation. While our results show that the
$\a$-type constraints can always be reduced to finitely many generators, the
choice of such a generating family is far from unique. From the perspective of
practical implementations, it would be especially valuable to identify Markov
bases that minimise the degree of the monomials appearing in the corresponding
$\a$-type constraints. Better control of this degree could make the
associated dissipative or Hamiltonian implementations significantly more
tractable.

A third direction concerns polynomial phase gates. Our analysis leaves open
the possibility of implementing non-Clifford gates through low-degree
polynomials $P$ in operators of the form $e^{2i\pi P(\n)}$ on suitable
Tiger codes. Understanding the trade-off between the error-correcting
properties of a code and the minimal degree required to implement such
non-Clifford gates is left entirely open.
Likewise, identifying codes for which $e^{2i\pi P(\n)}$ can be implemented
mode by mode, without the need to engineer interaction Hamiltonians, would
significantly improve the practical relevance of these constructions.

Finally, as already mentioned in the introduction, we did not address in this
work the actual error-protection performance of Tiger codes. In particular,
finding dissipative dynamics that stabilise their codespaces, and analysing the
corresponding convergence rates would provide a more refined assessment of
their ability to protect quantum information against noise.

\section*{Acknowledgements}
We acknowledge the Plan France 2030 through the project ANR-22-PETQ-0006.
C.P. thanks Rémi Robin for his precious help on understanding 
how to implement Tiger codes using dissipative dynamics and Samo Novák
for fruitfull discussions about the Pontryagin duality.

\clearpage
\printbibliography
\clearpage

\appendix

\section{Image and kernel of integer matrices}\label{app:image_and_kernel}

\begin{proposition}
   For $A\in\Z^{n\times m}$, there exists $B\in \Z^{m\times p}$ such that $\ker{A}=\im{B}$.
\end{proposition}

\begin{proof}
   Write the smith normal form of $A$,
   \begin{equation}
      LAC=\diag{n\times m}{d_1,\ldots,d_t}=D.
   \end{equation}
   Then, as $L$ is invertible, the $m-t$ last columns of $C$ form a basis of $\ker{A}$.
   So, taking $B$ to be the last columns of $C$, we get the result.
\end{proof}

The other implication is false. If you consider $A=(2)$, then $\im{A}=2\Z$ 
which cannot be written as the kernel of a matrix $B$.

The usual definition of a Markov basis involves $\ker{A}$
instead of $\im{A}$ in \cref{def:Markov_basis}. According to the previous proposition,
a Markov basis for $\ker{A}$ is also a Markov basis for $\im{B}$ for a well suited
$B\in \Z^{m\times p}$. However, the other implication is false, making 
\cref{def:Markov_basis} more general than the usual definition of Markov basis 
\cite{MarkovBasesDiacoSturm}.

\section{Combinatorics}

\begin{lemma}\label{lemma:P1}
   $\forall k\in\N^*,P_{k}^{(2)}(1)=\frac{(-1)^{k+1}}{2^{k}}$.
\end{lemma}

\begin{proof}
   We shall prove that for all $k\in\N$,
   \begin{equation}
      P_{k+1}^{(2)}(1)
      =\sum_{i=0}^k \sum_{j=0}^i (-1)^{i-j}\binom{i}{j}\binom{2j}{k}\binom{1/2}{i+1}
      =\frac{(-1)^{k}}{2^{k+1}}.
   \end{equation}
   First, we simplify $S_{k,i}=\sum_{j=0}^i (-1)^{i-j}\binom{i}{j}\binom{2j}{k}$. To do so,
   we use generating functions:
   \begin{equation}
      \begin{aligned}
         \sum_{k\in\N}S_{k,i}x^k
         &=\sum_{j=0}^i (-1)^{i-j}\binom{i}{j}\sum_{k\in\N}\binom{2j}{k}x^k\\
         &=\sum_{j=0}^i \binom{i}{j}(-1)^{i-j}(x+1)^{2j}\\
         &=\left((x+1)^2-1\right)^i\\
         &=x^i(x+2)^i\\
         &=\sum_{k=0}^i\binom{i}{k}2^{i-k} x^{i+k}\\
         &=\sum_{k=0}^i\binom{i}{k-i}2^{2i-k} x^{k}.
      \end{aligned}
   \end{equation}
   Hence, $S_{k,i}=\binom{i}{k-i}2^{2i-k}$. So, again using generating function,
   \begin{equation}
      \begin{aligned}
         \sum_{k\in\N}P_{k+1}^{(2)}(1)x^k
         &=\sum_{k\in\N}\sum_{i=0}^{k}\binom{i}{k-i}2^{2i-k}\binom{1/2}{i+1}x^k\\
         &=\sum_{i\in\N}\binom{1/2}{i+1}\sum_{k\geq i}\binom{i}{k-i}2^{2i-k}x^k\\
         &=\sum_{i\in\N}\binom{1/2}{i+1}\sum_{k\geq 0}\binom{i}{k}2^{i-k}x^{k+i}\\
         &=\sum_{i\in\N}\binom{1/2}{i+1}x^i(x+2)^i\\
         &=\frac{\left(x(x+2)+1\right)^{1/2}-1}{x(x+2)}\\
         &=\frac{1}{x+2}
      \end{aligned}
   \end{equation}
   which gives for $x$ sufficiently small 
   \begin{equation}
      \sum_{k\in\N}P_{k+1}^{(2)}(1)x^k
      =\frac{1}{2}\frac{1}{1+\frac{x}{2}}
      =\sum_{k\in\N}\frac{(-1)^k}{2^{k+1}}x^k.
   \end{equation}
   Hence, $\forall k\in\N,P_{k+1}^{(2)}(1)=\frac{(-1)^k}{2^{k+1}}$.
\end{proof}

\end{document}